\documentclass[journal, final]{IEEEtran}
\IEEEoverridecommandlockouts
\newif\ifisdraft
\isdrafttrue
\usepackage[utf8]{inputenc}
\usepackage[T1]{fontenc}
\usepackage{microtype}

\usepackage{amsmath,amssymb,amsfonts}
\usepackage{mathtools}
\usepackage{bm}
\usepackage{amsthm}

\usepackage{booktabs}
\usepackage{array}
\usepackage{tabularx}
\usepackage{makecell}
\usepackage{enumitem}
\usepackage{cuted}
\usepackage{multicol}
\usepackage{comment}

\usepackage{algorithm}
\usepackage{algpseudocode}

\usepackage{graphicx}
\usepackage{xcolor}
\usepackage{tikz}
\usetikzlibrary{arrows.meta,calc,fit,positioning,shapes.misc}

\definecolor{headercolor}{gray}{0.82}
\definecolor{zebracolor}{gray}{0.95}
\definecolor{lightgray}{gray}{0.93}
\definecolor{darkred}{rgb}{0.6,0,0}

\usepackage{cite}
\usepackage{doi}
\usepackage{bibunits}
\defaultbibliographystyle{IEEEtran}
\defaultbibliography{references}

\usepackage{hyperref}

\ifisdraft
  \hypersetup{
    colorlinks=true,
    linkcolor=red!50!black,
    citecolor=blue!50!black,
    urlcolor=black
  }
\else
  \hypersetup{hidelinks}
\fi

\usepackage{cleveref}

\newtheorem{theorem}{Theorem}[section]

\newtheorem{definition}[theorem]{Definition}
\newtheorem{remark}[theorem]{Remark}
\newtheorem{proposition}[theorem]{Proposition}

\newtheorem{corollary}[theorem]{Corollary}

\providecommand{\Ftwo}{\mathbb{F}_2}
\providecommand{\calC}{\mathcal{C}}

\providecommand{\xor}{\oplus}
\providecommand{\one}{\mathbf{1}}

\DeclareMathOperator*{\argmin}{arg\,min}

\newcommand{\orcid}[1]{%
  \href{https://orcid.org/#1}{\includegraphics[width=10pt]{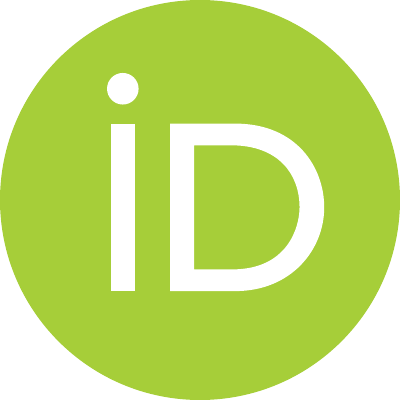}}%
}

\makeatletter

\let\lpgrand@original@section\section
\let\lpgrand@original@subsection\subsection

\NewDocumentCommand{\lpgrand@section}{s o m}{%
  \IfBooleanTF{#1}{%
    \lpgrand@original@section*{#3}%
  }{%
    \IfNoValueTF{#2}{%
      \lpgrand@original@section{#3}%
      \addcontentsline{atoc}{section}{%
        \protect\numberline{\thesection}#3}%
    }{%
      \lpgrand@original@section[#2]{#3}%
      \addcontentsline{atoc}{section}{%
        \protect\numberline{\thesection}#2}%
    }%
  }%
}

\NewDocumentCommand{\lpgrand@subsection}{s o m}{%
  \IfBooleanTF{#1}{%
    \lpgrand@original@subsection*{#3}%
  }{%
    \IfNoValueTF{#2}{%
      \lpgrand@original@subsection{#3}%
      \addcontentsline{atoc}{subsection}{%
        \protect\numberline{\thesubsection}#3}%
    }{%
      \lpgrand@original@subsection[#2]{#3}%
      \addcontentsline{atoc}{subsection}{%
        \protect\numberline{\thesubsection}#2}%
    }%
  }%
}

\newcommand{\AppendixOnlyTOC}{%
  \lpgrand@original@section*{Appendix Contents}%
  \begingroup
    \footnotesize
    \setcounter{tocdepth}{2}%
    \setlength{\parskip}{0pt}%
    \@starttoc{atoc}%
  \endgroup
  \let\section\lpgrand@section
  \let\subsection\lpgrand@subsection
}

\makeatother

\usepackage{xr-hyper}
\makeatletter
\let\savedbibcite\bibcite
\def\bibcite#1#2{}
\let\bibcite\savedbibcite
\makeatother

\begin{document}

\begin{bibunit}

\title{Low-Pathwidth GRAND: Exact Likelihood-Ordered Enumeration for BPSK Transmission over Correlated Gaussian Noise}

\author{Behrooz~Razeghi~$\!{\orcid{0000-0001-9568-4166}}$,~\IEEEmembership{Senior~Member,~IEEE}%
\thanks{B.~Razeghi is with Harvard University, Cambridge, MA, USA (e-mail: behroozrazeghi@seas.harvard.edu).}%
\thanks{This work was supported by the Swiss National Science Foundation (SNSF) under Grant No.~222339.}%
\thanks{\urlstyle{sf}Code: \url{https://github.com/BehroozRazeghi/lp-grand}.}%
}

\maketitle

% less than 200 words version
%
\begin{abstract}
The finite-block maximum-likelihood (ML) guarantee of soft-input GRAND requires querying noise-effect patterns in nonincreasing conditional-likelihood order. Under correlated Gaussian noise, additive reliability metrics and independent-block approximations need not preserve this order because the matched metric contains cross-coordinate interactions; the first codebook hit need not induce an ML codeword. We develop Low-Pathwidth GRAND (LP-GRAND) for binary phase-shift keying (BPSK) with precision matrix $Q$. The candidate-dependent part of the Gaussian negative log-likelihood is an observation-dependent quadratic pseudo-Boolean energy whose interaction graph has edge $\{i,j\}$ exactly when $Q_{ij}\neq0$. If $Q$ has half-bandwidth at most $\nu$, this energy admits a trellis with at most $2^\nu$ states per layer; a path decomposition of width $w$ yields at most $2^{w+1}$ bag assignments per layer. In real arithmetic, suffix dynamic programming and best-first complete-path enumeration enumerate patterns in nondecreasing energy. With complete enumeration and no abandonment, the first codebook hit induces an ML codeword for any nonempty binary codebook with equiprobable codewords. LP-GRAND agreed with exhaustive codeword ML in all $10{,}000$ frames for two $[20,12]$ codes. At nominal $E_b/N_0=2$ dB, its empirical BLER was lower than that of each block-based approximation for six $[64,52]$ codes.
\end{abstract}


\section{Introduction}

Soft decoders built around a code-specific representation---such as a factor graph, trellis, polarization transform, or concatenated construction---typically couple the decoding procedure to a particular code family~\cite{gallager1962low,kschischang2001factor,berrou1993near, arikan2009channel,tal2015list}. Consequently, a receiver that supports multiple short codebooks, cyclic-redundancy-check (CRC)-constrained codebooks, random linear codes, or codebooks without dedicated soft decoders may require distinct decoder implementations. Guessing Random Additive Noise Decoding (GRAND) instead separates the generation of candidate noise-effect patterns from the codebook-dependent membership test~\cite{duffy2019capacity,duffy2022ordered}. Rather than enumerating codewords directly, a soft-input GRAND decoder orders candidate noise-effect patterns according to the received vector and the channel model and tests whether the candidate word induced by each pattern belongs to the codebook. In the binary hard-decision representation considered here, a noise-effect pattern is represented by a vector $z\in\{0,1\}^n$ and induces the candidate word $y\xor z$, where $y\in\{0,1\}^n$ is the hard-decision word. Codebook membership is therefore determined by whether $y\xor z\in\calC$.

The code-agnostic membership interface of GRAND is not sufficient by itself to guarantee finite-block maximum-likelihood (ML) decoding; the guarantee also depends on the candidate query order~\cite{duffy2019capacity}. For equiprobable codewords and a fixed received vector, querying noise-effect patterns $z$ in nonincreasing order of the likelihoods of their induced words $y\xor z$, with ties resolved by a fixed deterministic rule, ensures that the first queried pattern satisfying $y\xor z\in\calC$ induces an ML codeword, provided that the search is not terminated before a codeword is found. Such premature termination is referred to as abandonment in the GRAND literature. The number of membership queries is the position of this first codebook hit in the ordered candidate list and therefore depends on both the likelihood order induced by the observed received vector and channel model and the particular codebook. For a binary linear code with parity-check matrix $H\in\Ftwo^{(n-k)\times n}$, codebook membership is equivalent to the zero-syndrome condition $H(y\xor z)^{\mathsf T}=\mathbf{0}_{n-k}$, where the matrix-vector product is evaluated over $\Ftwo$. Thus, once a codebook-membership test is available, the objective of soft-input GRAND is to generate candidate noise-effect patterns in an order consistent with the conditional likelihoods of their induced words for the observed received vector.

For a memoryless channel, the conditional likelihood of a candidate word factorizes across coordinates, allowing soft-input GRAND queries to be generated from coordinatewise reliability information~\cite{solomon2020soft}. Several practical query generators replace the exact conditional-likelihood order induced by the observed received vector and the channel model with a coordinate-additive reliability surrogate or with an approximation that assumes independence across nonoverlapping blocks~\cite{duffy2021guessing,duffy2022ordered,duffy2023using}. For channels with memory, however, the matched conditional likelihood generally does not factorize across coordinates. Correlated receiver noise, residual equalization error, and temporally localized interference can induce statistical dependence that appears as cross-coordinate interaction terms in the matched negative log-likelihood metric over binary hard-decision noise-effect patterns. Interleaving can separate symbols that are locally coupled in the physical symbol sequence, but it requires additional buffering and may place statistically coupled coordinates far apart in the coordinate ordering used by the decoder~\cite{an2022keep,duffy2023using}. Consequently, for low-latency receivers that preserve the physical symbol order, coordinate-additive query metrics and block-independence approximations can omit interaction terms present in the matched negative log-likelihood and therefore need not preserve the candidate order induced by the matched channel likelihood.

ORBGRAND provides a computationally structured soft-input query generator by replacing the exact conditional-likelihood order induced by the received observation with an integer-valued surrogate constructed from reliability ranks~\cite{duffy2022ordered}. Specifically, the bit positions are ordered by nondecreasing reliability magnitude, with rank one assigned to the least reliable position. Let $\tilde z_i\in\{0,1\}$ indicate whether the coordinate assigned reliability rank $i$ is included in the candidate noise-effect pattern. Basic ORBGRAND assigns $\tilde z$ the logistic weight $w_L(\tilde z)= \sum_{i=1}^n i \, \tilde z_i$ and generates patterns in nondecreasing order of this weight. Piecewise-linear ORBGRAND variants replace the linear rank-weight model $i\mapsto i$ with an integer-valued piecewise-linear approximation to the ordered reliability-magnitude profile~\cite{duffy2022ordered, riaz2024sub}. These additive rank-domain metrics permit query generation through structured integer-partition methods and support parallelizable pattern generation.

The additive rank-domain metric used by basic ORBGRAND permits computationally structured pattern enumeration but cannot represent the cross-coordinate interaction terms present in the negative log-likelihood metric matched to a channel with memory. Its query order therefore need not coincide with the candidate order induced by the matched likelihood. ORBGRAND-AI models statistical dependence within nonoverlapping blocks by computing likelihoods for alternative assignments within each block and then applying an ORBGRAND-style query generator under an independence approximation across blocks~\cite{duffy2023using}. Thus, it uses joint likelihood information within each block but omits statistical dependence between different blocks. Increasing the block length $b$ allows the local likelihood calculation to include a larger group of jointly modeled coordinates. However, a BPSK block of length $b$ has $2^b-1$ nonbaseline assignments relative to its selected baseline assignment. Consequently, the number of local substitutions grows exponentially with $b$, whereas cross-block dependence remains absent from the approximating global query-ordering metric.

An exact treatment of channel memory has also been developed for finite-block linear Gaussian intersymbol-interference (ISI) channels through SGRAND-ISI, which constructs an ML-equivalent GRAND ordering under its specified channel model~\cite{li2026grand_isi}. In the overlapping ISI setting, SGRAND-ISI and the precision-domain formulation considered here induce the same candidate-word order whenever they represent the same finite-block Gaussian likelihood. The present work instead addresses the more general problem of constructing this matched order directly from a supplied sparse precision matrix through a validated bounded-width path decomposition.

The proposed formulation therefore avoids the across-block independence approximation by using the exact negative log-likelihood metric induced by the specified Gaussian model whenever the corresponding precision-induced interaction graph admits a path decomposition of bounded width. For BPSK transmission over correlated Gaussian noise, let $Q=\Sigma^{-1}$ denote the precision matrix. For a fixed observation $r$, the candidate-dependent part of the Gaussian negative log-likelihood associated with the candidate word $y\xor z$, where $z\in\{0,1\}^n$, can be represented exactly as a quadratic pseudo-Boolean energy, up to an additive constant independent of $z$. The interaction graph of this energy contains the edge $\{i,j\}$, for $1\leq i<j\leq n$, if and only if $Q_{ij}\neq0$. For stationary first-order Gauss--Markov noise, $Q$ is tridiagonal, and the resulting energy has the chain form $W_r(z)=\sum_{i=1}^n\alpha_i z_i+ \sum_{i=1}^{n-1}\beta_i z_i z_{i+1}$. Under the coordinate order $1,\ldots,n$, binary noise-effect patterns are in one-to-one correspondence with complete paths through a two-state trellis, and the total cost of the path associated with $z$ equals $W_r(z)$. Consequently, enumerating distinct complete paths in nondecreasing total cost, with equal-cost paths ordered by a fixed deterministic rule, produces the GRAND query order induced by the matched Gaussian likelihood. In particular, the first $K$ emitted paths induce the first $K$ candidate queries in this order.

We refer to the proposed decoder as Low-Pathwidth GRAND (LP-GRAND). Unlike additive reliability-rank surrogates and block-independence approximations, LP-GRAND orders candidate noise-effect patterns according to a quadratic energy that is equal, up to a candidate-independent constant, to the Gaussian negative log-likelihood determined by the observed received vector and the specified precision matrix. Given a validated path decomposition of width $w$ for the corresponding precision-induced interaction graph, the resulting layered graphical representation retains all unary and pairwise terms of this quadratic energy and supports enumeration of complete paths in nondecreasing energy. The graphical metric is independent of the codebook; the codebook is accessed only through the membership test. Consequently, in real arithmetic, LP-GRAND preserves GRAND's code-independent candidate-generation interface while producing the candidate order induced by the matched Gaussian likelihood.

The main algorithmic challenge is that standard sum-product and min-sum recursions do not produce the ordered sequence of binary noise-effect patterns required by GRAND~\cite{kschischang2001factor}. On an acyclic graphical model, sum-product computes the partition function and exact marginal distributions, whereas min-sum identifies a minimum-energy assignment. Neither recursion directly generates a ranked sequence of complete assignments. LP-GRAND addresses this challenge by combining suffix dynamic programming with best-first enumeration of complete paths in a layered directed acyclic graph. For a precision matrix with half-bandwidth at most $\nu$, the resulting trellis has at most $2^\nu$ states per layer and $\mathcal{O}(n2^\nu)$ total size. Thus, for fixed $\nu$, the graphical representation grows linearly with the blocklength $n$.

\vspace{-3pt}

\subsection{Contributions}

The principal contributions of this work are as follows:
\begin{enumerate}[leftmargin=*]
\item For a fixed received observation $r$ and a specified symmetric positive-definite precision matrix $Q$, we derive an observation-dependent quadratic pseudo-Boolean energy $W_r(z)$ satisfying $E_r(z)=E_r(0)+W_r(z)$, where $E_r(z)$ is the candidate-dependent quadratic term in the Gaussian negative log-likelihood associated with the candidate word
induced by $z$. We prove that, in real arithmetic, enumerating all $z\in \{0,1\}^n$ in nondecreasing order of $W_r(z)$, with ties resolved by a fixed deterministic rule, ensures that the first codebook hit is an ML codeword for any nonempty binary codebook with equiprobable codewords. Decoding under a finite membership-query budget is treated separately.
\item For a precision matrix with half-bandwidth at most $\nu$, we construct a trellis with at most $2^\nu$ states per layer and combine suffix dynamic programming with best-first complete-path enumeration to generate the exact candidate order induced by the matched Gaussian likelihood. For an interaction graph induced by a general sparse precision matrix, we construct and validate a path decomposition from a specified vertex ordering and build the associated layered graphical representation and complete-path enumerator. A path decomposition of width $w$ yields a layered representation with at most $2^{w+1}$ states per layer, one for each binary assignment to the bag associated with that layer. The enumeration is exact with respect to the specified quadratic energy for any validated path decomposition, whereas the size of the layered representation depends on the width of the decomposition produced from the selected vertex ordering. The vertex-ordering heuristics are not guaranteed to produce a minimum-width path decomposition. The decomposition width bounds the number of states per layer but does not bound the number of membership queries required to reach the first codebook hit.
\item We establish conditions for exact candidate ordering and sufficient conditions for preserving strict pairwise order under coefficient quantization and precision-matrix mismatch. Enumeration using the complete quadratic energy retains all unary and pairwise terms induced by the specified precision matrix. By contrast, explicit sparsification yields an order that is exact for the metric induced by the approximating precision matrix $\widetilde Q$, but not necessarily for the metric induced by the original precision matrix $Q$. We derive both uniform and pair-specific sufficient conditions under which the strict order of two candidates is preserved under these perturbations. The same graphical representation also supports exact computation of the partition function over the full binary assignment space $\{0,1\}^n$.
\end{enumerate}

\vspace{3pt}

See Appendix~\ref{app:sec:related_work} for extended related work.

% \subsection{Organization}

\vspace{-1pt}

\section{Preliminaries and Problem Formulation}
\label{sec:model}

Throughout, binary vectors are written as row vectors, whereas real-valued signal, noise, and received-signal vectors are written as column vectors. Addition of binary vectors is performed componentwise over $\mathbb{F}_2$ and is denoted by $\xor$. Let $\calC \subseteq \{0,1\}^n$ be a nonempty binary codebook of cardinality $|\calC|=M$. If $\calC$ is a binary linear $[n,k]$ code, then $M=2^k$. For each $c\in \{0,1\}^n$, define the BPSK modulation mapping $x(c) \triangleq \bigl(1-2c_1,\ldots,1-2c_n\bigr)^{\mathsf T} \in \{-1,+1\}^n$.

Let $C$ be uniformly distributed over $\calC$. The transmitted BPSK signal is $x(C)$, and the received vector is $R=x(C)+N$, $N\sim\mathcal{N}\bigl(\mathbf{0}_n,\Sigma\bigr)$, where $\Sigma\in\mathbb{R}^{n\times n}$ is symmetric positive definite. Let $Q\triangleq\Sigma^{-1}$ denote the corresponding precision matrix, and let $r\in\mathbb{R}^n$ be a realization of $R$. Let $y$ denote the coordinatewise sign-decision reference word, defined by $y_i\triangleq 0$ if $r_i\geq 0$ and $y_i\triangleq 1$ otherwise, for $i=1,\ldots,n$. Thus, ties at $r_i=0$ are resolved in favor of $y_i=0$. Finally, let $s\triangleq x(y)\in\{-1,+1\}^n$ denote the BPSK-modulated hard-decision vector.
For a binary noise-effect pattern $z\in\{0,1\}^n$, define the induced candidate word as $\widehat{c}(z)\triangleq y\xor z$ and its BPSK-modulated representation as
\begin{equation}
x_z \triangleq x\bigl(\widehat{c}(z)\bigr) = s\odot\bigl(\mathbf{1}_n-2z^{\mathsf T}\bigr),
\label{eq:xz_def}
\end{equation}
where $\odot$ denotes the Hadamard product and $\mathbf{1}_n\in\mathbb{R}^n$ is the length-$n$ all-ones column vector. Because binary vectors are represented as row vectors, $z^{\mathsf T}$ is viewed as a column vector in $\mathbb{R}^n$ in~\eqref{eq:xz_def}. For fixed $y$, the mapping $z\mapsto y\xor z$ is a bijection on $\{0,1\}^n$.

\begin{remark}
For fixed $r$ and any candidate word $c\in\{0,1\}^n$, the vector $z=y(r)\xor c$ uniquely identifies the coordinates at which $c$ differs from the hard-decision word $y(r)$. Hence, $z$ provides an observation-dependent parameterization of the candidate words and is not assumed to be a realization of a codeword-independent binary additive-noise process. We use the term ``noise effect'' for the candidate pattern indexed by $z$ in the GRAND query sequence. The query order over $z$ is induced by the conditional Gaussian likelihood $f_{R|C}\bigl(r\mid y(r)\xor z\bigr)$.
\end{remark}

For a fixed observation $r$, the Gaussian negative log-likelihood of the candidate word $\widehat{c}(z)$ is
\begin{equation}
\!\!
-\log f_{R|C}\bigl(r\mid\widehat{c}(z)\bigr) = E_r(z)
+\frac{n}{2}\log(2\pi) +\frac{1}{2}\log\det(\Sigma),
\end{equation}
where 
\begin{equation}
E_r(z) \triangleq \frac{1}{2}(r-x_z)^{\mathsf T}Q(r-x_z)
\label{eq:energy}
\end{equation}
is the candidate-dependent quadratic term in the Gaussian negative log-likelihood associated with the candidate word induced by $z$. The final two terms are independent of $z$. Because $C$ is uniformly distributed over $\calC$, the sets of ML and maximum a posteriori (MAP) codewords coincide, and their common decision set is
\begin{equation}
\argmin_{c\in\calC} \frac{1}{2}(r-x(c))^{\mathsf T}Q \,(r-x(c)).
\label{eq:ml_codeword}
\end{equation}

Because the mapping $z\mapsto y\xor z$ is bijective, ordering the noise-effect patterns $z\in\{0,1\}^n$ in nondecreasing order of $E_r(z)$ is equivalent to ordering their induced candidate words in nonincreasing order of the conditional likelihood $f_{R|C}\bigl(r\mid\widehat{c}(z)\bigr)$. Consequently, a GRAND decoder that follows this query order, resolves equal-energy candidates using a fixed deterministic tie-breaking rule, and continues until a codebook hit occurs ensures that the first pattern satisfying $y\xor z\in\calC$ induces an ML codeword~\cite{duffy2019capacity}. If multiple codewords attain the minimum in~\eqref{eq:ml_codeword}, the tie-breaking rule selects one of them.

% ======================================================================
\section{Gaussian Noise-Effect Metric}
\label{sec:metric}
% ======================================================================

Let $D_s\triangleq\operatorname{diag}(s_1,\ldots,s_n)$ and $a\triangleq r-s$. Here, $r$, $s$, $a$, and $x_z$ are real-valued column vectors, whereas $z\in\{0,1\}^n$ is a binary row vector and $z^{\mathsf T}$ denotes its real-valued column representation. From \eqref{eq:xz_def}, $x_z = s\odot\bigl(\mathbf 1_n- 2z^{\mathsf T}\bigr) = s-2D_s z^{\mathsf T}$. Hence,
\begin{equation}
r-x_z = a+2D_s z^{\mathsf T}.
\label{eq:residual_linear}
\end{equation}
Substituting \eqref{eq:residual_linear} into \eqref{eq:energy} and using the symmetry of $Q$ gives
\begin{align}
E_r(z)
&=
\frac{1}{2} \bigl(a+2D_s z^{\mathsf T}\bigr)^{\mathsf T} Q \bigl(a+2D_s z^{\mathsf T}\bigr) \nonumber\\
&=
E_r(0) +2zD_sQa +2zD_sQD_s z^{\mathsf T},
\label{eq:energy_expand}
\end{align}
where $E_r(0) = \frac{1}{2}a^{\mathsf T}Qa$. 
% \begin{equation}
% E_r(0) = \frac{1}{2}a^{\mathsf T}Qa.
% \label{eq:zero_flip_energy}
% \end{equation}

% \subsection{Quadratic Pseudo-Boolean Form}

\begin{proposition}[Quadratic pseudo-Boolean representation]
\label{prop:quadratic_metric}
For fixed observation $r$, let $s=x(y)$, $a=r-s$, and $Q=\Sigma^{-1}$. Then
\begin{equation}
E_r(z)=E_r(0)+W_r(z),
\label{eq:energy_decomposition}
\end{equation}
where $W_r:\{0,1\}^n\to\mathbb R$ is the quadratic pseudo-Boolean function
\begin{equation}
W_r(z) = \sum_{i=1}^n \alpha_i z_i + \sum_{1\leq i<j\leq n}\beta_{ij}z_i z_j,
\label{eq:quadratic_metric}
\end{equation}
with coefficients
\begin{equation}
\alpha_i = 2s_i(Qa)_i+2Q_{ii},
\qquad i=1,\ldots,n,
\label{eq:alpha}
\end{equation}
and
\begin{equation}
\beta_{ij} = 4s_i s_jQ_{ij},
\qquad 1\leq i<j\leq n.
\label{eq:beta}
\end{equation}
\end{proposition}

\begin{proof}
See Appendix~\ref{app:proof_prop_quadratic_metric}.
\end{proof}

\begin{remark}[Precision-induced interaction graph]
\label{rem:precision_graph}
Define the undirected graph $\mathcal{G}_Q=(\mathcal{V},\mathcal{E}_Q)$ with $\mathcal{V}=\{1,\ldots,n\}$ and edge set
\begin{equation}
\mathcal{E}_Q = \bigl\{ \{i,j\}: 1\leq i<j\leq n,\; Q_{ij}\neq0 \bigr\}.
\end{equation}
Because $s_i s_j\in\{-1,+1\}$, \eqref{eq:beta} implies $\beta_{ij}\neq0$ if and only if  $Q_{ij}\neq0$.
Hence, $\mathcal{G}_Q$ is the interaction graph of $W_r$ and is determined solely by the off-diagonal sparsity pattern of $Q$, rather than that of $\Sigma$. In particular, its edge set is independent of the observation $r$. The unary coefficients depend on $r$ through both $a$ and $s$. For each $\{i,j\}\in \mathcal{E}_Q$, the magnitude $|\beta_{ij}|=4|Q_{ij}|$ is independent of $r$, whereas its sign depends on $r$ through $s_i s_j$.
\end{remark}

Further details on the stationary first-order Gauss--Markov model, precision-matrix modifications, whitening, and the equivalent finite-block linear-Gaussian formulation are given in Appendix~\ref{ssec:gm_supplement_start}.

\section{LP-GRAND Decoding}
\label{sec:decoder}
% =======================================

Let $z^{(1)},z^{(2)},\ldots,z^{(2^n)}$ be a permutation of $\{0,1\}^n$ such that
\begin{equation}
W_r\bigl(z^{(1)}\bigr) \leq
W_r\bigl(z^{(2)}\bigr) \leq
\cdots \leq W_r\bigl(z^{(2^n)}\bigr),
\label{eq:order}
\end{equation}
with equal-energy patterns ordered by the deterministic tie-breaking policy of the enumerator. Since $E_r(z)=E_r(0)+W_r(z)$ and the Gaussian negative log-likelihood differs from $E_r(z)$ only by a candidate-independent constant, the ordering in~\eqref{eq:order} arranges the induced candidate words $y\xor z^{(1)},\ldots,y\xor z^{(2^n)}$ in nondecreasing negative log-likelihood, or equivalently, in nonincreasing conditional likelihood evaluated at $r$.
For an arbitrary binary codebook $\calC$, define the membership test
\begin{equation}
\Phi:\{0,1\}^n\to\{0,1\}, \qquad
\Phi(c) \triangleq \one\{c\in\calC\}.
\label{eq:membership_test}
\end{equation}
For a binary linear $[n,k]$ code with parity-check matrix $H_{\mathrm{pc}}\in\Ftwo^{(n-k)\times n}$, $\Phi(c)=1$ if and only if $H_{\mathrm{pc}}c^{\mathsf T}=\mathbf{0}_{n-k}$, where the matrix-vector product is evaluated over $\Ftwo$.

\begin{algorithm}[t]
\caption{LP-GRAND Decoding}
\label{alg:lp_grand}
\begin{algorithmic}[1]
\Require Received vector $r\in\mathbb{R}^n$; 
specified precision matrix $Q=Q^{\mathsf T}\succ0$; 
membership test $\Phi:\{0,1\}^n\to\{0,1\}$; 
trellis or path-decomposition structure $\mathcal D$ that represents all unary terms and every interaction edge $\{i,j\}\in\mathcal E_Q$ exactly once; 
membership-query budget $\tau\in\{1,\ldots,2^n\}$.
% \Ensure An element of
% $\{c\in\{0,1\}^n:\Phi(c)=1\}\cup\{\mathsf{ABANDON}\}$.
\Ensure A codeword $\widehat c$ satisfying $\Phi(\widehat c)=1$, if one
is found within $\tau$ membership queries; otherwise, an abandonment
declaration.

\State Set $y_i\gets0$ if $r_i\geq0$, and $y_i\gets1$ otherwise,
for $i=1,\ldots,n$.
\State Set
$s\gets(1-2y_1,\ldots,1-2y_n)^{\mathsf T}$,
$a\gets r-s$, and $h\gets Qa$.

\For{$i=1,\ldots,n$}
    \State $\alpha_i\gets2s_i h_i+2Q_{ii}$.
\EndFor

\State $\mathcal E_Q\gets
\bigl\{\{i,j\}:1\leq i<j\leq n,\ Q_{ij}\neq0\bigr\}$.

\ForAll{$\{i,j\}\in\mathcal E_Q$}
    \State $\beta_{ij}\gets4s_i s_jQ_{ij}$.
\EndFor

\State Define the coefficient representation of
$\displaystyle
W_r(z)\gets
\sum_{i=1}^{n}\alpha_i z_i+
\!\!\! \sum_{\{i,j\}\in\mathcal E_Q} \!\!\! \beta_{ij}z_i z_j, \,\, \forall z\in\{0,1\}^n$.

\State $\mathsf{Gen}\gets
\Call{InitializeExactEnumerator}{W_r,\mathcal D}$.

\Statex \textbf{Enumerator contract:}
Successive calls to $\Call{Next}{\mathsf{Gen}}$ return every $z\in\{0,1\}^n$ exactly once in nondecreasing $W_r(z)$, with equal-energy configurations ordered by the enumerator's deterministic tie-breaking policy.

\For{$q=1,\ldots,\tau$}
    \State $z\gets\Call{Next}{\mathsf{Gen}}$.
    \State $\widehat c\gets y\xor z$.
    \If{$\Phi(\widehat c)=1$}
        \State \Return $\widehat c$.
    \EndIf
\EndFor

% \State \Return $\mathsf{ABANDON}$.
\State \Return an abandonment declaration.
\end{algorithmic}
\end{algorithm}

\begin{theorem}[ML optimality under complete enumeration]
\label{thm:ml}
Let $\calC\subseteq\{0,1\}^n$ be a nonempty binary codebook, and consider $R=x(C)+N$, $N\sim\mathcal{N}\bigl(\mathbf{0}_n,\Sigma\bigr)$, $\Sigma=\Sigma^{\mathsf T}\succ0$, where $C$ is uniformly distributed over $\calC$. For a fixed received observation $r$, suppose that the enumerator in Algorithm~\ref{alg:lp_grand} emits the patterns in the order defined by~\eqref{eq:order}. If $\tau=2^n$, then the first codebook-hit index
\begin{equation}
q_\star \triangleq
\min \bigl\{ q\in\{1,\ldots,2^n\}: y\xor z^{(q)}\in\calC \bigr\}
\label{eq:first_codebook_hit}
\end{equation}
is well defined, and $\widehat{c}_{\mathrm{ML}}\triangleq y\xor z^{(q_\star)}$ is an ML codeword. Because the prior on $\calC$ is uniform, $\widehat{c}_{\mathrm{ML}}$ is also a MAP codeword.
\end{theorem}

\begin{proof}
See Appendix~\ref{app:proof_thm_ml}.
\end{proof}

\begin{corollary}[Zero-syndrome membership test]
\label{cor:linear_membership}
In the setting of Theorem~\ref{thm:ml}, suppose that $\calC$ is a binary linear $[n,k]$ code with parity-check matrix $H_{\mathrm{pc}}\in\Ftwo^{(n-k)\times n}$. Then
\begin{equation}
\Phi(y\xor z)=1 \quad\Longleftrightarrow\quad
H_{\mathrm{pc}}(y\xor z)^{\mathsf T} = \mathbf{0}_{n-k},
\label{eq:syndrome_membership}
\end{equation}
where the matrix-vector product is evaluated over $\Ftwo$. Hence, the candidate order is determined by $r$ and $Q$ through $W_r$, whereas codebook dependence is confined to the membership test $\Phi$.
\end{corollary}

% ========================================
\section{Graphical Likelihood-Ordered Enumeration}
\label{sec:graphical}
% ========================================

By Theorem~\ref{thm:ml}, LP-GRAND is ML under complete enumeration when the noise-effect patterns are generated in nondecreasing order of $W_r$. In this section we construct layered graphical representations that generate this order exactly when the interaction graph admits a path decomposition of bounded width. We first consider precision matrices that are banded under the coordinate ordering $1,\ldots,n$ and then extend the construction to general path decompositions.

\subsection{Banded Precision Matrices and Trellis Construction}

\begin{definition}[Half-bandwidth]
\label{def:half_bandwidth}
A symmetric matrix $Q\in\mathbb{R}^{n\times n}$ has half-bandwidth at
most $\nu$, where $0\leq\nu\leq n-1$, if $Q_{ij}=0$ whenever
$|i-j|>\nu$.
\end{definition}

If $Q$ has half-bandwidth at most $\nu$, then $\beta_{ij}=0$ whenever
$|i-j|>\nu$, and \eqref{eq:quadratic_metric} becomes
\begin{equation}
W_r(z) = \sum_{i=1}^n \alpha_i z_i
+ \sum_{i=1}^n \sum_{d=1}^{\min\{\nu,i-1\}} \beta_{i-d,i}z_{i-d}z_i.
\label{eq:banded_metric}
\end{equation}
Each pairwise term $\beta_{i-d,i}z_{i-d}z_i$ is included in the branch cost at layer $i$, corresponding to the larger of its two coordinate indices.
For $\nu\geq1$, define the zero-padded extension of $z$ by
\begin{equation}
\bar z_j \triangleq
\begin{cases}
0,   & j\leq0,\\
z_j, & 1\leq j\leq n.
\end{cases}
\label{eq:zero_padded_bits}
\end{equation}
At the beginning of layer $i$, define the state
\begin{equation}
u_i \triangleq
\bigl( \bar z_{i-\nu}, \bar z_{i-\nu+1},
\ldots, \bar z_{i-1} \bigr)
\in \{0,1\}^{\nu}.
\label{eq:banded_state}
\end{equation}
Thus, $u_i$ contains the preceding $\nu$ binary coordinates, with zero padding for indices not exceeding zero. In particular, $u_1=(0,\ldots,0)\in\{0,1\}^{\nu}$.

For a state $u=(u_1,\ldots,u_\nu)\in\{0,1\}^{\nu}$ at layer $i$ and a branch label $b\in\{0,1\}$, define the branch cost
\begin{equation}
\gamma_i(u,b) \triangleq
\alpha_i b + \sum_{d=1}^{\min\{\nu,i-1\}} \beta_{i-d,i}u_{\nu-d+1}b.
\label{eq:branch}
\end{equation}
The successor-state mapping is
\begin{equation}
T(u,b) \triangleq
\begin{cases}
(b), & \nu=1,\\
(u_2,\ldots,u_\nu,b), & \nu\geq2.
\end{cases}
\label{eq:state_transition}
\end{equation}
A branch labeled $b=z_i$ connects $u_i$ to $u_{i+1}=T(u_i,z_i)$.
For $\nu=0$, define the state space as $\{\varnothing\}$ and set $\gamma_i(\varnothing,b) \triangleq \alpha_i b$, $T(\varnothing,b) \triangleq \varnothing$.

\begin{theorem}[Exact trellis representation]
\label{thm:trellis}
Suppose that $Q$ has half-bandwidth at most $\nu$, where $0\leq\nu\leq n-1$. Consider the layered trellis specified by the branch costs $\gamma_i$ and transition mapping $T$, with initial state $u_1=(0,\ldots,0)$ for $\nu\geq1$ and $u_1=\varnothing$ for $\nu=0$. Then length-$n$ paths beginning at $u_1$ are in one-to-one correspondence with binary patterns $z\in\{0,1\}^n$. For the path labeled by $z=(z_1,\ldots,z_n)$,
\begin{equation}
\sum_{i=1}^n \gamma_i(u_i,z_i) = W_r(z),
\label{eq:path_cost_identity}
\end{equation}
where $u_i$ denotes the state at the beginning of layer $i$.
\end{theorem}

\begin{proof}
See Appendix~\ref{app:proof_thm_trellis}.
\end{proof}

\begin{corollary}[Likelihood-ordered trellis enumeration]
\label{cor:exact_trellis_queries}
Under the assumptions of Theorem~\ref{thm:trellis}, suppose that an enumerator emits distinct length-$n$ trellis paths beginning at the initial state in nondecreasing total path cost, using the deterministic tie-breaking policy defining the sequence in~\eqref{eq:order}. Then, for every $1\leq K\leq2^n$, the branch-label sequences of the first $K$ emitted paths are the first $K$ noise-effect patterns in the LP-GRAND ordering, and their induced words are the first $K$ candidate queries. The trellis has at most $2^\nu$ states per layer, at most $(n+1)2^\nu$ states across all layers, and at most $n2^{\nu+1}$ branches. Consequently, for fixed $\nu$, both the number of states and the number of branches grow linearly with $n$.
\end{corollary}

\begin{proof}
See Appendix~\ref{app:proof_cor_exact_trellis_queries}.
\end{proof}

\begin{remark}[Negative branch costs]
\label{rem:negative_branch_costs}
The branch costs in~\eqref{eq:branch} need not be nonnegative. This does not affect the path-cost identity in Theorem~\ref{thm:trellis} or the resulting ordering of paths. If a path-enumeration method requires nonnegative branch costs, let $\mathcal{A}_i$ denote the set of branches at layer $i$ and define
\begin{equation}
m_i \triangleq \min_{(u,b)\in\mathcal{A}_i}\gamma_i(u,b),
\qquad
\gamma_i'(u,b) \triangleq \gamma_i(u,b)-m_i.
\label{eq:shifted_branch_cost}
\end{equation}
Then $\gamma_i'(u,b)\geq0, \forall (u,b)\in\mathcal{A}_i$. Because every length-$n$ path beginning at the initial state contains exactly one branch at each layer,
\begin{equation}
\sum_{i=1}^n\gamma_i'(u_i,z_i) =
W_r(z)-\sum_{i=1}^n m_i.
\end{equation}
The offset $\sum_{i=1}^n m_i$ is independent of the path. Hence, this
layerwise shift preserves all strict path-cost comparisons, tie classes, and the resulting path order.
\end{remark}

\subsection{Bounded-Pathwidth Graphical Enumeration}

The trellis associated with a banded precision matrix is a special case of the separator-state representation induced by a path decomposition. This subsection extends that representation to binary pairwise energies whose interaction graphs admit path decompositions of low width.

\begin{definition}[Path decomposition]
\label{def:path_decomposition}
Let $\mathcal{G}=(\mathcal{V},\mathcal{E})$ be a finite undirected graph. A path decomposition of $\mathcal{G}$ is a sequence of bags $\mathcal{B}_1,\ldots,\mathcal{B}_m$, where $\mathcal{B}_t\subseteq\mathcal{V}$, satisfying
\begin{enumerate}[label=(\roman*),leftmargin=*]
\item $\displaystyle\bigcup_{t=1}^m\mathcal{B}_t=\mathcal{V}$;
\item for every edge $\{i,j\}\in\mathcal{E}$, there exists $t\in\{1,\ldots,m\}$ such that $\{i,j\}\subseteq\mathcal{B}_t$;
\item for every $v\in\mathcal{V}$, the index set
$\bigl\{t\in \{1,\ldots,m\}: v\in\mathcal{B}_t\bigr\}$ is an interval of $\{1,\ldots,m\}$.
\end{enumerate}
The width of the decomposition is $\max_{1\leq t\leq m}\bigl(|\mathcal{B}_t|-1\bigr)$, and the pathwidth $\operatorname{pw}(\mathcal{G})$ is the minimum width over all path decompositions of $\mathcal{G}$.
\end{definition}

Given a vertex ordering $\pi=(\pi_1,\ldots,\pi_n)$, define the active frontier immediately before processing $\pi_t$ as
\begin{equation}
\mathcal{F}_t(\pi) \triangleq
\left\{ \pi_i: 1\leq i<t, \, \exists j \in \{t,\ldots,n\}
\text{ s.t. }
\{\pi_i,\pi_j\}\in\mathcal{E} \right\}.
\label{eq:active_frontier}
\end{equation}
Define the associated bag by
\begin{equation}
\mathcal{B}_t \triangleq \mathcal{F}_t(\pi) \cup \{\pi_t\},
\qquad t=1,\ldots,n.
\label{eq:frontier_bag}
\end{equation}
The sequence $(\mathcal{B}_t)_{t=1}^n$ forms a path decomposition. Vertex coverage holds because $\pi_t\in\mathcal{B}_t$ for every $t$. Moreover, if $\{\pi_i,\pi_j\} \in \mathcal{E}$ with $i<j$, then $\pi_i\in\mathcal{F}_j(\pi)$, and hence $\{\pi_i,\pi_j\}\subseteq\mathcal{B}_j$. For each $i\in\{1,\ldots,n\}$, define
\begin{equation}
\ell_i \triangleq \max\left( \{i\} \cup
\left\{ j\in\{i+1,\ldots,n\}: \{\pi_i,\pi_j\}\in\mathcal{E}\right\} \right).
\end{equation}
Then $\pi_i\in\mathcal{B}_t$ if and only if $i\leq t\leq\ell_i$. Hence, the set of bag indices containing $\pi_i$ is the interval $\{i,\ldots,\ell_i\}$, which establishes the running-intersection property.

The frontier width associated with a vertex ordering $\pi$ is
\begin{equation}
w_{\mathcal{G}}(\pi) \triangleq \max_{1\leq t\leq n}
\left|\mathcal{F}_t(\pi)\right|.
\label{eq:induced_frontier_width}
\end{equation}
Since $\pi_t\notin\mathcal{F}_t(\pi)$, $|\mathcal{B}_t|-1=|\mathcal{F}_t(\pi)|$, and therefore the path decomposition in~\eqref{eq:frontier_bag} has width $w_{\mathcal{G}}(\pi)$. By the vertex-separation characterization of pathwidth, $\operatorname{pw}(\mathcal{G})=\min_{\pi}w_{\mathcal{G}}(\pi)$, where the minimum is taken over all vertex orderings of $\mathcal{G}$~\cite{kinnersley1992vertex}.

We evaluate four deterministic vertex orderings: the coordinate ordering $1,\ldots,n$ and three graph-based heuristics, namely reverse Cuthill--McKee (RCM)~\cite{chan1980linear}, greedy minimum degree~\cite{george1989evolution}, and greedy minimum fill~\cite{rose1972graph,kjaerulff1990triangulation}. For the minimum-degree and minimum-fill heuristics, we take the elimination sequence directly as $\pi$, without reversal. For each ordering $\pi$, we construct the frontier bags on the original interaction graph and compute the induced width $w_{\mathcal{G}}(\pi)$ according to~\eqref{eq:active_frontier}--\eqref{eq:induced_frontier_width}. None of the three graph-based heuristics is guaranteed to produce an ordering $\pi$ satisfying $w_{\mathcal{G}}(\pi)=\operatorname{pw}(\mathcal{G})$.

\begin{theorem}[Exact path-decomposition representation]
\label{thm:pathwidth}
Let $\mathcal{G}=(\mathcal{V},\mathcal{E})$ be a finite undirected graph, and consider the binary pairwise energy
\begin{equation}
W(z) =\sum_{i\in\mathcal{V}}\theta_i(z_i)
+ \!\! \! \sum_{\{i,j\}\in\mathcal{E}} \theta_{\{i,j\}}(z_i,z_j),
\quad z\in\{0,1\}^{\mathcal{V}},
\label{eq:general_pairwise_energy}
\end{equation}
where $\theta_i:\{0,1\}\to\mathbb{R}$ and $\theta_{\{i,j\}}:\{0,1\}^2\to\mathbb{R}$.
Let $\mathcal{B}_1,\ldots,\mathcal{B}_m$ be a path decomposition of
$\mathcal{G}$ of width $w$. Then there exists a layered directed acyclic graph containing at most $2^{w+1}$ bag-assignment states in each layer such that its source-to-sink paths are in one-to-one correspondence with the assignments $z\in\{0,1\}^{\mathcal{V}}$. The local costs can be assigned so that each unary and pairwise term in
\eqref{eq:general_pairwise_energy} is included exactly once, and the
total cost of the path corresponding to $z$ is $W(z)$.
Consequently, an enumerator that emits distinct source-to-sink paths in nondecreasing total cost emits the corresponding assignments in nondecreasing order of $W(z)$. For $W=W_r$, applying the deterministic tie-breaking policy defining the sequence in~\eqref{eq:order} produces the LP-GRAND ordering of
noise-effect patterns and hence the corresponding candidate-word query order.
\end{theorem}

\begin{proof}
See Appendix~\ref{app:proof_thm_pathwidth}.
\end{proof}

Applying sum-product to the same layered representation, with each local cost $\lambda_t(a_t)$ assigned the factor $\exp\{-\lambda_t(a_t)\}$, computes the full-space partition function. Equivalently, the computation can be implemented by a log-sum-exp recursion in the log domain.

For a precision matrix $Q$ with half-bandwidth at most $\nu$, define $\mathcal{B}_t \triangleq \bigl\{\max\{1,t-\nu\},\ldots,t\bigr\}$, $t=1,\ldots,n$. These bags form a path decomposition of the interaction graph of $Q$ with width at most $\nu$. Indeed, every edge $\{i,j\}$ with $i<j$ satisfies $j-i\leq\nu$ and is therefore contained in $\mathcal{B}_j$. Moreover, each vertex $i$ belongs to the consecutive bags $\mathcal{B}_i,\ldots,\mathcal{B}_{\min\{n,i+\nu\}}$. For $t=1,\ldots,n-1$, let $\mathcal{S}_t\triangleq \mathcal{B}_t\cap\mathcal{B}_{t+1}$. Then $|\mathcal{S}_t|\leq\nu$, and, for $\nu\geq1$, the trellis state $u_{t+1}$ is the zero-padded representation of the binary assignment to $\mathcal{S}_t$, with components corresponding to nonpositive indices fixed to zero. For $\nu=0$, both the separator and the trellis state are empty. Hence, the banded trellis is a separator-state specialization of the general bag-assignment construction.

\begin{remark}[State and transition counts]
\label{rem:pathwidth_accounting}
The bound $2^{w+1}$ applies to the number of bag-assignment states in each layer. Between consecutive bags, a pair of assignments is compatible if the assignments agree on $\mathcal{B}_t\cap\mathcal{B}_{t+1}$. Each compatible pair corresponds uniquely to an assignment on $\mathcal{B}_t\cup\mathcal{B}_{t+1}$. Consequently, the number of transitions between these layers is $2^{|\mathcal{B}_t\cup\mathcal{B}_{t+1}|} \leq 2^{2(w+1)}$. Our implementation groups assignments according to their restrictions to the separator $\mathcal{B}_t\cap\mathcal{B}_{t+1}$, thereby avoiding exhaustive pairwise compatibility testing. This separator-based indexing does not reduce the worst-case number of compatible transitions. For the banded trellis, each nonterminal state has at most two outgoing branches. Backward suffix dynamic programming therefore evaluates at most $n2^{\nu+1}$ branches and stores at most $(n+1)2^\nu$ suffix costs.
\end{remark}

\begin{remark}[Output-sensitive enumeration complexity]
\label{rem:output_sensitive_cost}
For the banded best-first enumerator, let $B_K$ denote the number of extract-min operations performed before the first $K$ length-$n$ trellis paths are emitted, and let $H_K$ denote the maximum priority-queue size during this process. With a binary heap and constant-time branch extension and key evaluation, the priority-queue time required to emit these paths is $\mathcal{O} \bigl(B_K\log(1+H_K)\bigr)$, because each extracted nonterminal path prefix generates at most two insertions. The priority queue stores at most $H_K$ records. If path reconstruction stores one predecessor record for each generated path prefix, it requires an additional $\mathcal{O}(B_K)$ records. Neither bounded pathwidth nor fixed $\nu$ provides a nontrivial codebook-independent upper bound on the number of membership queries required to reach the first codebook hit. In the worst case, Algorithm~\ref{alg:lp_grand} tests all $2^n$ noise-effect patterns.
\end{remark}

Appendix~\ref{app:k_shortest_paths} describes the suffix-guided best-first enumerator, establishes that it emits complete paths in nondecreasing order of total cost, and gives output-sensitive bounds on its time and storage requirements.

The bounds for coefficient quantization and precision-matrix mismatch, as well as the results on the full-space partition function, auxiliary-model first-hit correctness estimates, and abandonment, are given in Appendices~\ref{sec:robustness} and~\ref{sec:softout}.

\section{Numerical Evaluation}
\label{sec:numerical_evaluation}
% ======================================================================

We evaluate: (i) agreement of the pseudo-Boolean metric, the ordering of all $2^n$ candidate words, and the partition function with exhaustive computation; (ii) exact enumeration on nonbanded interaction graphs of low pathwidth; (iii) agreement between LP-GRAND and direct codeword-wise ML decoding for small fixed codes; (iv) agreement among independently implemented exact enumerators using the same candidate metric; (v) BLER and algorithmic work over a random-code ensemble; (vi) sensitivity to the code rate, noise correlation coefficient, and membership-query budget; (vii) performance for selected fixed codes at the principal blocklength considered; and (viii) the effects of coefficient quantization and precision-matrix mismatch on the candidate ordering, together with the empirical calibration of the auxiliary-model first-hit
correctness estimates.

Pointwise BLER estimates are reported with $95\%$ Wilson score intervals \cite{wilson1927probable}. We report the algorithm-specific work measures separately: codebook-membership tests, priority-queue extract-min operations, generated partial paths, branch evaluations, local metric evaluations, rejected generator states, and peak frontier or list size. Because these counters represent different operations, we do not combine them into a single complexity measure. Wall-clock times are reported as our software-level measurements only and are not used as measures of hardware complexity.

\subsection{Experimental Protocol and Reproducibility}

For the channel model, we use stationary first-order Gauss--Markov noise, $N_1\sim\mathcal{N}(0,\sigma^2)$ and $N_i=\rho N_{i-1}+V_i$, $i=2,\ldots,n$, where $V_2,\ldots,V_n\stackrel{\mathrm{i.i.d.}}{\sim}\mathcal{N}\bigl(0,\sigma^2(1-\rho^2)\bigr)$ are independent of $N_1$, and $|\rho|<1$. This construction gives $\operatorname{Cov}(N_i,N_j)=\sigma^2\rho^{|i-j|}$. For a binary code of rate $R_{\mathrm c}\triangleq k/n$, unit-energy BPSK, and a specified nominal value of $(E_b/N_0)_{\mathrm{dB}}$, the marginal noise variance is set to
\begin{equation}
\sigma^2 = \left[ 2R_{\mathrm c} 10^{(E_b/N_0)_{\mathrm{dB}}/10} \right]^{-1}.
\label{eq:numerical_noise_variance}
\end{equation}
For $\rho=0$, this is the standard AWGN relation $\sigma^2=N_0/2$. For $\rho\neq0$, $(E_b/N_0)_{\mathrm{dB}}$ serves as a nominal SNR parameter defined through the marginal noise variance. This parameterization keeps $\operatorname{Var}(N_i)=\sigma^2$ fixed as $\rho$ varies, but it does not preserve the joint noise distribution. In particular, varying $\rho$ changes the innovation variance, power spectral density, finite-block covariance determinant, and entropy rate. The correlation sweep therefore compares channels at fixed marginal noise variance, rather than at fixed innovation variance or fixed spectral shape.

In the random-code ensemble experiment, we use $[n,k]=[64,52]$, $\rho=0.5$, and a membership-query budget of $q_{\max}=20000$. For each frame, we construct a systematic generator matrix of the form $G=[\,I_k\;\;P\,], P\in\mathbb{F}_2^{k\times(n-k)}$, where the entries of $P$ are drawn independently from $\operatorname{Bernoulli}(1/2)$, and select the transmitted message uniformly from $\{0,1\}^k$. Within each frame, all methods use the same code realization, transmitted message, and received vector. A query-based decoder that does not return a codeword within $q_{\max}$ membership tests declares abandonment, which is counted as a block error. The reported BLER is averaged over the random systematic code ensemble, transmitted messages, and channel realizations.
For the baseline denoted ExactMemoryless-GRAND, we replace $\Sigma$ by $\sigma^2 I_n$ and enumerate candidates in nondecreasing order of the resulting exact memoryless Gaussian metric
$W_{\mathrm{mem},r}(z)\triangleq \frac{2}{\sigma^2}\sum_{i=1}^{n}|r_i|z_i$. Thus, this baseline retains the coordinatewise reliability magnitudes but ignores the noise correlation. The results labeled basic ORBGRAND use the logistic weight $w_{\mathrm L}(\tilde z)=\sum_{i=1}^{n}i\tilde z_i$, with rank one assigned to the least reliable coordinate; no piecewise-linear approximation to the ordered reliability magnitudes is used.

Let $b$ divide $n$, let $J\triangleq n/b$, and partition $\{1,\ldots,n\}$ into the consecutive blocks $B_j \triangleq \{(j-1)b+1,\ldots,jb\}$, $j=1,\ldots,J$. We define the exact block-product metric by
\begin{equation}
E_{\mathrm{blk}}(c) \triangleq \frac{1}{2} \sum_{j=1}^{J}
\bigl(r_{B_j}-x(c)_{B_j}\bigr)^{\mathsf T}
\Sigma_{B_j,B_j}^{-1}
\bigl(r_{B_j}-x(c)_{B_j}\bigr),
\label{eq:block_product_metric}
\end{equation}
where $\Sigma_{B_j,B_j}$ is the principal submatrix of $\Sigma$ indexed by $B_j$. In the implementation, each block energy is evaluated using a Cholesky factor of $\Sigma_{B_j,B_j}$, without explicitly forming its inverse. Up to terms independent of $c$, \eqref{eq:block_product_metric} is the negative logarithm of the product of the exact Gaussian marginal likelihoods of the individual blocks. We refer to exact enumeration under this block-product approximation as ExactBlockProduct; here, ``exact'' qualifies the enumeration of $E_{\mathrm{blk}}$, not the matched full-block Gaussian likelihood.

For each block $B_j$, we evaluate all $2^b$ binary assignments and select a minimum-energy assignment using a fixed deterministic tie-breaking rule. This assignment serves as the blockwise baseline. Each remaining assignment defines one substitution relative to the baseline and is assigned its local excess energy. A global candidate selects either the baseline or one alternative assignment in each block, so its excess block-product metric is the sum of the corresponding local excess energies. Under complete enumeration, the resulting subset-sum enumerator outputs all candidates in nondecreasing $E_{\mathrm{blk}}$, with ties resolved deterministically.
The block-product metric is not obtained by simply retaining the corresponding principal blocks of the full precision matrix, since $\Sigma_{B_j,B_j}^{-1} \neq \bigl(\Sigma^{-1}\bigr)_{B_j,B_j}$ in general. Relative to the matched full-block likelihood, the block-product approximation removes dependence across block boundaries and may also alter the within-block precision coefficients. For the first-order Gauss--Markov model, the altered diagonal coefficients occur at the artificial block boundaries.

For comparison, we implemented the blockwise approximate-independence construction described for ORBGRAND-AI~\cite{duffy2023using}. Our implementation uses the same nonbaseline block assignments and local excess energies as ExactBlockProduct. We pool the substitutions from all blocks, sort them by increasing excess energy using a fixed deterministic tie-breaking rule, and assign each substitution its
position in this ordered list as its reliability rank. Conflict-free sets of substitutions are then enumerated in nondecreasing sum of their ranks. ExactBlockProduct and our ORBGRAND-AI implementation therefore use the same local alternatives but assign them exact excess-energy weights and rank weights, respectively.

Under the valid-membership-test budget used in the main experiments, a generator state containing more than one substitution from the same block is discarded before the codebook-membership test and does not count toward $q_{\max}$. It is nevertheless included in the priority-queue-removal and conflict-rejection counts. We report codebook-membership tests, priority-queue removals, conflict rejections, local block-energy evaluations, and peak priority-queue size separately. We also evaluate the alternative convention in which every removed generator state counts toward the budget. The reported ORBGRAND-AI results are obtained from our implementation of the published rank-based construction under these stated budget conventions, so we do not claim to reproduce a particular external software or hardware implementation.

\subsection{Exact Cost-Order Validation for Stored Coefficients}
\label{ssec:complete_order_validation}

Our analytical exactness results are stated in real arithmetic. We therefore validate the cost ordering induced by the coefficients stored by the implementation in IEEE~754 binary64 format, without introducing floating-point comparison tolerances. For each pair $n\in\{8,10,12\}$, $\nu\in\{0,1,2,3,4\}$, $\nu<n$, we generated five banded symmetric positive-definite matrices from banded Cholesky factors, followed by diagonal rescaling. For three of the five instances associated with each $(n,\nu)$ pair, the generation parameters and received-vector scale were chosen to produce smaller positive differences between candidate costs. The resulting test set contains $75$ instances.

For each instance, all $2^n$ configurations were evaluated. Each stored binary64 unary and pairwise coefficient was converted exactly to a dyadic rational. We then multiplied all coefficients by a common positive power of two that cleared their denominators. This produced an integer metric with the same strict comparisons and equal-cost classes as the exact real-valued metric defined by the stored binary64 coefficients. The integer costs were evaluated using arbitrary-precision integer arithmetic.

For each instance, we sorted all configurations by exact integer cost and compared these costs with the output sequences of three separately implemented enumerators: the suffix-guided best-first LP-GRAND enumerator, the state-list $K$-best Viterbi dynamic program \cite{seshadri1994list} with $K=2^n$, and an integer-key trellis enumerator. For each enumerator, we verified that every configuration appeared exactly once and that the exact integer costs were nondecreasing along the output sequence. Configurations with the same exact cost could appear in any order. Equal-cost classes were determined by exact integer equality, without a numerical tolerance.

\begin{table}[!b]
\centering
\caption{Exact cost-order validation for the stored binary64 coefficients. A failure denotes an instance in which an enumerator omitted or duplicated a configuration, or returned a configuration with a smaller exact integer cost after one with a larger cost. The minimum positive gap is the smallest difference between distinct exact integer costs after rescaling to the original metric units. The final column gives the maximum absolute difference between the trellis and exhaustive floating-point evaluations of $\log Z_r$.}
\label{tab:exactness}
\scriptsize
\setlength{\tabcolsep}{2.6pt}
\resizebox{\columnwidth}{!}{%
\begin{tabular}{ccccccc}
\toprule
$n$
& Cases
& \shortstack{Best-first\\failures}
& \shortstack{State-list\\failures}
& \shortstack{Integer-key\\failures}
& \shortstack{Minimum positive\\gap}
& \shortstack{Maximum\\$|\Delta\log Z_r|$}
\\
\midrule
$8$
& $25$
& $0$
& $0$
& $0$
& $5.66\!\times\!10^{-12}$
& $8.88\!\times\!10^{-16}$
\\
$10$
& $25$
& $0$
& $0$
& $0$
& $3.33\!\times\!10^{-12}$
& $1.78\!\times\!10^{-15}$
\\
$12$
& $25$
& $0$
& $0$
& $0$
& $1.23\!\times\!10^{-13}$
& $1.78\!\times\!10^{-15}$
\\
\bottomrule
\end{tabular}}
\end{table}

For every test instance, all three enumerators returned every configuration exactly once in nondecreasing exact integer cost. Configurations with equal costs could occur in different orders. Across all tested configurations, the largest absolute difference between the direct quadratic-form value $E_r(z)-E_r(0)$ and the pseudo-Boolean value $W_r(z)$ was $5.68\times10^{-14}$. The largest absolute difference between $W_r(z)$ and the sum of the corresponding trellis branch costs was $2.84\times10^{-14}$. These differences result from evaluating algebraically equivalent expressions in binary64 arithmetic; the ordering comparisons were performed using exact integer costs. The tests therefore verify the cost order defined by the stored binary64 coefficients. They do not measure the rounding error between the stored coefficients and the real-arithmetic coefficients defined by~\eqref{eq:alpha} and~\eqref{eq:beta}.

\subsection{Nonbanded Low-Pathwidth Validation}
\label{ssec:nonbanded_validation}

For the small-instance validation, we generated five sparse symmetric positive-definite precision matrices for each combination of $n\in\{8,10,12\}$ and interaction-graph topology: path, binary tree, or ladder. The weights were generated independently for each matrix, and an independent random coordinate permutation was then applied. For each permuted instance, we considered the coordinate, reverse Cuthill--McKee, minimum-degree, and minimum-fill orderings and constructed the corresponding frontier bags using \eqref{eq:frontier_bag}. This procedure produced $180$ graph--ordering pairs. For each pair, we compared the complete cost order over all $2^n$ assignments and the log-partition function with exhaustive computation. Every assignment appeared exactly once in nondecreasing exact stored-coefficient cost; assignments having the same cost could appear in any order. The maximum absolute difference between the graphical and exhaustive evaluations of the log-partition function was $8.88\times10^{-16}$.

For the larger nonbanded instances, we generated randomly permuted path, binary-tree, and ladder precision graphs for $n\in\{64,128,256\}$. We evaluated the same four vertex orderings using the induced frontier width in~\eqref{eq:induced_frontier_width} and recorded the half-bandwidth under the coordinate ordering separately. For each graph--ordering pair with induced width at most six, we enumerated the first $K=1024$ assignments. Under the RCM ordering, the median induced widths were one for the path instances and two for the ladder instances. Across the $54$ feasible path and ladder graph--ordering pairs, the compared enumerators returned the same sequence of exact stored-coefficient cost levels and the same set of assignments at each level. Over the first $1024$ assignments, the maximum absolute difference between direct evaluation of the pseudo-Boolean energy and the corresponding graphical path cost was $8.88\times10^{-16}$.

For the permuted binary-tree instances, the induced widths obtained from the tested orderings ranged from $16$ to $65$. None of these instances satisfied the width-six limit, and they were therefore not enumerated. The induced frontier widths under the coordinate and RCM orderings are shown in the left panel of Fig.~\ref{fig:pathwidth_timing}. These tests evaluate the graphical enumeration procedure on synthetic low-pathwidth interaction graphs; they do not estimate decoding performance averaged over a channel or code ensemble.

To test end-to-end decoding with a precision matrix that is nonbanded under the coordinate ordering, we fixed a random linear $[64,52]$ code and applied a fixed random coordinate permutation to the covariance matrix of stationary first-order Gauss--Markov noise with $\rho=0.5$. The resulting precision matrix has coordinate-order half-bandwidth $56$, whereas the RCM ordering has induced frontier width one. We evaluated all methods on the same $50$ frames at the nominal operating point $(E_b/N_0)_{\mathrm{dB}}=2$. LP-GRAND using the RCM decomposition produced no observed block errors and required an average of $102.8$ codebook-membership tests. Table~\ref{tab:nonbanded_bler} reports the results for all methods. This experiment verifies end-to-end decoding for this permuted-chain channel and fixed code; it does not establish average performance over a broader class of sparse channels or codes.

\begin{table}[!b]
\centering
\caption{Fixed-code results for the permuted first-order Gauss--Markov precision chain with $n=64$, $k=52$, $\rho=0.5$, nominal $(E_b/N_0)_{\mathrm{dB}}=2$, and $50$ common frames. The coordinate-order half-bandwidth is $56$, and the RCM-induced frontier width is one. The intervals are $95\%$ Wilson score intervals.}
\label{tab:nonbanded_bler}
\scriptsize
\setlength{\tabcolsep}{3.0pt}
\resizebox{\columnwidth}{!}{%
\begin{tabular}{lrrr}
\toprule
Method
& Errors/frames
& Observed BLER $[95\%\ \mathrm{CI}]$
& \shortstack{Average valid\\membership tests}
\\
\midrule
LP-GRAND, RCM decomposition
& $0/50$
& $0.00\ [0.000,0.071]$
& $102.8$
\\
ExactBlockProduct, $b=8$
& $18/50$
& $0.36\ [0.241,0.499]$
& $987.7$
\\
ORBGRAND-AI (our implementation), $b=8$
& $17/50$
& $0.34\ [0.224,0.478]$
& $1055.6$
\\
ExactMemoryless-GRAND
& $25/50$
& $0.50\ [0.366,0.634]$
& $1690.2$
\\
\bottomrule
\end{tabular}}
\end{table}

\subsection{Direct ML Verification and Same-Metric Exact Enumeration}
\label{ssec:direct_ml_and_exact_generators}

We first verify LP-GRAND against exhaustive codeword ML decoding for two fixed binary $[20,12]$ codes: a systematic random linear code (RLC) and a systematic CRC-8 code. The CRC generator polynomial is
\begin{equation}
g(x)=x^8+x^2+x+1.
\end{equation}
The CRC encoder processes the message bits MSB first, starts from the all-zero register state, and applies no final XOR. For each received vector, exhaustive decoding evaluates the matched Gaussian energy of all $2^{12}=4096$ codewords.

Each code is evaluated over $5000$ frames at nominal $(E_b/N_0)_{\mathrm{dB}}=2$ and $\rho=0.5$. Let $\widehat E_{r,\mathrm{bin64}}(c)$ denote the value obtained by direct IEEE~754 binary64 evaluation of $\frac{1}{2}\bigl(r-x(c)\bigr)^{\mathsf T}Q\bigl(r-x(c)\bigr)$. For each frame, define the binary64 reference set
\begin{equation}
\mathcal M_{\mathrm{bin64}}(r) \triangleq \operatorname*{arg\,min}_{c\in\calC} \widehat E_{r,\mathrm{bin64}}(c).
\label{eq:computed_ml_set}
\end{equation}
A decoder agrees with the exhaustive reference when its output belongs to $\mathcal M_{\mathrm{bin64}}(r)$. Equality between stored binary64 energy values is tested exactly, without a numerical tolerance.

Neither code produced a frame with more than one minimizer under the stored binary64 energies. LP-GRAND returned the unique element of $\mathcal M_{\mathrm{bin64}}(r)$ in every tested frame. Its decisions and observed BLER therefore agreed with exhaustive codeword decoding on all tested frames. This is a finite-precision comparison and does not quantify the difference between the stored binary64 metric and its real-arithmetic counterpart.

\begin{figure*}[t]
\centering
\includegraphics[width=0.8\linewidth]
{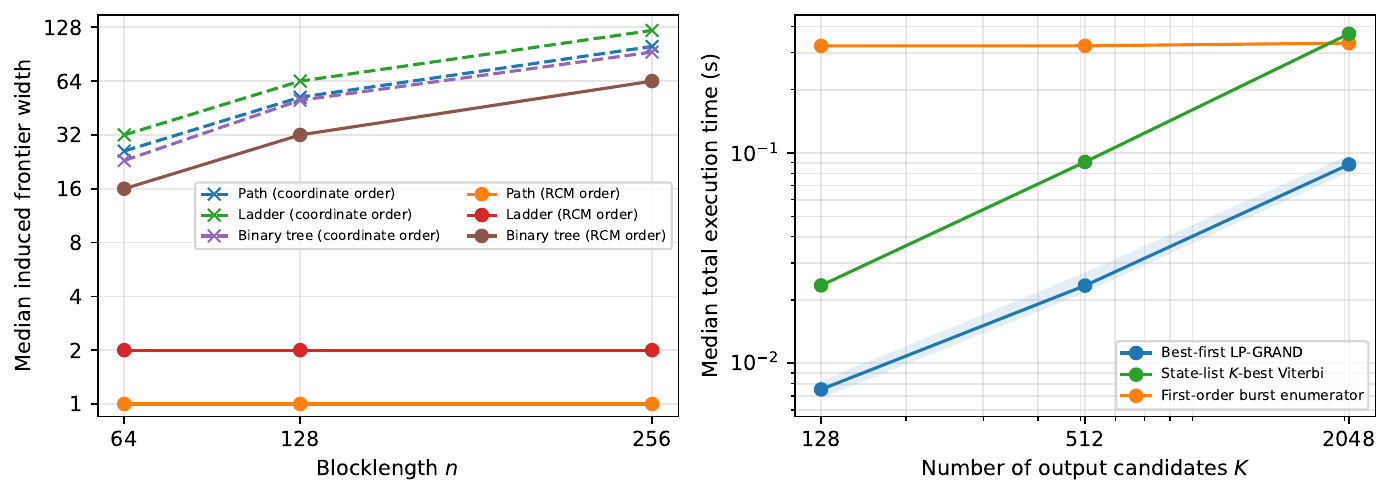}
\caption{Left: induced frontier width under the coordinate and reverse Cuthill--McKee orderings for randomly permuted path, binary-tree, and ladder interaction graphs with $n\in\{64,128,256\}$; the vertical axis uses a base-two logarithmic scale. Right: median total execution time, with interquartile ranges, of the three same-metric candidate generators for the two-tap finite-block ISI model as a function of $K$. The exact cost sequences agree for the first $K$ outputs; any differences in candidate order are confined to equal-cost classes. The timing results correspond to our reference Python implementations.}
\label{fig:pathwidth_timing}
\end{figure*}

\begin{table}[!b]
\centering
\caption{Exhaustive codeword-ML verification for two fixed $[20,12]$ codes over $5000$ frames per code at nominal $(E_b/N_0)_{\mathrm{dB}}=2$ and $\rho=0.5$. The BLER intervals are $95\%$ Wilson score intervals. Reference-set agreement is the fraction of frames for which the decoder output belongs to $\mathcal M_{\mathrm{bin64}}(r)$ in~\eqref{eq:computed_ml_set}.}
\label{tab:direct_ml}
\scriptsize
\setlength{\tabcolsep}{2.6pt}
\resizebox{\columnwidth}{!}{%
\begin{tabular}{llccc}
\toprule
Code
& Method
& Observed BLER $[95\%\ \mathrm{CI}]$
& \shortstack{Average valid\\membership tests}
& \shortstack{Reference-set\\agreement}
\\
\midrule
RLC
& Exhaustive ML decoding
& $0.0198\ [0.0163,0.0240]$
& --
& $1.0000$
\\
RLC
& LP-GRAND
& $0.0198\ [0.0163,0.0240]$
& $8.29$
& $1.0000$
\\
RLC
& ORBGRAND-AI (our implementation), $b=4$
& $0.0360\ [0.0312,0.0415]$
& $12.30$
& $0.9736$
\\
RLC
& Basic ORBGRAND
& $0.1248\ [0.1159,0.1342]$
& $37.37$
& $0.8770$
\\
\midrule
CRC-8
& Exhaustive codeword ML
& $0.0206\ [0.0170,0.0249]$
& --
& $1.0000$
\\
CRC-8
& LP-GRAND
& $0.0206\ [0.0170,0.0249]$
& $8.72$
& $1.0000$
\\
CRC-8
& ORBGRAND-AI (our implementation), $b=4$
& $0.0404\ [0.0353,0.0462]$
& $12.93$
& $0.9728$
\\
CRC-8
& Basic ORBGRAND
& $0.1306\ [0.1215,0.1402]$
& $37.34$
& $0.8730$
\\
\bottomrule
\end{tabular}}
\end{table}

We next compare three separately implemented candidate generators under a common finite-block Gaussian likelihood metric. Consider the causal two-tap ISI model $\bar r=H_{\mathrm{isi}}x(c)+w$, where $w\sim\mathcal N\bigl(\mathbf 0_n,\sigma_w^2I_n\bigr)$, $\sigma_w^2=0.8$, $n=64$, and the input symbol immediately preceding the block is known to be zero. The channel matrix is defined by $(H_{\mathrm{isi}})_{i,i}=1$, $i=1,\ldots,n$, and $(H_{\mathrm{isi}})_{i,i-1}=0.5$, $i=2,\ldots,n$, with all other entries zero. Since $H_{\mathrm{isi}}$ is unit lower triangular, it is nonsingular. Define $r\triangleq H_{\mathrm{isi}}^{-1}\bar r$. Then $r=x(c)+N$, $N\sim\mathcal N\!\left( \mathbf 0_n, \sigma_w^2H_{\mathrm{isi}}^{-1}H_{\mathrm{isi}}^{-\mathsf T} \right)$, and the corresponding precision matrix is
\begin{equation}
Q = \frac{1}{\sigma_w^2} H_{\mathrm{isi}}^{\mathsf T}H_{\mathrm{isi}}.
\label{eq:two_tap_precision}
\end{equation}
Consequently, for every candidate $c$,
\begin{equation}
\frac{1}{2} \bigl(r-x(c)\bigr)^{\mathsf T}Q\bigl(r-x(c)\bigr)
= \frac{1}{2\sigma_w^2} \left\lVert\bar r-H_{\mathrm{isi}}x(c)\right\rVert_2^2.
\label{eq:isi_metric_equivalence}
\end{equation}
The matrix $Q$ is tridiagonal and has half-bandwidth $\nu=1$.
The following generators use the metric in \eqref{eq:isi_metric_equivalence}:
\begin{enumerate}
\item the suffix-guided best-first LP-GRAND enumerator;
\item a state-list $K$-best Viterbi dynamic program that retains the $K$ lowest-cost partial paths at each state; and
\item a separately implemented first-order burst enumerator used as a validation reference.
\end{enumerate}

The burst enumerator first selects a minimum-energy binary sequence using a fixed deterministic tie-breaking rule. For any other sequence, the coordinates that differ from this baseline decompose uniquely into maximal contiguous runs. Because the metric contains only unary and nearest-neighbor pairwise terms, the excess cost of a set of separated runs equals the sum of the individual run excess costs, where each run cost includes the pairwise terms at its boundaries with unchanged coordinates. Our implementation evaluates the excess cost of each of the $n(n+1)/2$ possible contiguous runs and enumerates compatible sets of runs in nondecreasing total excess cost. Sets containing overlapping or adjacent runs are rejected, since such runs do not form distinct maximal components. This construction uses the error-burst decomposition principle of first-order SGRAND-ISI~\cite{li2026grand_isi}. Our burst enumerator was implemented separately and is used only as a same-metric validation reference; it is not a reproduction of the published SGRAND-ISI generator or its data structures.

For ten independently generated received vectors and $K\in\{128,512,2048\}$, we rescored the first $K$ outputs of each generator using the exact scaled-integer metric described in Section~\ref{ssec:complete_order_validation}. For every received vector and every value of $K$, the three generators produced identical nondecreasing sequences of exact costs. Any differences in candidate order occurred only among candidates with the same exact cost. Since all three generators use the same metric, the work counts and execution times in Table~\ref{tab:exact_enum} compare the stated implementations rather than different likelihood approximations.

\begin{table}[!b]
\centering
\caption{Same-metric enumeration for the two-tap finite-block ISI model with $n=64$ and $\nu=1$. Entries are medians over ten received vectors. The reported times are measurements of the stated reference Python implementations.}
\label{tab:exact_enum}
\scriptsize
\setlength{\tabcolsep}{2.5pt}
\resizebox{\columnwidth}{!}{%
\begin{tabular}{llccr}
\toprule
Generator
& $K$
& \shortstack{Preprocessing time (s)}
& \shortstack{Total time (s)}
& Recorded operations
\\
\midrule
Best-first LP-GRAND
& $128$
& $2.43\!\times\!10^{-4}$
& $7.51\!\times\!10^{-3}$
& $2526$ heap removals
\\
State-list $K$-best
& $128$
& $2.28\!\times\!10^{-2}$
& $2.35\!\times\!10^{-2}$
& $14847$ retained paths
\\
First-order burst reference
& $128$
& $3.22\!\times\!10^{-1}$
& $3.23\!\times\!10^{-1}$
& $2080$ run costs; $25$ rejected conflicts
\\
\midrule
Best-first LP-GRAND
& $512$
& $2.34\!\times\!10^{-4}$
& $2.34\!\times\!10^{-2}$
& $8201$ heap removals
\\
State-list $K$-best
& $512$
& $8.80\!\times\!10^{-2}$
& $9.07\!\times\!10^{-2}$
& $57343$ retained paths
\\
First-order burst reference
& $512$
& $3.20\!\times\!10^{-1}$
& $3.24\!\times\!10^{-1}$
& $2080$ run costs; $179$ rejected conflicts
\\
\midrule
Best-first LP-GRAND
& $2048$
& $2.49\!\times\!10^{-4}$
& $8.82\!\times\!10^{-2}$
& $29489$ heap removals
\\
State-list $K$-best
& $2048$
& $3.59\!\times\!10^{-1}$
& $3.69\!\times\!10^{-1}$
& $221183$ retained paths
\\
First-order burst reference
& $2048$
& $3.16\!\times\!10^{-1}$
& $3.33\!\times\!10^{-1}$
& $2080$ run costs; $1221$ rejected conflicts
\\
\bottomrule
\end{tabular}}
\end{table}

Since the last column in Table~\ref{tab:exact_enum} counts a different operation for each generator, its values are not directly comparable. The runtimes were obtained with our implementations and do not represent those of the published SGRAND-ISI implementation. The right panel of Fig.~\ref{fig:pathwidth_timing} shows the median total execution time, with interquartile ranges, of the three same-metric candidate generators as a function of $K$.

\subsection{Dependence on Code Rate, Correlation, and Query Budget}
\label{ssec:rate_correlation_abandonment}

Figure~\ref{fig:rate_sweep} reports results for $k\in\{40,44,48,52,56\}$ with $n=64$, nominal $(E_b/N_0)_{\mathrm{dB}}=2.5$, $\rho=0.5$, and $q_{\max}=20{,}000$. Each rate is evaluated over $5000$ random-code ensemble frames. All methods use the same code, transmitted message, and received vector within a frame. Separate frame sets are generated for the different rates.

LP-GRAND had the lowest empirical BLER and the smallest sample mean number of valid codebook-membership tests at every tested rate. At $R_{\mathrm c}=48/64=0.75$, its empirical BLER was $0.0028$, and the sample mean number of valid membership tests was $132.3$. The corresponding values were $0.0062$ and $227.5$ for ExactBlockProduct with $b=8$, and $0.0126$ and $350.6$ for our ORBGRAND-AI implementation with $b=8$.

As specified in~\eqref{eq:numerical_noise_variance}, fixing the nominal $E_b/N_0$ while varying $R_{\mathrm c}$ also changes the marginal noise variance. The rate sweep therefore changes both code redundancy and marginal noise variance; it is not a comparison at fixed symbol SNR.

\begin{figure*}[t]
\centering
\includegraphics[width=0.75\linewidth]
{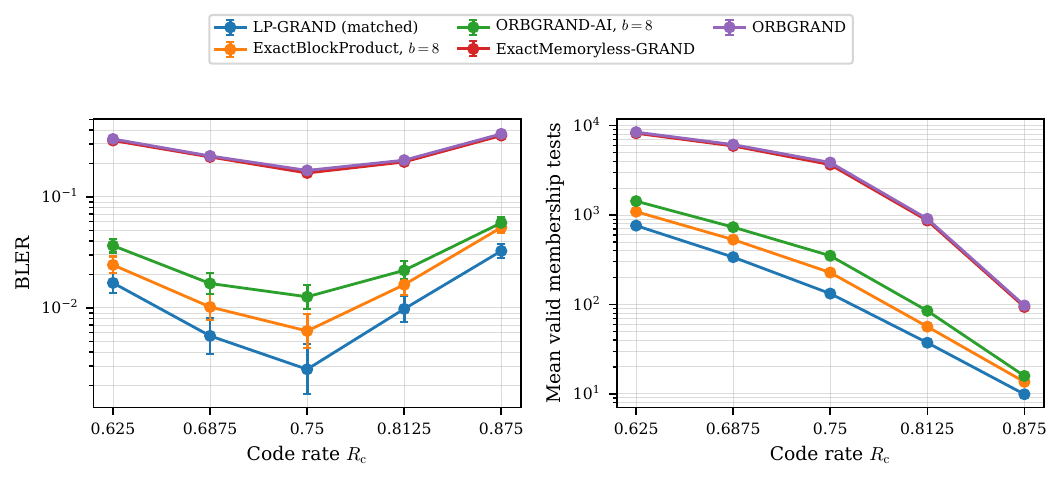}
\vspace{-4pt}
\caption{Random-linear-code ensemble results as a function of code rate with $n=64$, nominal $(E_b/N_0)_{\mathrm{dB}}=2.5$, $\rho=0.5$, $q_{\max}=20{,}000$, and $5000$ frames per rate. Left: empirical BLER. Right: sample mean number of valid codebook-membership tests. All methods use the same frames at each rate; separate frame sets are used for different rates.}
\label{fig:rate_sweep}
\end{figure*}

The left panel of Fig.~\ref{fig:rho_qmax} shows the empirical BLER as a function of the first-order Gauss--Markov correlation coefficient. Each value of $\rho$ is evaluated over $5000$ random-code ensemble frames at nominal $(E_b/N_0)_{\mathrm{dB}}=2$. At $\rho=0$, $\Sigma=\sigma^2 I_n$ and $Q=\sigma^{-2}I_n$. LP-GRAND, ExactMemoryless-GRAND, and ExactBlockProduct with $b=8$ therefore use the same candidate metric. Each had an empirical BLER of $0.3252$.

At $\rho=0.5$, the empirical BLERs of LP-GRAND, ExactMemoryless-GRAND, and ExactBlockProduct with $b=8$ were $0.0280$, $0.3396$, and $0.0434$, respectively. Our ORBGRAND-AI implementation with $b=8$ had an empirical BLER of $0.0542$. The matched Gaussian metric contains nearest-neighbor interaction terms when $\rho\neq0$. ExactMemoryless-GRAND omits these terms, whereas ExactBlockProduct retains the within-block dependence but omits dependence across block boundaries.

The right panel of Fig.~\ref{fig:rho_qmax} shows the empirical BLER as a function of the valid membership-test budget. These values are obtained from the verified $10^4$-frame paired experiment at nominal $(E_b/N_0)_{\mathrm{dB}}=2$ and $\rho=0.5$. For a budget $q_{\max}$, a frame is declared abandoned when its recorded first-hit query count exceeds $q_{\max}$; no additional decoding runs are required. For LP-GRAND, increasing $q_{\max}$ from $100$ to $5000$ reduced the empirical BLER from $0.1336$ to $0.0281$ and increased the codeword-return rate from $0.8718$ to $0.9977$.

At $q_{\max}=5000$, LP-GRAND abandoned $23$ frames and made $281$ block errors. In the remaining $258$ error events, the decoder returned a codeword, but the returned ML codeword differed from the transmitted codeword. Thus, these events are decoding errors rather than ordering failures.
At $q_{\max}=20{,}000$, LP-GRAND abandoned no frames and made $274$ block ML decision errors.

\begin{figure*}[t]
\centering
\begin{minipage}{0.4\linewidth}
\centering
\includegraphics[width=\linewidth]{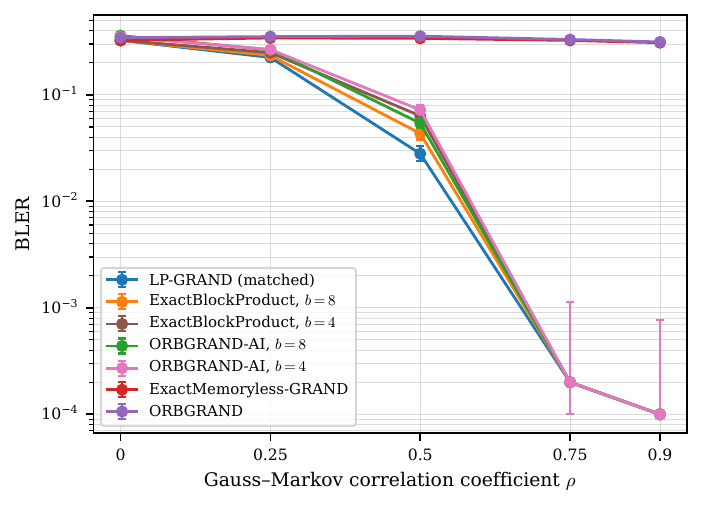}
\end{minipage}
\begin{minipage}{0.4\linewidth}
\centering
\includegraphics[width=\linewidth]{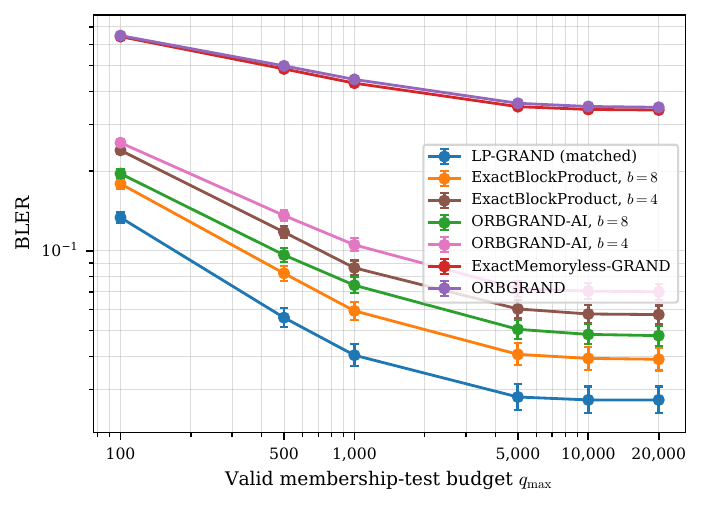}
\end{minipage}
\caption{Left: empirical BLER as a function of the first-order Gauss--Markov correlation coefficient at nominal $(E_b/N_0)_{\mathrm{dB}}=2$, using $5000$ random-code ensemble frames for each value of $\rho$.  Right: observed BLER as a function of the valid membership-test budget on a logarithmic horizontal axis, computed by truncating the first-hit ranks from the verified paired $10^4$-frame experiment at nominal $(E_b/N_0)_{\mathrm{dB}}=2$ and $\rho=0.5$.
}
\label{fig:rho_qmax}
\end{figure*}

\subsection{Fixed-Code Results at Blocklength 64}
\label{ssec:fixed_code_results}

\begin{table}[!b]
\centering
\caption{Fixed-code results for $n=64$, $k=52$, $\rho=0.5$, nominal
$(E_b/N_0)_{\mathrm{dB}}=2$, and $q_{\max}=20{,}000$. Each codebook is
evaluated over $1000$ common frames. Each entry gives the empirical BLER,
its $95\%$ Wilson score interval, and the error count in parentheses.}
\vspace{-4pt}
\label{tab:fixed64}
\scriptsize
\setlength{\tabcolsep}{2.4pt}
\renewcommand{\arraystretch}{1.15}
\resizebox{\columnwidth}{!}{%
\begin{tabular}{lccc}
\toprule
Code
& LP-GRAND
& ExactBlockProduct, $b=8$
& ORBGRAND-AI, $b=8$
\\
\midrule
RLC-0
& $0.022\ [0.015,0.033]\ (22)$
& $0.040\ [0.030,0.054]\ (40)$
& $0.046\ [0.035,0.061]\ (46)$
\\
RLC-1
& $0.028\ [0.019,0.040]\ (28)$
& $0.036\ [0.026,0.049]\ (36)$
& $0.046\ [0.035,0.061]\ (46)$
\\
RLC-2
& $0.027\ [0.019,0.039]\ (27)$
& $0.041\ [0.030,0.055]\ (41)$
& $0.058\ [0.045,0.074]\ (58)$
\\
RLC-3
& $0.016\ [0.010,0.026]\ (16)$
& $0.029\ [0.020,0.041]\ (29)$
& $0.040\ [0.030,0.054]\ (40)$
\\
RLC-4
& $0.024\ [0.016,0.035]\ (24)$
& $0.041\ [0.030,0.055]\ (41)$
& $0.059\ [0.046,0.075]\ (59)$
\\
CRC-12
& $0.026\ [0.018,0.038]\ (26)$
& $0.039\ [0.029,0.053]\ (39)$
& $0.054\ [0.042,0.070]\ (54)$
\\
\bottomrule
\end{tabular}%
}
\end{table}

We also evaluated five independently drawn systematic $[64,52]$ random linear codes (RLCs) and one systematic $[64,52]$ CRC-12 code. Each codebook was held fixed over $1000$ frames, and all methods used the same transmitted message and received vector within each frame. The CRC encoder uses the generator polynomial $g(x)=x^{12}+x^{11}+x^3+x^2+x+1$, MSB-first polynomial division, an all-zero initial register, and no final XOR. The operating point was nominal $(E_b/N_0)_{\mathrm{dB}}=2$, with $\rho=0.5$ and $q_{\max}=20{,}000$.

Table~\ref{tab:fixed64} gives the results for the six codebooks. For each codebook, LP-GRAND had the lowest empirical BLER, followed by ExactBlockProduct with $b=8$ and our ORBGRAND-AI implementation with $b=8$. Pooling the five RLC experiments gives error counts of $117/5000$, $187/5000$, and $249/5000$, respectively. The corresponding sample means of the valid codebook-membership-test counts are $102.4$, $153.7$, and $218.6$. These totals are based on the five RLCs tested here. LP-GRAND's agreement with exhaustive ML decoding for the two fixed $[20,12]$ codes is reported earlier in Section~\ref{ssec:direct_ml_and_exact_generators}.

\vspace{10pt}
Additional random-code-ensemble, ordering-robustness, and calibration results are provided in Appendix~\ref{sec:supp_additional_numerical}.

\section{Conclusion}
\label{sec:conclusion}

For BPSK transmission over Gaussian noise with a specified symmetric positive-definite precision matrix, LP-GRAND expresses the candidate-dependent matched negative log-likelihood as an observation-dependent quadratic pseudo-Boolean energy. Using a validated path decomposition, suffix dynamic programming, and best-first complete-path enumeration, it enumerates noise-effect patterns in nondecreasing energy. A path decomposition of width $w$ yields at most $2^{w+1}$ bag assignments per layer and therefore bounds the size of the layered graphical representation, but not the number of membership queries required to obtain the first codebook hit. Under complete enumeration, deterministic tie resolution, and no abandonment, the first codebook hit induces an ML codeword for any nonempty binary codebook with equiprobable codewords. If the specified precision matrix is mismatched or sparsified, the enumeration remains exact for the resulting modified metric, but a first codebook hit need not be ML under the true channel model. A finite query budget may additionally cause abandonment before any codeword is found. In complete-order validation using the stored coefficients, the enumerated order agreed with exhaustive sorting in all tested cases. For two fixed $[20,12]$ codes, LP-GRAND agreed with exhaustive matched-metric minimization over the codebook in all $10{,}000$ tested frames. For each of six fixed $[64,52]$ codebooks at nominal $E_b/N_0=2$ dB, its empirical BLER was lower than that of each of the two tested block-based approximations.

% ----------------------------
%       references
% ----------------------------
% \bibliographystyle{IEEEtran}
% \bibliography{references}

\putbib
\end{bibunit}

\clearpage
% ---------------------------
%     Appendix Sections
% --------------------------
\appendices

% \AppendixOnlyTOC

\begin{bibunit}

\renewcommand{\thesubsection}{\thesection.\arabic{subsection}}
\renewcommand{\thesubsectiondis}{\thesection.\arabic{subsection}.}

\section{Additional Gaussian Noise-Effect Metric Results}
\label{ssec:gm_supplement_start}

\subsection{Stationary First-Order Gauss--Markov Noise}
\label{ssec:gauss_markov}

Consider the stationary first-order Gauss--Markov covariance matrix
\begin{equation}
\Sigma_{ij} = \sigma^2\rho^{|i-j|},
\qquad 1\leq i,j\leq n,
\label{eq:gm_covariance}
\end{equation}
where $\sigma^2>0$ is the marginal noise variance and $|\rho|<1$. For $n\geq2$,  the associated precision matrix is
\begin{equation}
Q = \frac{1}{\sigma^2(1-\rho^2)}
\begin{bmatrix}
1 & -\rho & 0 & \cdots & 0\\
-\rho & 1+\rho^2 & -\rho & \ddots & \vdots\\
0 & -\rho & 1+\rho^2 & \ddots & 0\\
\vdots & \ddots & \ddots & \ddots & -\rho\\
0 & \cdots & 0 & -\rho & 1
\end{bmatrix}.
\label{eq:gm_precision}
\end{equation}
Hence, $Q$ is tridiagonal. If $\rho\neq0$, then $\mathcal{G}_Q$ is the path graph with edge set $\mathcal{E}_Q = \bigl\{\{i,i+1\}\mid i=1,\ldots,n-1\bigr\}$.  Accordingly, \eqref{eq:quadratic_metric} reduces to the nearest-neighbor form
\begin{equation}
W_r(z) = \sum_{i=1}^n \alpha_i z_i + \sum_{i=1}^{n-1}\beta_i z_i z_{i+1},
\label{eq:chain_metric}
\end{equation}
with $\alpha_i$ as in~\eqref{eq:alpha} and
\begin{equation}
\begin{aligned}
\beta_i &\triangleq \beta_{i,i+1} = 4s_i s_{i+1}Q_{i,i+1}\\
&= -\frac{4\rho}{\sigma^2(1-\rho^2)} s_i s_{i+1}, \qquad i=1,\ldots,n-1.
\end{aligned}
\label{eq:gm_beta}
\end{equation}

\begin{remark}[Memoryless Gaussian special case]
\label{rem:memoryless_case}
When $\rho=0$, \eqref{eq:gm_covariance} reduces to $\Sigma=\sigma^2I_n$, corresponding to independent and identically distributed Gaussian noise, and $Q=\sigma^{-2}I_n$. Therefore, $\beta_i=0$ for all $i$, and $\alpha_i = 2s_i(Qa)_i+2Q_{ii} = 2|r_i|/ \sigma^2$, where the final equality follows from $a=r-s$ and $s_i r_i=|r_i|$. Accordingly,
\begin{equation}
W_r(z) = \frac{2}{\sigma^2} \sum_{i=1}^n |r_i|z_i,
\label{eq:memoryless_metric}
\end{equation}
which is coordinate-additive in $z$. SGRAND orders candidates according to this exact reliability-magnitude-weighted metric, whereas ORBGRAND replaces the reliability magnitudes with an integer-valued rank-domain surrogate~\cite{solomon2020soft,duffy2022ordered}.
\end{remark}

\subsection{Exactness Under Precision-Matrix Modification}
\label{ssec:exact_contract}

The candidate ordering and partition-function computation are exact with respect to the symmetric positive-definite precision matrix used to define the quadratic metric. When the original matrix $Q$ is used without modification, every nonzero off-diagonal entry of $Q$ is retained in the interaction graph; no thresholding is applied to determine its off-diagonal sparsity pattern. All pseudo-Boolean coefficients and local costs in the graphical representation, and hence the partition function, are derived from the same matrix $Q$. For a banded construction of half-bandwidth $\nu$, exactness with respect to $Q$ additionally requires $Q_{ij}=0$ for all $i,j$ satisfying $|i-j|>\nu$. If this condition does not hold, the banded representation omits out-of-band interaction terms and therefore defines a different quadratic metric.

More generally, modifications of $Q$ such as sparsification or diagonal
loading change the reference metric underlying the exactness guarantee. When $Q$ is deliberately sparsified, let $\mathcal{S}$ denote the specified sparsification operator and let $\lambda\geq0$ denote a diagonal-loading parameter. Define
\begin{equation}
\widetilde Q \triangleq \mathcal{S}(Q)+\lambda I_n,
\label{eq:sparsified_precision}
\end{equation}
and assume that $\widetilde Q=\widetilde Q^{\mathsf T}\succ0$. All pseudo-Boolean coefficients and local costs of the graphical representation are then recomputed from $\widetilde Q$, and the partition function is evaluated using the resulting energy. The candidate enumeration and partition-function computation are therefore exact for the metric induced by $\widetilde Q$, but do not generally reproduce their counterparts for the metric induced by $Q$. We quantify the modification by the relative spectral-norm perturbation $\delta_Q \triangleq (\lVert Q-\widetilde Q\rVert_2) / \lVert Q\rVert_2$.

Let $\widehat Q$ be a possibly nonsymmetric numerical approximation to the precision matrix $Q$, and define its symmetric part as $\widehat Q_{\mathrm{sym}} \triangleq (\widehat Q+\widehat Q^{\mathsf T})/2$. For every $v\in\mathbb{R}^n$, $v^{\mathsf T}\widehat Qv = v^{\mathsf T}\widehat Q_{\mathrm{sym}}v$,  because the skew-symmetric part of $\widehat Q$ contributes zero to the quadratic form. Thus, symmetrization does not change the quadratic metric associated with $\widehat Q$. Interpreting $\widehat Q_{\mathrm{sym}}$ as a valid Gaussian precision matrix, and applying constructions that require such a matrix, additionally requires $\widehat Q_{\mathrm{sym}}\succ0$. By contrast, sparsification and diagonal loading generally change the quadratic metric. All exactness statements in our work are formulated in real arithmetic; finite-precision implementation and validation are considered separately.

\begin{proposition}[Banded Cholesky whitening representation]
\label{prop:whitening_relation}

Let $Q=Q^{\mathsf T}\succ0$, and let $U$ be its upper-triangular Cholesky factor with positive diagonal entries, such that
$Q=U^{\mathsf T}U$. Then, for every $x\in\mathbb{R}^n$,
\begin{equation}
\frac{1}{2}(r-x)^{\mathsf T}Q(r-x)
= \frac{1}{2}\lVert U(r-x)\rVert_2^2
= \frac{1}{2}\lVert Ur-Ux\rVert_2^2.
\label{eq:whitened_metric}
\end{equation}
If $N\sim\mathcal{N}(0,Q^{-1})$ and $V\triangleq UN$, then
$V\sim\mathcal{N}(0,I_n)$. Therefore, defining
$\overline R\triangleq UR$ gives
\begin{equation}
\overline R=Ux(C)+V.
\label{eq:whitened_channel}
\end{equation}
If $Q$ has half-bandwidth at most $\nu$ under the coordinate ordering $1,\ldots,n$, then its Cholesky factor $U$ has upper bandwidth at most $\nu$, that is, $U_{ij}=0$ whenever $j-i>\nu$. Consequently,
\begin{equation}
(Ux)_i = \sum_{j=i}^{\min\{i+\nu,n\}}U_{ij}x_j,
\qquad i=1,\ldots,n.
\label{eq:banded_whitening_action}
\end{equation}
Thus, each output coordinate of the whitening transform depends on at most $\nu+1$ consecutive input coordinates. Bandedness alone does not imply that $U$ is Toeplitz or shift invariant.
\end{proposition}

\begin{proof}
See Appendix~\ref{app:proof_prop_whitening_relation}.
\end{proof}

\begin{remark}[Finite-block metric equivalence]
\label{rem:finite_block_metric_equivalence}

For a fixed observation, two candidate metrics $E_1$ and $E_2$ are order-equivalent if
\begin{equation}
E_2(x)=\eta E_1(x)+\kappa,
\qquad \eta>0,
\end{equation}
for every $x$ in the candidate set, where $\kappa$ is independent of $x$. Order-equivalent metrics induce the same strict pairwise orderings, tie classes, and sets of minimizing candidates.
Consider the finite-block linear Gaussian model $\overline R =Hx(C)+V$, $V\sim\mathcal N(\mathbf 0_m,\sigma_v^2I_m)$, where $H\in\mathbb R^{m\times n}$ has full column rank. Define
\begin{equation}
Q_H \triangleq \frac{H^{\mathsf T}H}{\sigma_v^2},
\qquad r_H \triangleq Q_H^{-1} \frac{H^{\mathsf T}\overline r}{\sigma_v^2}.
\end{equation}
Then, for every $x\in\{-1,+1\}^n$,
\begin{equation}
\frac{1}{2\sigma_v^2} \left\lVert \overline r-Hx \right\rVert_2^2
= \frac{1}{2} (r_H-x)^{\mathsf T} Q_H (r_H-x) + \kappa(\overline r),
\label{eq:linear_gaussian_metric_equivalence}
\end{equation}
where
\begin{equation}
\kappa(\overline r) \triangleq \frac{1}{2\sigma_v^2}
\lVert\overline r\rVert_2^2 - \frac{1}{2}r_H^{\mathsf T}Q_Hr_H
\end{equation}
is independent of $x$. Hence, the finite-block linear-channel negative log-likelihood metric and the precision-domain metric induce the same candidate-word order.

For a specified finite-block ISI matrix $H$, LP-GRAND and an ISI GRAND decoder that enumerates candidates in exact likelihood order therefore produce the same candidate-word order when they use the same channel matrix, noise variance, and boundary convention. This equivalence holds at the metric and candidate-order levels; their candidate-generation architectures and any auxiliary reference sequences may differ.

A stationary finite-block AR(1) precision matrix is not generally equal to $H^{\mathsf T}H/\sigma_v^2$ when $H$ is a finite-block Toeplitz convolution matrix associated with a two-tap channel, because the endpoint terms depend on the adopted finite-block boundary convention. Equal bandwidth alone therefore does not establish finite-block metric equivalence.
\end{remark}

\vspace{20pt}

% =======================================
\section{Stability of the Candidate Ordering}
\label{sec:robustness}

\subsection{Coefficient Quantization}

Let $\Delta>0$, and let $\mathcal{E}_Q$ denote the edge set of the interaction graph of $W_r$. Define the integer-valued coefficients
\begin{equation}
\widetilde{\alpha}_i \triangleq
\left\lfloor \frac{\alpha_i}{\Delta} \right\rceil,
\qquad i=1,\ldots,n,
\label{eq:quantized_alpha}
\end{equation}
and
\begin{equation}
\widetilde{\beta}_{ij} \triangleq
\left\lfloor \frac{\beta_{ij}}{\Delta} \right\rceil,
\qquad \{i,j\}\in\mathcal{E}_Q,
\label{eq:quantized_beta}
\end{equation}
where $\lfloor\cdot\rceil$ denotes rounding to the nearest integer, with a fixed deterministic convention for half-integer ties. The resulting integer-valued metric is
\begin{equation}
\widetilde{W}_{r,\Delta}(z) \triangleq
\sum_{i=1}^n \widetilde{\alpha}_i z_i
+ \sum_{\{i,j\}\in\mathcal{E}_Q} \widetilde{\beta}_{ij}z_i z_j.
\label{eq:integer_quantized_metric}
\end{equation}
Define the corresponding rescaled metric by $W_{r,\Delta}^{\mathrm q}(z) \triangleq \Delta\widetilde{W}_{r,\Delta}(z)$. Since $\Delta>0$, $\widetilde{W}_{r,\Delta}$ and $W_{r,\Delta}^{\mathrm q}$ induce identical strict candidate comparisons and tie classes.

% suppose that $Q$ has half-bandwidth at most $\nu$, where $0\leq\nu\leq n-1$?

Define the total number of unary and pairwise coefficients by $M_Q\triangleq n+|\mathcal{E}_Q|$. Suppose that $Q$ has half-bandwidth at most $\nu$, where $0\leq\nu\leq n-1$, and define
\begin{equation}
M_\nu(n) \triangleq n+\sum_{d=1}^{\nu}(n-d)
= n(\nu+1)-\frac{\nu(\nu+1)}{2}.
\label{eq:banded_coefficient_count}
\end{equation}
Since $Q_{ij}=0$ whenever $|i-j|>\nu$, $|\mathcal{E}_Q| \leq \sum_{d=1}^{\nu}(n-d)$, and hence
\begin{equation}
M_Q\leq M_\nu(n).
\label{eq:coefficient_count_bound}
\end{equation}

\begin{theorem}[Uniform metric-error bound under coefficient quantization]
\label{thm:quant}
For every $z\in\{0,1\}^n$,
\begin{equation}
\left|W_r(z)-W_{r,\Delta}^{\mathrm q}(z)\right|
\leq \frac{\Delta}{2}M_Q.
\label{eq:uniform_quantization_bound_general}
\end{equation}
If $Q$ has half-bandwidth at most $\nu$, then $M_Q\leq M_\nu(n)$ and hence
\begin{equation}
\left|W_r(z)-W_{r,\Delta}^{\mathrm q}(z)\right|
\leq \frac{\Delta}{2}M_\nu(n).
\label{eq:uniform_quantization_bound_banded}
\end{equation}
\end{theorem}

\begin{proof}
See Appendix~\ref{app:proof_thm_quant}.
\end{proof}

\begin{corollary}[Uniform sufficient condition for strict pairwise-order preservation]
\label{cor:quant_order}
Let $z,z'\in\{0,1\}^n$. If
\begin{equation}
\left| W_r(z)-W_r(z') \right| > \Delta M_Q,
\label{eq:uniform_order_condition}
\end{equation}
then
\begin{equation}
\operatorname{sgn} \bigl( W_r(z)-W_r(z') \bigr)
= \operatorname{sgn} \bigl( \widetilde{W}_{r,\Delta}(z)
- \widetilde{W}_{r,\Delta}(z') \bigr).
\label{eq:uniform_order_preservation}
\end{equation}
Hence, $z$ and $z'$ have the same strict order under $W_r$ and $\widetilde{W}_{r,\Delta}$. The bandwidth-based condition $\left| W_r(z)-W_r(z') \right| > \Delta M_\nu(n)$ is also sufficient.
\end{corollary}

\begin{proof}
See Appendix~\ref{app:proof_cor_quant_order}.
\end{proof}

\begin{corollary}[Pair-specific quantization bound]
\label{cor:pair_specific_quantization}
For $z,z'\in\{0,1\}^n$, define
\begin{equation}
D_Q(z,z') \triangleq \sum_{i=1}^n |z_i-z'_i|
+ \sum_{\{i,j\}\in\mathcal{E}_Q}
\left| z_i z_j-z'_i z'_j \right|.
\label{eq:pair_specific_difference_count}
\end{equation}
Then
\begin{equation}
\left| \left[ W_{r,\Delta}^{\mathrm q}(z) - W_{r,\Delta}^{\mathrm q}(z') \right]
- \left[ W_r(z)-W_r(z') \right] \right|
\leq \frac{\Delta}{2}D_Q(z,z').
\label{eq:pair_specific_gap_bound}
\end{equation}
Consequently, $z$ and $z'$ have the same strict order under $W_r$ and $W_{r,\Delta}^{\mathrm q}$ if
\begin{equation}
\left| W_r(z)-W_r(z') \right|
> \frac{\Delta}{2}D_Q(z,z').
\label{eq:pair_specific_certificate}
\end{equation}
\end{corollary}

\begin{proof}
See Appendix~\ref{app:proof_cor_pair_specific_quantization}.
\end{proof}

The quantity $D_Q(z,z')$ counts only the unary and pairwise monomials whose values differ between $z$ and $z'$. Consequently, $D_Q(z,z')\leq M_Q$, and the pair-specific bound is never weaker than the corresponding uniform bound.

\begin{remark}
\label{rem:quantized_metric_interpretation}
In our implementation, we use rounding to nearest, with ties to even. Because the unary and pairwise coefficients are quantized independently, the resulting coefficients need not satisfy the relations in \eqref{eq:alpha} and~\eqref{eq:beta} for any single symmetric positive-definite precision matrix and received observation. Accordingly, $W_{r,\Delta}^{\mathrm q}$ is regarded as a quadratic pseudo-Boolean metric rather than necessarily as a Gaussian negative log-likelihood. Enumeration is exact with respect to $W_{r,\Delta}^{\mathrm q}$, although the candidate order it induces may differ from that induced by $W_r$.
\end{remark}

\begin{corollary}[First-order Gauss--Markov case]
\label{cor:gm_quantization}
Consider the first-order Gauss--Markov model in~\eqref{eq:gm_precision}, with $n\geq2$. If $\rho\neq0$, then
\begin{equation}
\mathcal{E}_Q = \bigl\{ \{i,i+1\}:i=1,\ldots,n-1 \bigr\},
\end{equation}
and hence $M_Q=2n-1$. Therefore, by Theorem~\ref{thm:quant}, for every $z\in\{0,1\}^n$,
\begin{equation}
\left| W_r(z)-W_{r,\Delta}^{\mathrm q}(z) \right|
\leq \frac{\Delta}{2}(2n-1).
\end{equation}
If $\rho=0$, then $\mathcal{E}_Q=\varnothing$ and $M_Q=n$, which yields
\begin{equation}
\left| W_r(z)-W_{r,\Delta}^{\mathrm q}(z) \right|
\leq \frac{\Delta n}{2},
\qquad z\in\{0,1\}^n.
\end{equation}
\end{corollary}

\subsection{Precision-Matrix Mismatch}

Suppose that the Gaussian noise covariance is $\Sigma=\Sigma^{\mathsf T}\succ0$, with precision matrix $Q=\Sigma^{-1}$, whereas the decoder uses a symmetric positive-definite matrix $\widehat{Q}=\widehat{Q}^{\mathsf T}\succ0$. Define the resulting mismatched decoding metric by
\begin{equation}
\widehat{E}_r(z) \triangleq \frac{1}{2} (r-x_z)^{\mathsf T} \widehat{Q} \, (r-x_z).
\label{eq:mismatched_energy}
\end{equation}

\begin{proposition}[Uniform metric-error and order-stability bounds under precision-matrix mismatch]
\label{prop:precision_mismatch}
Suppose that $\lVert Q-\widehat{Q}\rVert_2 \leq \varepsilon$ for some $\varepsilon\geq0$, where $\lVert\cdot\rVert_2$ denotes the spectral norm. Then, for every $z\in\{0,1\}^n$,
\begin{equation}
\left| E_r(z)-\widehat{E}_r(z) \right|
\leq \frac{\varepsilon}{2} \left( \lVert r\rVert_2+\sqrt{n} \right)^2.
\label{eq:uniform_energy_mismatch_bound}
\end{equation}
Furthermore, if $z,z'\in\{0,1\}^n$ satisfy
\begin{equation}
\label{eq:uniform_mismatch_order_condition}
\left| E_r(z)-E_r(z') \right|> \varepsilon\left(\lVert r\rVert_2+\sqrt{n}\right)^2,
\end{equation}
then
\begin{equation}
\operatorname{sgn} \bigl( E_r(z)-E_r(z') \bigr) =
\operatorname{sgn} \bigl( \widehat{E}_r(z)-\widehat{E}_r(z') \bigr).
\label{eq:mismatch_order_preservation}
\end{equation}
Hence, $z$ and $z'$ have the same strict order under $E_r$ and $\widehat{E}_r$.
\end{proposition}

\begin{proof}
See Appendix~\ref{app:proof_prop_precision_mismatch}.
\end{proof}

\begin{proposition}[Pair-specific order stability under precision-matrix mismatch]
\label{prop:pair_specific_precision_mismatch}
Let $\Delta_Q \triangleq \widehat{Q}-Q$. For any $z,z'\in\{0,1\}^n$,
\begin{align}
&\left| \left[ \widehat{E}_r(z)-\widehat{E}_r(z') \right]
- \left[ E_r(z)-E_r(z') \right] \right| \nonumber\\
&\qquad\leq
\frac{\lVert\Delta_Q\rVert_2}{2} \lVert x_z-x_{z'}\rVert_2 \lVert x_z+x_{z'}-2r \rVert_2.
\label{eq:pair_specific_precision_bound}
\end{align}
Consequently, $z$ and $z'$ have the same strict order under $E_r$ and $\widehat{E}_r$ if
\begin{equation}
\left| E_r(z)-E_r(z') \right| > \frac{\lVert\Delta_Q\rVert_2}{2} \lVert x_z-x_{z'}\rVert_2 \lVert x_z+x_{z'}-2r\rVert_2.
\label{eq:pair_specific_precision_certificate}
\end{equation}
Furthermore,
\begin{equation}
\lVert x_z-x_{z'}\rVert_2 = 2\sqrt{d_{\mathrm H}(z,z')},
\label{eq:bpsk_hamming_distance}
\end{equation}
where $d_{\mathrm H}$ denotes the Hamming distance. Hence, the right-hand side of~\eqref{eq:pair_specific_precision_bound} can be written as $\lVert\Delta_Q\rVert_2 \sqrt{d_{\mathrm H}(z,z')} \lVert x_z+x_{z'}-2r\rVert_2$.
\end{proposition}

\begin{proof}
See Appendix~\ref{app:proof_prop_pair_specific_precision_mismatch}.
\end{proof}

\begin{remark}[Comparison of pair-specific and uniform energy-gap perturbation bounds]
\label{rem:mismatch_bound_comparison}
The pair-specific bound in \eqref{eq:pair_specific_precision_bound} does not exceed the uniform energy-gap perturbation bound in~\eqref{eq:pairwise_gap_mismatch}. With $a\triangleq r-x_z$ and $b\triangleq r-x_{z'}$,
\begin{align}
\frac{\lVert\Delta_Q\rVert_2}{2} \lVert a-b\rVert_2 \lVert a+b\rVert_2
&\leq \frac{\lVert\Delta_Q\rVert_2}{4}
\left( \lVert a-b\rVert_2^2 + \lVert a+b\rVert_2^2 \right)
\nonumber\\
&= \frac{\lVert\Delta_Q\rVert_2}{2} \left( \lVert a\rVert_2^2 +
\lVert b\rVert_2^2 \right)
\nonumber\\
&\leq \lVert\Delta_Q\rVert_2 \left( \lVert r\rVert_2+\sqrt{n} \right)^2,
\end{align}
where the first inequality follows from $2uv\leq u^2+v^2$, and the equality follows from the parallelogram identity. Under the assumption $\lVert Q-\widehat{Q}\rVert_2\leq\varepsilon$, the energy-gap perturbation is therefore bounded by $\varepsilon\bigl(\lVert r\rVert_2+\sqrt{n}\bigr)^2$. Consequently, if $\left|E_r(z)-E_r(z')\right| > \varepsilon\bigl(\lVert r\rVert_2+\sqrt{n}\bigr)^2$, the matched and mismatched energy gaps have the same sign.
\end{remark}

\vspace{20pt}

% ======================================================================
\section{Full-Space Weights, Auxiliary First-Hit Correctness, and Abandonment}
\label{sec:softout}
% ======================================================================

The trellis used for likelihood-ordered enumeration also supports sum-product computation over the full assignment space $\{0,1\}^n$. Define
\begin{equation}
Z_r \triangleq \sum_{u\in\{0,1\}^n} \exp\bigl\{-W_r(u)\bigr\}
\label{eq:fullspace_partition}
\end{equation}
and
\begin{equation}
p_r(z) \triangleq \frac{\exp\bigl\{-W_r(z)\bigr\}}{Z_r},
\qquad z\in\{0,1\}^n.
\label{eq:fullspace_weight}
\end{equation}
Then $p_r$ is a probability mass function on $\{0,1\}^n$ and satisfies $\sum_{z\in\{0,1\}^n}p_r(z)=1$. LP-GRAND enumerates the noise-effect patterns in nondecreasing order of $W_r(z)$, equivalently in nonincreasing order of $p_r(z)$, subject to the same fixed deterministic tie-breaking rule. The mass $p_r(z)$ is associated with the pattern $z$, rather than with its rank in the enumeration.

Since $E_r(z)=E_r(0)+W_r(z)$, the candidate-independent term $E_r(0)$ cancels under normalization. Thus, $p_r$ is the normalized Gaussian likelihood mass over all $2^n$ BPSK candidate words. More precisely, for fixed $r$ and $y=y(r)$, the mapping $z\mapsto y\xor z$ is a bijection on $\{0,1\}^n$, and $p_r(z)$ equals the posterior mass assigned to the candidate word $y\xor z$ under a uniform prior on all binary words.

This full-space distribution differs from the posterior under the uniform codebook prior. If $C$ is uniform on $\calC$, then, for $z\in\{0,1\}^n$,
\begin{equation}
\!\!\!
\Pr\!\left\{C=y\xor z\mid R=r\right\} =
\frac{ \one\{y\xor z\in\calC\} \exp\bigl\{-W_r(z)\bigr\}}{\displaystyle
\sum_{\substack{u\in\{0,1\}^n\\ y\xor u\in\calC}} \exp\bigl\{-W_r(u)\bigr\}}.
\label{eq:codebook_conditioned_posterior}
\end{equation}
Therefore, $p_r$ generally differs from the codebook-conditioned posterior in~\eqref{eq:codebook_conditioned_posterior}.

\subsection{Partition Function and Full-Space Weights}

Under the zero-padded state convention introduced in Section~\ref{sec:graphical}, let the trellis state space be
\begin{equation}
\mathcal{U}_\nu \triangleq
\begin{cases}
\{0,1\}^{\nu}, & \nu\geq1,\\
\{\varnothing\}, & \nu=0.
\end{cases}
\label{eq:trellis_state_space_partition}
\end{equation}
Let $F_i(v)$ denote the total unnormalized weight of all partial paths
that have processed coordinates $1,\ldots,i$ and terminate in state
$v$. Initialize
\begin{equation}
F_0(v) =
\begin{cases}
1, & v=u_1,\\
0, & v\neq u_1,
\end{cases}
\qquad v\in\mathcal{U}_\nu,
\label{eq:forward_initialization}
\end{equation}
where $u_1$ is the initial state defined in Section~\ref{sec:graphical}. For $i=1,\ldots,n$, apply the forward sum-product recursion
\begin{equation}
F_i(v) = \sum_{u\in\mathcal{U}_\nu}
\sum_{\substack{b\in\{0,1\}\\ T(u,b)=v}}
F_{i-1}(u) \exp\bigl\{-\gamma_i(u,b)\bigr\},
\quad v\in\mathcal{U}_\nu.
\label{eq:forward_partition_recursion}
\end{equation}
Since each $z\in\{0,1\}^n$ corresponds to a unique trellis path whose branch metrics sum to $W_r(z)$, summing over all states after processing coordinate $n$ yields
\begin{equation}
Z_r = \sum_{v\in\mathcal{U}_\nu}F_n(v).
\label{eq:trellis_partition}
\end{equation}

\begin{proposition}[Exact full-space partition function and normalization]
\label{prop:fullspace_weights}
Suppose that $Q$ has half-bandwidth at most $\nu$, where $0\leq\nu\leq n-1$. In exact arithmetic, the recursion \eqref{eq:forward_initialization}--\eqref{eq:trellis_partition} computes the partition function in~\eqref{eq:fullspace_partition}. A full-state implementation that propagates the weights along both outgoing branches of every state at each layer performs exactly $2n2^\nu$ branch updates. An implementation restricted to reachable states performs no more than this number. The arithmetic complexity is therefore $\mathcal{O}(n2^\nu)$, and retaining only two consecutive forward layers requires $\mathcal{O}(2^\nu)$ working memory. For fixed $\nu$, both bounds are linear in $n$. For an emitted noise-effect pattern $z^{(q)}$, its normalized full-space weight is
\begin{equation}
p_r\bigl(z^{(q)}\bigr) = \frac{\exp\bigl\{-W_r(z^{(q)})\bigr\}}{Z_r}.
\label{eq:emitted_fullspace_weight}
\end{equation}
\end{proposition}

\begin{proof}
See Appendix~\ref{app:proof_prop_fullspace_weights}.
\end{proof}

\begin{remark}[Numerically stable log-domain evaluation]
\label{rem:log_domain_partition}
Our implementation evaluates the forward recursion and the cumulative full-space mass in the log domain. Let $\Lambda_i(v)\triangleq\log F_i(v)$, with the convention $\log 0=-\infty$. For a finite index set $\mathcal J$, if $\mathcal J=\varnothing$ or $a_j=-\infty$ for every $j\in\mathcal J$, define $\log\left(\sum_{j\in\mathcal J}e^{a_j}\right)=-\infty$. Otherwise, letting $m\triangleq\max_{j\in\mathcal J}a_j$, evaluate the logarithm as
\begin{equation}
\log\left(\sum_{j\in\mathcal J}e^{a_j}\right) = m+\log\left(\sum_{j\in\mathcal J}e^{a_j-m}\right).
\label{eq:stable_log_sum}
\end{equation}
which is applied to the forward recursion in~\eqref{eq:forward_partition_recursion} and to
\begin{equation}
\log Z_r = \log\left( \sum_{v\in\mathcal{U}_\nu}e^{\Lambda_n(v)} \right).
\label{eq:log_partition}
\end{equation}
For $C_q\triangleq\sum_{t=1}^{q}p_r\bigl(z^{(t)}\bigr)$, define $\ell_q \triangleq\log C_q$, with $\ell_0=-\infty$. After processing the $q$-th emitted pattern,
\begin{equation}
\ell_q = \log\left[ e^{\ell_{q-1}} + e^{-W_r(z^{(q)})-\log Z_r} \right],
\label{eq:log_covered_mass}
\end{equation}
where the two-term instance of~\eqref{eq:stable_log_sum} is used. We computed the  remaining full-space mass as
\begin{equation}
1-C_q = -\operatorname{expm1}(\ell_q),
\label{eq:stable_tail_evaluation}
\end{equation}
where $\operatorname{expm1}(x)\triangleq e^x-1$. This form avoids loss of precision from direct subtraction when $C_q$ is close to one. Before evaluating~\eqref{eq:stable_tail_evaluation}, in our implementation, we replace $\ell_q$ by $\min\{0,\ell_q\}$ to enforce the theoretical constraint $\ell_q\leq0$ under floating-point roundoff.
\end{remark}

\subsection{First-Hit Correctness Under an Auxiliary Random-Codebook Occupancy Model}

We apply the SOGRAND random-codebook occupancy argument to the first codebook hit, using the full-space weights in~\eqref{eq:fullspace_weight}~\cite{yuan2025soft,duffy2024soft}. Fix $r$ and condition throughout on $R=r$. Let $L\triangleq2^n$, and let $z^{(1)},\ldots,z^{(L)}$ denote the LP-GRAND query sequence ordered by nondecreasing $W_r$, with the fixed deterministic tie-breaking rule used by the enumerator. Define $c^{(q)} \triangleq y(r)\xor z^{(q)}$, $q=1,\ldots,L$. The sequence $c^{(1)},\ldots,c^{(L)}$ is therefore a permutation of $\{0,1\}^n$. Define also $w_q \triangleq p_r\bigl(z^{(q)}\bigr)$, $q=1,\ldots,L$. Consider an auxiliary random-codebook ensemble in which the codebook $\calC$ is selected uniformly from all $M$-element subsets of $\{0,1\}^n$, where $1\leq M\leq L$, and the transmitted codeword $C$ is then selected uniformly from $\calC$. Under this ensemble, the marginal prior on $C$ is uniform over $\{0,1\}^n$. Let $J\in\{1,\ldots,L\}$ be the unique query position satisfying $C=c^{(J)}$. It follows that
\begin{equation}
\Pr \left\{J=q\mid R=r\right\} = w_q,
\qquad q=1,\ldots,L.
\label{eq:occupancy_query_position_distribution}
\end{equation}
Conditional on $J=j$, define the set of query positions occupied by the remaining codewords as $\mathcal I \triangleq \left\{ q\in\{1,\ldots,L\}\setminus\{j\}: c^{(q)}\in\calC \right\}$. For every $I\subseteq\{1,\ldots,L\}\setminus\{j\}$ with $|I|=M-1$,
\begin{equation}
\Pr \left\{ \mathcal I=I \mid J=j,R=r \right\}
= \binom{L-1}{M-1}^{-1}.
\label{eq:uniform_random_codebook_occupancy}
\end{equation}
Thus, conditional on $J$, the positions of the remaining $M-1$ codewords are sampled uniformly without replacement from the other $L-1$ query positions.
Define the first codebook-hit index by $H \triangleq \min\bigl({J}\cup\mathcal I\bigr)$. When $M=1$, $\mathcal I=\varnothing$ and hence $H=J$. Finally, define the full-space tail mass after query $q$ by
\begin{equation}
T_q \triangleq \sum_{t=q+1}^{L}w_t
= \Pr \left\{J>q\mid R=r\right\},
\quad q=1,\ldots,L,
\label{eq:fullspace_query_tail}
\end{equation}
so that $T_L=0$.
For a fixed binary linear $[n,k]$ code, one sets $M=2^k$ when applying this auxiliary model. The uniform occupancy law in~\eqref{eq:uniform_random_codebook_occupancy} is exact for the random-codebook ensemble and is not implied solely by the linear structure of the fixed code.

\begin{theorem}[First-hit correctness probability under the auxiliary occupancy model]
\label{thm:single_app}
Under the auxiliary random-codebook occupancy model defined above, let $q\in\{1,\ldots,L-M+1\}$. For integers $a$ and $b$, adopt the convention $\binom{a}{b}=0$ unless $0\leq b\leq a$. Define
\begin{equation}
A_q \triangleq w_q\binom{L-q}{M-1},
\qquad B_q \triangleq T_q\binom{L-q-1}{M-2}.
\label{eq:first_hit_terms}
\end{equation}
Then the conditional probability that the first codebook hit is the transmitted codeword is
\begin{equation}
P_{\mathrm{cor}}(q\mid r) \triangleq
\Pr \left\{J=q\mid H=q,R=r\right\} = \frac{A_q}{A_q+B_q}.
\label{eq:first_hit_correctness}
\end{equation}
\end{theorem}

\begin{proof}
See Appendix~\ref{app:proof_thm_single_app}.
\end{proof}

For $M>1$ and $q\in\{1,\ldots,L-M+1\}$, \eqref{eq:first_hit_correctness} is equivalent to
\begin{equation}
P_{\mathrm{cor}}(q\mid r) = \frac{w_q(L-q)}{w_q(L-q)+(M-1)T_q}.
\label{eq:first_hit_correctness_equivalent}
\end{equation}
This expression is exact under the auxiliary random-codebook occupancy model. For a fixed structured codebook, the occupied query positions do not generally follow the uniform-subset occupancy law, and the full-space masses $w_q$ need not equal the corresponding codebook-conditioned query-position probabilities. Consequently, \eqref{eq:first_hit_correctness} yields an auxiliary-model estimate of the first-hit correctness probability for a fixed structured codebook.

\subsection{Full-Space Cumulative Mass and Auxiliary Posterior Tail}

For fixed $r$, let $z^{(1)}(r),\ldots,z^{(2^n)}(r)$ denote the LP-GRAND query sequence ordered according to $W_r$, including the fixed deterministic tie-breaking rule in~\eqref{eq:order}. For $\tau\in\{0,\ldots,2^n\}$, define the queried prefix by
\begin{equation}
\mathcal{A}_\tau(r) \triangleq \left\{ z^{(q)}(r):1\leq q\leq\tau \right\}.
\label{eq:queried_prefix}
\end{equation}
In particular, $\mathcal{A}_0(r)=\varnothing$. Its normalized full-space cumulative mass is
\begin{equation}
C_\tau(r) \triangleq \sum_{z\in\mathcal{A}_\tau(r)}p_r(z) = \sum_{q=1}^{\tau}
p_r\bigl(z^{(q)}(r)\bigr),
\label{eq:prefix_coverage}
\end{equation}
where the empty sum is defined as zero. Define the corresponding remaining full-space mass by
\begin{equation}
T_\tau(r) \triangleq 1-C_\tau(r) = \sum_{z\notin\mathcal{A}_\tau(r)}p_r(z).
\label{eq:fullspace_tail}
\end{equation}
Consequently, $C_0(r)=0$, $C_{2^n}(r)=1$, $T_0(r)=1$, $T_{2^n}(r)=0$. 
When $p_r$ is formed from the same metric used for enumeration, $C_\tau(r)$ and $T_\tau(r)$ are, respectively, the exact normalized metric masses of the queried prefix and its complement. For a fixed codebook, these quantities alone do not equal the codebook-conditioned posterior probability that the transmitted codeword lies in the queried prefix, nor do they determine the conditional decoder-correctness probability. Such probabilities also depend on the locations of the codewords within the query sequence.

An exact posterior interpretation is obtained under the following auxiliary full-space model. Let $C^{\mathrm{aux}} \sim \operatorname{Unif}\bigl(\{0,1\}^n\bigr)$, $R^{\mathrm{aux}} = x\bigl(C^{\mathrm{aux}}\bigr)+N$, where $N\sim\mathcal{N}(0,\Sigma)$, and define
\begin{equation}
Z^{\mathrm{aux}} \triangleq y\bigl(R^{\mathrm{aux}}\bigr)
\xor C^{\mathrm{aux}}.
\label{eq:auxiliary_noise_effect_pattern}
\end{equation}
For the matched Gaussian metric,
\begin{equation}
p_r(z) = \Pr \left\{ Z^{\mathrm{aux}}=z \mid R^{\mathrm{aux}}=r \right\}.
\label{eq:auxiliary_posterior_identity}
\end{equation}
Consequently,
\begin{equation}
T_\tau(r)= \Pr\left\{ Z^{\mathrm{aux}}\notin\mathcal{A}_\tau(r)
\mid R^{\mathrm{aux}}=r \right\}.
\label{eq:auxiliary_tail_probability}
\end{equation}
Thus, for any $\varepsilon\in[0,1]$, the condition $T_\tau(r)\leq\varepsilon$ guarantees, under the auxiliary full-space model, that the posterior probability assigned to noise-effect patterns outside the queried prefix is at most $\varepsilon$.

If the normalized weights are formed from a metric different from the matched Gaussian metric, they do not generally equal the posterior in~\eqref{eq:auxiliary_posterior_identity} under the true channel. They nevertheless retain their interpretation as normalized metric-induced masses. If the modified metric is induced by a symmetric positive-definite precision matrix, the weights may also be interpreted as posterior masses under the corresponding auxiliary Gaussian channel. An independently coefficient-quantized quadratic pseudo-Boolean metric need not admit such a Gaussian posterior interpretation.

\begin{proposition}[Excess error probability under abandonment]
\label{prop:abandonment_tail_loss}
Let $C$ denote the transmitted codeword, uniformly distributed over $\calC$, let $R$ denote the received vector, and define the realized noise-effect pattern by $Z^n\triangleq y(R)\xor C$. The full decoder enumerates all $2^n$ noise-effect patterns in the received-vector-dependent order $z^{(1)}(R),\ldots,z^{(2^n)}(R)$ and returns the first candidate word $y(R)\xor z^{(q)}(R)$ that belongs to $\calC$. Let $\tau:\mathbb{R}^n\to\{0,\ldots,2^n\}$ be measurable. The decoder with abandonment uses the same query order, tie-breaking rule, and membership test, but returns $\mathsf{ABANDON}$ if no codeword is found among the first $\tau(R)$ queries.

Let $P_{\mathrm e}^{\mathrm{full}}$ and $P_{\mathrm e}^{\mathrm{ab}}$ denote the corresponding error probabilities, with abandonment counted as an error. Then
\begin{equation}
0 \leq P_{\mathrm e}^{\mathrm{ab}}
- P_{\mathrm e}^{\mathrm{full}}
\leq \Pr \left\{ Z^n\notin\mathcal{A}_{\tau(R)}(R) \right\}.
\label{eq:general_abandonment_bound}
\end{equation}
For $P_R$-almost every $r$, define the codebook-conditioned posterior tail by
\begin{equation}
T_{\tau}^{\calC}(r) \triangleq \Pr \left\{ Z^n\notin\mathcal{A}_{\tau}(r)
\mid R=r \right\}.
\label{eq:codebook_conditioned_tail}
\end{equation}
Then
\begin{equation}
P_{\mathrm e}^{\mathrm{ab}} \leq P_{\mathrm e}^{\mathrm{full}}
+ \mathbb{E}\left[ T_{\tau(R)}^{\calC}(R) \right].
\label{eq:codebook_tail_abandonment_bound}
\end{equation}
If, in addition,
\begin{equation}
p_r(z) = \Pr \left\{ Z^n=z\mid R=r \right\},
\qquad z\in\{0,1\}^n,
\label{eq:matched_flip_posterior}
\end{equation}
for $P_R$-almost every $r$, then $T_{\tau}^{\calC}(r)=T_{\tau}(r)$ and
\begin{equation}
P_{\mathrm e}^{\mathrm{ab}} \leq P_{\mathrm e}^{\mathrm{full}}
+ \mathbb{E} \left[ T_{\tau(R)}(R) \right].
\label{eq:posterior_tail_abandonment_bound}
\end{equation}
In particular, if $T_{\tau(r)}(r)\leq\varepsilon$ for $P_R$-almost every $r$, then
\begin{equation}
P_{\mathrm e}^{\mathrm{ab}} \leq P_{\mathrm e}^{\mathrm{full}}
+ \varepsilon.
\label{eq:uniform_tail_abandonment_bound}
\end{equation}
\end{proposition}

\begin{proof}
See Appendix~\ref{app:proof_prop_abandonment_tail_loss}.
\end{proof}

\begin{remark}[Scope of the posterior-tail bound]
\label{rem:abandonment_bound_scope}
The event bound in~\eqref{eq:general_abandonment_bound} holds for any codebook and any measurable query limit that depends on the received vector. Identifying the codebook-conditioned tail probability with $T_{\tau(r)}(r)$ requires~\eqref{eq:matched_flip_posterior} to hold for $P_R$-almost every $r$. In particular, if $p_r$ is only a normalized full-space metric mass and does not equal the conditional distribution of the realized noise-effect pattern under the actual codebook model, then the bound in~\eqref{eq:posterior_tail_abandonment_bound} does not generally hold.
\end{remark}

\vspace{20pt}

% ======================================================================
\section{Additional Numerical Evaluation}
\label{sec:supp_additional_numerical}

\subsection{Random-Code Ensemble Performance and Paired Comparisons}
\label{ssec:random_code_results}

Figure~\ref{fig:main_results} summarizes results from $10^4$ simulated frames at each nominal SNR. A new systematic $[64,52]$ random linear code is generated for every frame. The reported results are therefore averaged over the specified random-code ensemble rather than over repeated transmissions using a fixed code. Among the methods shown, LP-GRAND had the lowest empirical BLER, mean number of valid codebook-membership tests, and sample $99$th percentile of the number of valid membership tests at every simulated SNR.

\begin{figure*}[t]
\centering
\begin{minipage}{0.32\linewidth}
\centering
\includegraphics[width=\linewidth]{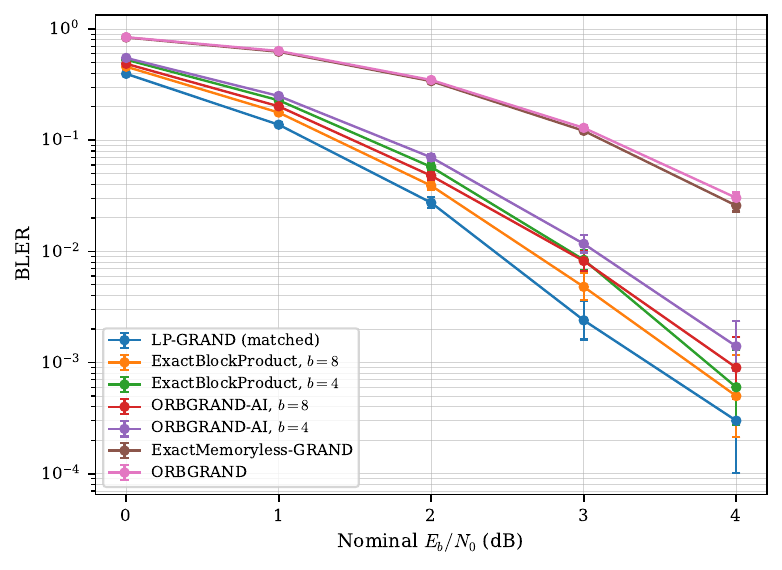}
\end{minipage}\hfill
\begin{minipage}{0.32\linewidth}
\centering
\includegraphics[width=\linewidth]{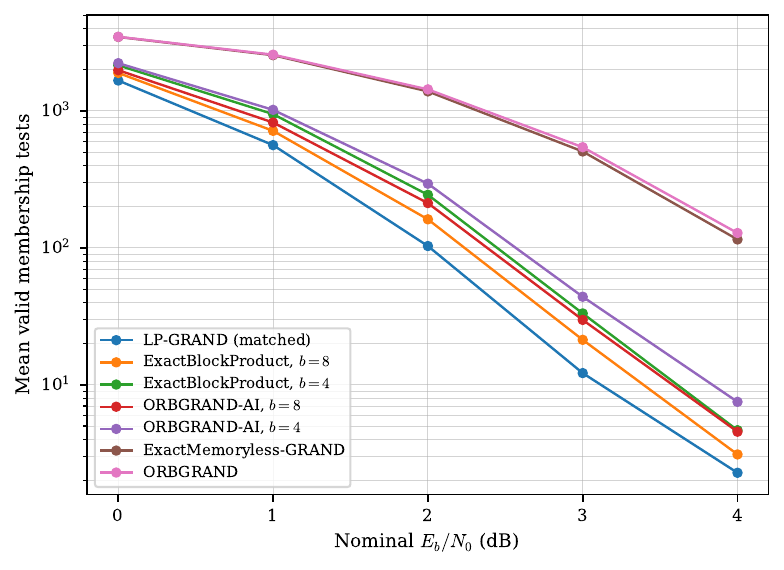}
\end{minipage}\hfill
\begin{minipage}{0.32\linewidth}
\centering
\includegraphics[width=\linewidth]{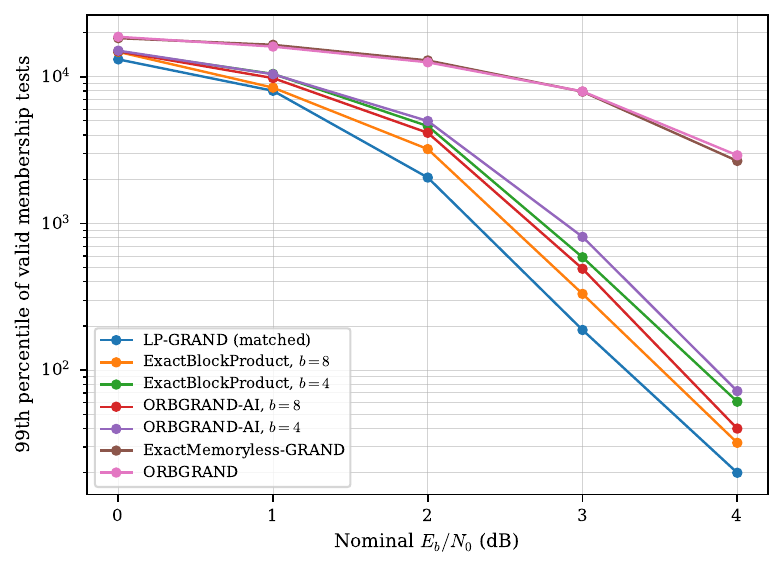}
\end{minipage}
\caption{Performance over the systematic $[64,52]$ random-linear-code ensemble with $\rho=0.5$, $q_{\max}=20{,}000$, and $10^4$ frames per nominal SNR. Left: BLER with $95\%$ Wilson score intervals. Center: sample mean number of valid codebook-membership tests. Right: sample $99$th percentile of the number of valid membership tests. Sample percentiles are computed by linear interpolation between adjacent order statistics. 
%ExactBlockProduct orders candidates by the block-product metric in~\eqref{eq:block_product_metric}, whereas our ORBGRAND-AI implementation orders them by the sum of their pooled block-substitution ranks. Neither ordering is generally identical to that induced by the matched full Gaussian energy.
}
\label{fig:main_results}
\end{figure*}

We also compare the methods on the same $10^4$ frames. For each frame, we generate one random code, one transmitted message, and one received vector and use them for all decoders. At nominal $(E_b/N_0)_{\mathrm{dB}}=2$, the empirical BLER differences are
\begin{align}
\widehat{P}_{\mathrm e}^{\mathrm{LP}} -\widehat{P}_{\mathrm e}^{\mathrm{EBP},\,b=8}
&=-0.0116,
\label{eq:paired_bler_difference_ebp}
\\
\widehat{P}_{\mathrm e}^{\mathrm{LP}} -\widehat{P}_{\mathrm e}^{\mathrm{AI},\,b=8}
&=-0.0205.
\label{eq:paired_bler_difference_ai}
\end{align}
Paired bootstrap resampling of the framewise error differences, using $10^4$ resamples, gives percentile $95\%$ intervals of $[-0.0141,-0.0090]$ and $[-0.0238,-0.0174]$, respectively.

\begin{table}[!b]
\centering
\caption{Empirical performance at nominal
$(E_b/N_0)_{\mathrm{dB}}=2$, $\rho=0.5$, and $q_{\max}=20{,}000$ over
$10^4$ paired random-code ensemble frames.}
\label{tab:main_2db}
\scriptsize
\setlength{\tabcolsep}{3.0pt}
\resizebox{\columnwidth}{!}{%
\begin{tabular}{lrrrr}
\toprule
Method
& Errors
& BLER
& \shortstack{Mean valid\\tests}
& \shortstack{$99$th percentile\\of valid tests}
\\
\midrule
LP-GRAND
& $274$
& $0.0274$
& $103.2$
& $2053.1$
\\
ExactBlockProduct, $b=8$
& $390$
& $0.0390$
& $161.7$
& $3210.2$
\\
ORBGRAND-AI, $b=8$
& $479$
& $0.0479$
& $212.4$
& $4145.0$
\\
ExactBlockProduct, $b=4$
& $575$
& $0.0575$
& $244.2$
& $4605.5$
\\
ORBGRAND-AI, $b=4$
& $701$
& $0.0701$
& $294.5$
& $4985.6$
\\
ExactMemoryless-GRAND
& $3400$
& $0.3400$
& $1389.5$
& $12902.2$
\\
ORBGRAND
& $3481$
& $0.3481$
& $1437.7$
& $12528.0$
\\
\bottomrule
\end{tabular}}
\end{table}

At nominal $(E_b/N_0)_{\mathrm{dB}}=4$, LP-GRAND, ExactBlockProduct with $b=8$, and our ORBGRAND-AI implementation with $b=8$ produced $3$, $5$, and $9$ block errors, respectively, among $10^4$ frames. These small error counts give substantial relative uncertainty and do not support a precise estimate of the SNR separation between the methods at a target BLER of $10^{-3}$. Appendix~\ref{app:exploratory_crossing} therefore reports the resulting log-linear interpolation only as an exploratory calculation.

Table~\ref{tab:baseline_work} reports sample means of selected operation counts at nominal $(E_b/N_0)_{\mathrm{dB}}=2$. Because the columns count different operations, they should not be interpreted as a common unit of computational complexity.

\begin{table}[!b]
\centering
\caption{Mean operation counts at nominal $(E_b/N_0)_{\mathrm{dB}}=2$. Entries are sample means over the same $10^4$ paired random-code ensemble frames used in Table~\ref{tab:main_2db}.}
\label{tab:baseline_work}
\scriptsize
\setlength{\tabcolsep}{2.8pt}
\resizebox{\columnwidth}{!}{%
\begin{tabular}{lrrcc}
\toprule
Method
& \shortstack{Valid\\tests}
& \shortstack{Queue\\removals}
& \shortstack{Local assignments\\evaluated}
& \shortstack{Suffix-DP\\state updates}
\\
\midrule
LP-GRAND
& $103.2$
& $1775.1$
& --
& $128$
\\
ExactBlockProduct, $b=8$
& $161.7$
& $516.9$
& $2048$
& --
\\
ORBGRAND-AI, $b=8$
& $212.4$
& $1276.8$
& $2048$
& --
\\
\bottomrule
\end{tabular}}
\end{table}

For LP-GRAND, queue removals count extractions of partial and complete trellis paths from the best-first priority queue. The suffix dynamic program updates $2^\nu$ states at each of the $n$ trellis stages, giving $n2^\nu=128$ state updates for $n=64$ and $\nu=1$. Each state update evaluates two outgoing branches, for a total of $n2^{\nu+1}=256$ branch evaluations. For the block-based generators, queue removals count subset-state extractions, including states rejected because more than one selected substitution belongs to the same block. The local-assignment count is $J2^b=8\cdot2^8=2048$, since all $2^b$ binary assignments are evaluated in each of the $J=8$ blocks. These counts describe different operations in the respective implementations and are not directly comparable.

To examine the effect of the abandonment-counter convention, Table~\ref{tab:orb_budget_modes} reports a separate result using $100$ common frames. In the principal experiments, $q_{\max}$ counts only conflict-free candidates submitted to the codebook-membership test. In Algorithm~1 of~\cite{duffy2023using}, the abandonment counter is incremented before the substitution-conflict check. In our implementation, the corresponding state-removal convention charges every priority-queue removal, including removal of a conflicting subset state, against the budget. Since the two conventions permit different numbers of codebook-membership tests for the same numerical value of $q_{\max}$, their BLER values are not obtained under an equal membership-test budget.

\begin{table}[!b]
\centering
\caption{Abandonment-counter result for our ORBGRAND-AI implementation with $b=8$ at nominal $(E_b/N_0)_{\mathrm{dB}}=2$ over $100$ common frames. State removals include conflicting subset states. The codeword-return rate is the fraction of frames in which the decoder returns a codeword before exhausting the budget.}
\label{tab:orb_budget_modes}
\scriptsize
\setlength{\tabcolsep}{2.6pt}
\resizebox{\columnwidth}{!}{%
\begin{tabular}{rlccrrr}
\toprule
$q_{\max}$
& Budget convention
& BLER
& \shortstack{Codeword-return\\rate}
& \shortstack{Mean valid\\tests}
& \shortstack{Mean state\\removals}
& \shortstack{Mean rejected\\conflicts}
\\
\midrule
$1000$
& Valid tests
& $0.08$
& $0.94$
& $108.5$
& $364.1$
& $255.7$
\\
$1000$
& State removals
& $0.10$
& $0.90$
& $66.2$
& $141.8$
& $75.6$
\\
$20000$
& Valid tests
& $0.07$
& $1.00$
& $396.2$
& $4695.8$
& $4299.6$
\\
$20000$
& State removals
& $0.07$
& $0.97$
& $242.5$
& $1076.4$
& $833.8$
\\
\bottomrule
\end{tabular}}
\end{table}

\subsection{Ordering Robustness Under Quantization and Precision-Matrix Mismatch}
\label{ssec:ordering_robustness}

Let $W$ and $\widehat W$ be two candidate metrics. Under the fixed deterministic tie-breaking rule, let $z_V^{(1)},z_V^{(2)},\ldots$ denote the configurations in nondecreasing $V$-cost, for $V\in\{W,\widehat W\}$. Define
\begin{equation}
\mathcal S_K(V) \triangleq \bigl\{z_V^{(1)},\ldots,z_V^{(K)}\bigr\},
\qquad V\in\{W,\widehat W\},
\end{equation}
and
\begin{equation}
\mathcal U_K(W,\widehat W) \triangleq \mathcal S_K(W)\cup\mathcal S_K(\widehat W).
\end{equation}
Thus, $\mathcal S_K(W)$ and $\mathcal S_K(\widehat W)$ describe top-$K$ membership, whereas comparisons within their union describe changes in relative order.

\begin{figure*}[t]
\centering
\begin{minipage}{0.4\linewidth}
\centering
\includegraphics[width=\linewidth]{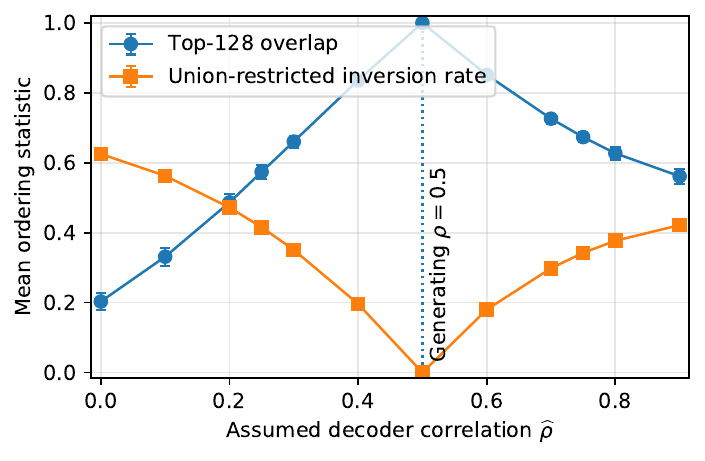}
\end{minipage}~
\begin{minipage}{0.4\linewidth}
\centering
\includegraphics[width=\linewidth]{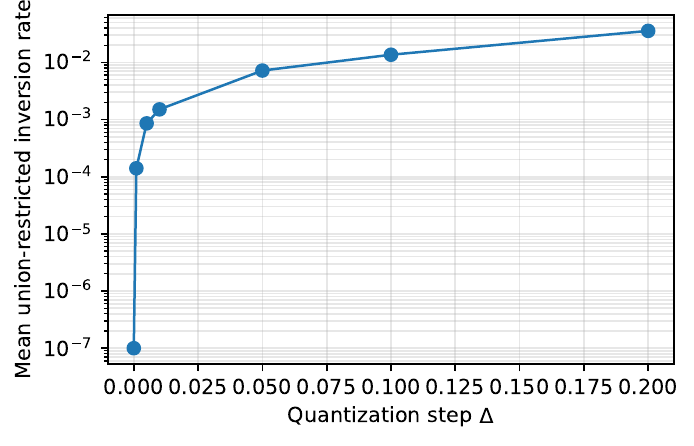}
\end{minipage}
\caption{Ordering sensitivity for $K=128$. Left: top-$128$ overlap and union-restricted inversion rate as functions of the assumed decoder correlation $\widehat\rho$, using $100$ common frames generated with $\rho=0.5$ at nominal $(E_b/N_0)_{\mathrm{dB}}=2.5$. Markers denote
sample means over frames, error bars denote nonparametric frame-bootstrap $95\%$ intervals, and the dotted vertical line marks the matched value. Right: mean union-restricted inversion rate as a function of the coefficient quantization step $\Delta$, using a separate set of $100$
common frames at the same operating point. The exact zero at $\Delta=0$ is displayed at $10^{-7}$ to permit the logarithmic vertical axis. Unordered pairs tied under either stored-coefficient metric are excluded from the inversion-rate denominator.}
\label{fig:robustness}
\end{figure*}

For distinct $z,z'\in\mathcal U_K(W,\widehat W)$, let
\begin{equation}
s_V(z,z') \triangleq \operatorname{sgn}\!\bigl(V(z)-V(z')\bigr),
\qquad V\in\{W,\widehat W\}.
\end{equation}
The comparable unordered pairs are
\begin{equation}
\begin{aligned}
\mathcal P_K(W,\widehat W) \triangleq \Bigl\{\{z,z'\}\ \Bigm|\ &z,z'\in\mathcal U_K(W,\widehat W),\ z\neq z',\\
&s_W(z,z')\neq0,\quad s_{\widehat W}(z,z')\neq0 \Bigr\}.
\end{aligned}
\end{equation}
A pair is therefore excluded when it is tied under either metric. When $\mathcal P_K(W,\widehat W)$ is nonempty, define the union-restricted inversion rate by
\begin{equation}
\begin{aligned}
I_K(W,\widehat W) \triangleq
\frac{\sum_{\{z,z'\}\in\mathcal P_K(W,\widehat W)}
\one\!\left\{ s_W(z,z')s_{\widehat W}(z,z')<0 \right\}}{|\mathcal P_K(W,\widehat W)|}
%
% \sum_{\{z,z'\}\in\mathcal P_K(W,\widehat W)}
% \\[-1mm]
% &\qquad\times
% \one\!\left\{
%  s_W(z,z')s_{\widehat W}(z,z')<0
% \right\}.
\end{aligned}
\label{eq:union_inversion_rate}
\end{equation}
The product of the two signs is unchanged when $z$ and $z'$ are interchanged, so~\eqref{eq:union_inversion_rate} is well defined for unordered pairs. We also use the top-$K$ overlap
\begin{equation}
O_K(W,\widehat W)
\triangleq
\frac{|\mathcal S_K(W)\cap\mathcal S_K(\widehat W)|}{K}.
\label{eq:topk_overlap}
\end{equation}
If position $K$ lies within a tie class under either metric, the top-$K$ set and hence $O_K(W,\widehat W)$ depend on the specified deterministic tie-breaking rule. By contrast, the inversion-rate denominator excludes pairs tied under either metric. For each frame, the number of pairs excluded because of a tie is $\binom{|\mathcal U_K(W,\widehat W)|}{2} -|\mathcal P_K(W,\widehat W)|$. Each reported mean is the arithmetic mean of the corresponding framewise statistic; in particular, the mean inversion rate is not formed by pooling all pair counts across frames.

All of our comparisons are exact with respect to the coefficients stored in IEEE~754 binary64 format. We convert every coefficient from both metrics to its exact dyadic-rational value and multiply all coefficients by one common positive power of two. Candidate-cost differences and ties are then evaluated using the resulting integer coefficients.

\paragraph{Coefficient quantization}
For this experiment, $W=W_r$ and $\widehat W=W_{r,\Delta}^{\mathrm q}$. We use $K=128$ and $100$ common systematic $[64,52]$ random-code ensemble frames at nominal $(E_b/N_0)_{\mathrm{dB}}=2.5$, $\rho=0.5$, and $q_{\max}=20{,}000$. The tested quantization steps are $\Delta\in\{0,0.001,0.005,0.01,0.05,0.1,0.2\}$. For a stored coefficient $a$ and $\Delta>0$, the implementation stores $\Delta\operatorname{rint}(a/\Delta)$, where $\operatorname{rint}$ uses round-to-nearest with ties to even; $\Delta=0$ denotes the unquantized metric.

At $\Delta=0.1$, the mean top-$128$ overlap is $0.9893$ and the mean union-restricted inversion rate is $1.36\times10^{-2}$. At $\Delta=0.2$, the corresponding values are $0.9697$ and $3.55\times10^{-2}$. The mean numbers of unordered pairs excluded because of an exact tie under at least one metric are $70.3$ and $153.0$, respectively. One block error was observed among the $100$ frames for the unquantized metric and for every tested positive value of $\Delta$, giving an empirical BLER of $0.01$ and a common $95\%$ Wilson interval of $[0.0018,0.0545]$. These error counts do not resolve a dependence of BLER on $\Delta$.

\paragraph{Precision-matrix mismatch}
We next compare the matched metric, constructed with the generating correlation $\rho=0.5$, with metrics constructed using $\widehat\rho\in \{0,0.1,0.2,0.25, 0.3,0.4,0.5, 0.6,0.7,0.75,0.8,0.9\}$. In our experiment we use $K=128$ and $100$ common systematic $[64,52]$ random-code ensemble frames at nominal $(E_b/N_0)_{\mathrm{dB}}=2.5$. At the matched value $\widehat\rho=0.5$, every frame has $O_{128}=1$ and $I_{128}=0$. At $\widehat\rho=0.25$, the mean overlap is $0.5738$, with a nonparametric frame-bootstrap $95\%$ interval of $[0.5532,0.5944]$, and the mean inversion rate is $0.4155$, with interval $[0.4011,0.4297]$. At $\widehat\rho=0.75$, the corresponding values are $0.6734$, with interval $[0.6590,0.6873]$, and $0.3424$, with interval $[0.3305,0.3551]$. No exact tie occurred under either stored-coefficient metric in this experiment, so every unordered pair in each union was comparable.

The measured ordering statistics are not symmetric about $\widehat\rho=\rho$. An order-preserving normalization of the first-order Gauss--Markov precision matrix gives a useful partial comparison. For $\eta\in[0,1)$, multiplying $Q(\eta)$ by $\sigma^2(1-\eta^2) / (1+\eta^2)$ preserves the candidate order. In the rescaled matrix, the interior diagonal entries equal one, the endpoint diagonal entries equal $1/(1+\eta^2)$, and the magnitude of each nearest-neighbor off-diagonal entry is
\begin{equation}
g(\eta) \triangleq \frac{\eta}{1+\eta^2}.
\end{equation}
In particular,
\begin{equation}
g(0.25)=\frac{4}{17},
\qquad
g(0.5)=\frac{2}{5},
\qquad
g(0.75)=\frac{12}{25},
\end{equation}
and hence
\begin{equation}
|g(0.75)-g(0.5)| =\frac{2}{25} <\frac{14}{85} =|g(0.25)-g(0.5)|.
\end{equation}
Thus, under this normalization, the interior coupling coefficient at $\widehat\rho=0.75$ is closer to its matched value than the coefficient at $\widehat\rho=0.25$. This comparison does not determine the complete candidate order because the endpoint diagonal entries and the received-vector-dependent unary coefficients also change with $\widehat\rho$.

A separate finite-budget experiment uses $50$ common frames at nominal $(E_b/N_0)_{\mathrm{dB}}=2.5$, generating correlation $\rho=0.5$, and $q_{\max}=20{,}000$. For $\widehat\rho\in\{0,0.25,0.5,0.75,0.9\}$, the observed error counts are $15$, $3$, $1$, $1$, and $1$, and the corresponding sample mean numbers of valid codebook-membership tests are $771.5$, $77.4$, $43.6$, $82.3$, and $110.2$. The empirical BLERs are therefore $0.30$, $0.06$, and $0.02$ for each of the final three settings. Their $95\%$ Wilson intervals are $[0.191,0.438]$, $[0.021,0.162]$, and $[0.004,0.105]$, respectively. The sample is too small to support a general performance ordering across mismatched decoder models.

\subsection{Calibration of the Random-Codebook First-Hit Estimate}
\label{ssec:first_hit_calibration}

\begin{figure*}[t]
\centering
\includegraphics[width=0.8\linewidth]
{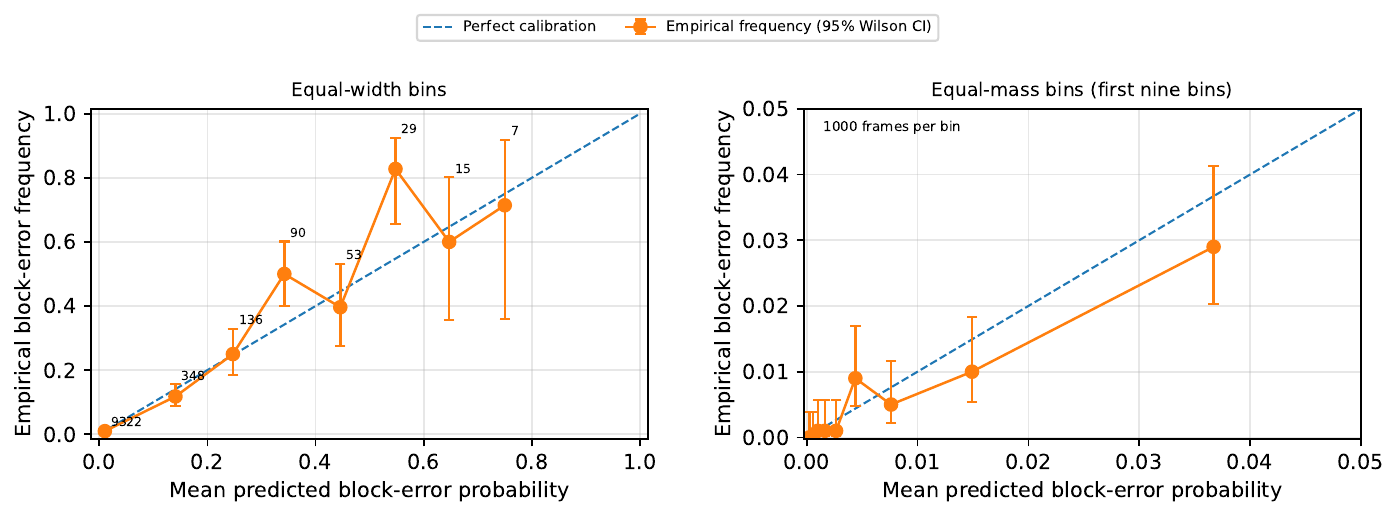}
\caption{Calibration of the random-codebook first-hit error estimate for $10^4$ systematic $[64,52]$ random-linear-code ensemble frames at nominal $(E_b/N_0)_{\mathrm{dB}}=2$ and $\rho=0.5$. Horizontal coordinates are the binwise mean predicted error probabilities, markers are the corresponding empirical block-error frequencies, and vertical bars are $95\%$ Wilson score intervals. The dashed line denotes perfect calibration. Left: ten equal-width bins on $[0,1]$; annotations give the number of frames in each nonempty bin. Right: the first nine equal-mass bins, each containing $1000$ frames. For the tenth equal-mass bin, the mean predicted error probability was $0.19357$, the empirical error frequency was $0.212$ $(212/1000)$, the $95\%$ Wilson score interval was $[0.1878,0.2384]$, and the prediction range was $[0.06310,0.79858]$.}
\vspace{30pt}
\label{fig:calibration_final}
\end{figure*}

We assess the first-hit correctness estimate in \eqref{eq:first_hit_correctness} using a separate experiment with $N_{\mathrm{cal}}=10^4$ frames from the systematic $[64,52]$ random-linear-code ensemble~\cite{yuan2025soft,duffy2024soft}. The operating point is nominal $(E_b/N_0)_{\mathrm{dB}}=2$, with $\rho=0.5$ and $q_{\max}=20{,}000$. For frame $i$, let $r_i$ denote the received vector, let $H_i$ denote the query index of the first candidate that belongs to the codebook, and let $\widehat C_i$ denote that codeword. We define
\begin{equation}
\widehat p_i \triangleq 1-P_{\mathrm{cor}}(H_i\mid r_i),
\qquad e_i \triangleq \one\!\left\{\widehat C_i\neq C_i\right\},
\label{eq:calibration_predictions}
\end{equation}
where $C_i$ is the transmitted codeword. Thus, $\widehat p_i$ is the predicted probability that the first codebook hit is incorrect, and $e_i$ is the observed block-error indicator. Every frame returned a codeword before reaching $q_{\max}$, so the calibration analysis contains no abandoned outcomes.

We evaluate the predictions using the Brier score \cite{glenn1950verification}
\begin{equation}
\operatorname{BS} \triangleq \frac{1}{N_{\mathrm{cal}}}
\sum_{i=1}^{N_{\mathrm{cal}}} \left(\widehat p_i-e_i\right)^2
\label{eq:calibration_brier}
\end{equation}
and the binary log loss
\begin{equation}
\operatorname{LL} \triangleq -\frac{1}{N_{\mathrm{cal}}} \sum_{i=1}^{N_{\mathrm{cal}}}
\left[ e_i\log \widehat p_i +(1-e_i)\log(1-\widehat p_i) \right].
\label{eq:calibration_logloss}
\end{equation}
The logarithm is natural. At the endpoints, we use the standard extended-real conventions $0\log 0\triangleq0$ and $-\log 0\triangleq+\infty$.

We also report reliability diagrams and expected calibration error (ECE) \cite{guo2017calibration} for two ten-bin partitions. For equal-width binning, the first nine bins are $[(b-1)/10,b/10)$, $b=1,\ldots,9$, and the final bin is $[0.9,1]$. For equal-mass binning, the frames are sorted by $\widehat p_i$ using a fixed deterministic tie-breaking rule and then divided into ten groups of $1000$ frames. For either partition, let $\mathcal I_b$ denote bin $b$, and let
\begin{equation}
\mathcal B
\triangleq
\left\{b\in\{1,\ldots,10\}:|\mathcal I_b|>0\right\}
\end{equation}
denote the nonempty bins. For each $b\in\mathcal B$, define
\begin{equation}
\overline p_b \triangleq \frac{1}{|\mathcal I_b|} \sum_{i\in\mathcal I_b}\widehat p_i,
\qquad
\overline e_b \triangleq \frac{1}{|\mathcal I_b|} \sum_{i\in\mathcal I_b}e_i.
\label{eq:calibration_bin_means}
\end{equation}
The ECE is
\begin{equation}
\operatorname{ECE} \triangleq \sum_{b\in\mathcal B} \frac{|\mathcal I_b|}{N_{\mathrm{cal}}}
\left|\overline e_b-\overline p_b\right|.
\label{eq:calibration_ece}
\end{equation}
Because ECE depends on the partition, we report it for both binning rules.

Across the $10^4$ frames, the empirical block-error rate was $0.0268$, and the sample mean of the predicted error probabilities was $0.02633$. The Brier score was $0.01945$, with a nonparametric frame-bootstrap $95\%$ interval of $[0.01729,0.02165]$, and the log loss was $0.07364$, with interval $[0.06601,0.08134]$. The equal-width ECE was $0.00407$, with interval $[0.00304,0.00732]$, and the equal-mass ECE was $0.00414$, with interval $[0.00197,0.00666]$.

Under equal-width binning, the lowest-probability bin contained $9322$ frames. Under equal-mass binning, each bin contained $1000$ frames, and the mean predicted error probability in each of the first nine bins was at most $0.03672$. In the tenth bin, the predicted probabilities ranged from $0.06310$ to $0.79858$, with mean $0.19357$. Its empirical error frequency was $0.212$ $(212/1000)$, with a $95\%$ Wilson score interval of $[0.1878,0.2384]$. The right panel of Fig.~\ref{fig:calibration_final} shows only the first nine equal-mass bins so that their low-probability range remains visible; the tenth-bin result is reported here and in the caption.

These results assess calibration for the tested systematic random-linear-code ensemble and operating point. They do not show that \eqref{eq:first_hit_correctness} is the exact codebook-conditioned posterior for an individual realized code. A systematic random linear codebook is not a uniformly sampled fixed-cardinality subset of $\{0,1\}^n$, and its codeword locations in the query order are dependent. The estimate in \eqref{eq:first_hit_correctness} should therefore be interpreted as an ensemble-based approximation for the structured codebooks considered here.

\vspace{20pt}

% =================================
\section{Proofs of the Main Results}
\label{sec:omitted_proofs}
% =================================

\subsection{Proof of Proposition~\ref{prop:quadratic_metric}}
\label{app:proof_prop_quadratic_metric}

\begin{proof}
The linear contribution in $z$ in~\eqref{eq:energy_expand} is
\begin{equation}
2zD_sQa = \sum_{i=1}^n 2s_i(Qa)_i z_i.
\label{eq:linear_energy_term}
\end{equation}
Since $Q$ is symmetric, $D_sQD_s$ is also symmetric. Hence,
\begin{align}
2zD_sQD_s z^{\mathsf T}
&= 2\sum_{i=1}^n s_i^2Q_{ii}z_i^2
+ 4\sum_{1\leq i<j\leq n} s_i s_jQ_{ij}z_i z_j \nonumber\\
&= 2\sum_{i=1}^n Q_{ii}z_i
+ 4\sum_{1\leq i<j\leq n} s_i s_jQ_{ij}z_i z_j,
\label{eq:quadratic_energy_term}
\end{align}
where the second equality follows from $s_i^2=1$ and $z_i^2=z_i$ for $z_i\in\{0,1\}$. Substituting \eqref{eq:linear_energy_term} and \eqref{eq:quadratic_energy_term} into~\eqref{eq:energy_expand} and collecting the unary and pairwise terms yields \eqref{eq:energy_decomposition}--\eqref{eq:beta}.
\end{proof}

\subsection{Proof of Proposition~\ref{prop:whitening_relation}}
\label{app:proof_prop_whitening_relation}

\begin{proof}
The identity in~\eqref{eq:whitened_metric} follows directly from
$Q=U^{\mathsf T}U$. Since $V=UN$ is a linear transformation of a
Gaussian random vector, it is Gaussian with zero mean and covariance
\begin{align}
\operatorname{Cov}(V)=
UQ^{-1}U^{\mathsf T}
= U\bigl(U^{-1}U^{-\mathsf T}\bigr)U^{\mathsf T}
= I_n.
\end{align}
Hence, $V\sim\mathcal{N}(0,I_n)$. Multiplying $R=x(C)+N$ by $U$ then yields~\eqref{eq:whitened_channel}.

For a symmetric positive-definite banded matrix, unpivoted Cholesky factorization preserves the bandwidth \cite{davis2006direct}. Therefore, if $Q$ has half-bandwidth at most $\nu$, then its upper-triangular Cholesky factor satisfies $U_{ij}=0$ whenever $j-i>\nu$, which establishes the stated bandedness property.
\end{proof}

\subsection{Proof of Theorem~\ref{thm:ml}}
\label{app:proof_thm_ml}

\begin{proof}
Define the codebook-compatible noise-effect set
\begin{equation}
\mathcal{Z}_{\calC}(y) \triangleq
\bigl\{ z\in\{0,1\}^n: y\xor z\in\calC \bigr\}.
\label{eq:feasible_noise_effect_set}
\end{equation}
Because the mapping $z\mapsto y\xor z$ is bijective on $\{0,1\}^n$ and $\calC$ is nonempty, $\mathcal{Z}_{\calC}(y)$ is nonempty. Hence, $q_\star$ in \eqref{eq:first_codebook_hit} exists under complete enumeration.
For every $z\in\mathcal{Z}_{\calC}(y)$,
\begin{equation}
f_{R|C}\bigl(r\mid y\xor z\bigr) =
K_\Sigma\exp\bigl\{-E_r(z)\bigr\},
\end{equation}
where
\begin{equation}
K_\Sigma \triangleq (2\pi)^{-n/2}\det(\Sigma)^{-1/2}
\end{equation}
is independent of $z$. Maximizing $f_{R|C}(r\mid c)$ over $c\in\calC$ is therefore equivalent, under the bijection $c=y\xor z$, to minimizing $E_r(z)$ over $\mathcal{Z}_{\calC}(y)$. Since $E_r(z)=E_r(0)+W_r(z)$ and $E_r(0)$ is independent of $z$, this is equivalent to minimizing $W_r(z)$ over $\mathcal{Z}_{\calC}(y)$.

By~\eqref{eq:order}, the first emitted pattern belonging to $\mathcal{Z}_{\calC}(y)$, namely $z^{(q_\star)}$, minimizes $W_r$ over
that set. Consequently, $\widehat{c}_{\mathrm{ML}}=y\xor z^{(q_\star)}$ maximizes $f_{R|C}(r\mid c)$ over $c\in\calC$ and is therefore an ML codeword. If several elements of $\mathcal{Z}_{\calC}(y)$ attain the minimum, the fixed deterministic tie-breaking rule selects one of the corresponding ML codewords. Because the prior on $\calC$ is uniform, the ML and MAP decision sets coincide.
\end{proof}

\subsection{Proof of Theorem~\ref{thm:trellis}}
\label{app:proof_thm_trellis}

\begin{proof}
For $\nu=0$, each length-$n$ path beginning at the initial state $\varnothing$ is uniquely determined by its branch-label sequence $z=(z_1,\ldots,z_n)$. Since the pairwise sum in \eqref{eq:banded_metric} is empty,
\begin{equation}
\sum_{i=1}^n\gamma_i(\varnothing,z_i) = \sum_{i=1}^n\alpha_i z_i = W_r(z).
\end{equation}
For $\nu\geq1$, each length-$n$ path beginning at $u_1=(0,\ldots,0)$ contains one branch label $z_i\in\{0,1\}$ at each layer and therefore determines a unique $z\in\{0,1\}^n$. Conversely, each $z\in\{0,1\}^n$, together with the initial state and the recursion $u_{i+1}=T(u_i,z_i)$, determines a unique length-$n$ path.
For $d=1,\ldots,\min\{\nu,i-1\}$, the state definition in
\eqref{eq:banded_state} gives $(u_i)_{\nu-d+1}=z_{i-d}$. Therefore,
\begin{equation}
\gamma_i(u_i,z_i) = \alpha_i z_i + \sum_{d=1}^{\min\{\nu,i-1\}} \beta_{i-d,i}z_{i-d}z_i.
\end{equation}
Summing over $i=1,\ldots,n$ includes each unary term exactly once and each pairwise term exactly once, at the layer indexed by the larger of its two coordinate indices. The identity in \eqref{eq:path_cost_identity} then follows from \eqref{eq:banded_metric}.
\end{proof}

\subsection{Proof of Corollary~\ref{cor:exact_trellis_queries}}
\label{app:proof_cor_exact_trellis_queries}

\begin{proof}
Theorem~\ref{thm:trellis} establishes a bijection between length-$n$ paths beginning at the initial state and noise-effect patterns $z\in\{0,1\}^n$, with the cost of the path labeled by $z$ equal to $W_r(z)$. Therefore, enumerating these paths in nondecreasing total cost,using the deterministic tie-breaking policy defining the sequence in~\eqref{eq:order}, produces the LP-GRAND ordering of noise-effect patterns and hence the corresponding candidate-word query order. Each layer contains at most $2^\nu$ states, giving at most $(n+1)2^\nu$ states across the $n+1$ state layers. Since each state in the first $n$ layers has at most two outgoing branches, the trellis contains at most $n2^{\nu+1}$ branches.
\end{proof}

\subsection{Proof of Theorem~\ref{thm:pathwidth}}
\label{app:proof_thm_pathwidth}

\begin{proof}
Assign each term in~\eqref{eq:general_pairwise_energy} to the first bag containing all variables on which that term depends. For each $i\in\mathcal{V}$, define
\begin{equation}
\tau(i) \triangleq \min\bigl\{ t:i\in\mathcal{B}_t \bigr\},
\end{equation}
and, for each $\{i,j\}\in\mathcal{E}$, define
\begin{equation}
\tau(i,j) \triangleq
\min\bigl\{ t:\{i,j\}\subseteq\mathcal{B}_t \bigr\}.
\end{equation}
These indices exist by the vertex- and edge-coverage properties of the path decomposition.
For each assignment $a\in\{0,1\}^{\mathcal{B}_t}$, define the local cost
\begin{equation}
\lambda_t(a) \triangleq
\sum_{\substack{i\in\mathcal{V}\\ \tau(i)=t}} \theta_i(a_i) +
\sum_{\substack{\{i,j\}\in\mathcal{E}\\ \tau(i,j)=t}}
\theta_{\{i,j\}}(a_i,a_j).
\label{eq:bag_local_cost}
\end{equation}
Layer $t$ contains one state for each assignment $a_t\in\{0,1\}^{\mathcal{B}_t}$. A directed edge connects a state $a_t$ in layer $t$ to a state $a_{t+1}$ in layer $t+1$ if and only if
\begin{equation}
a_t\big|_{\mathcal{B}_t\cap\mathcal{B}_{t+1}}
= a_{t+1}\big|_{\mathcal{B}_t\cap\mathcal{B}_{t+1}}.
\label{eq:bag_consistency}
\end{equation}
Add a source with an edge to every state in the first layer and a sink receiving an edge from every state in the last layer. Assign cost $\lambda_1(a_1)$ to the edge from the source to $a_1$, cost $\lambda_{t+1}(a_{t+1})$ to every edge entering state $a_{t+1}$ in layer $t+1$, and zero cost to every edge entering the sink.

Consider a source-to-sink path with bag assignments $a_1,\ldots,a_m$. For each $v\in\mathcal{V}$, the bags containing $v$ form an interval. The consistency condition in~\eqref{eq:bag_consistency} therefore forces all assignments associated with bags containing $v$ to agree on the value of $v$. Hence, the path defines a unique global assignment $z\in\{0,1\}^{\mathcal{V}}$. Conversely, each global assignment $z$ defines a unique source-to-sink path through $a_t=z\big|_{\mathcal{B}_t}$ for $t=1,\ldots,m$.
Each unary and pairwise term is assigned to exactly one bag and evaluated exactly once. Therefore, $\sum_{t=1}^m\lambda_t(a_t)= W(z)$. Since $|\mathcal{B}_t|\leq w+1$, each layer contains at most $2^{w+1}$ states.
\end{proof}

\subsection{Proof of Theorem~\ref{thm:quant}}
\label{app:proof_thm_quant}

\begin{proof}
Define the coefficient errors $e_i\triangleq\alpha_i-\Delta\widetilde{\alpha}_i$ for
$i=1,\ldots,n$ and $e_{ij}\triangleq\beta_{ij}-\Delta\widetilde{\beta}_{ij}$ for $\{i,j\}\in\mathcal{E}_Q$. By nearest-integer rounding, $|e_i|\leq\Delta/2$ and $|e_{ij}|\leq\Delta/2$. Therefore,
\begin{align}
\left| W_r(z)-W_{r,\Delta}^{\mathrm q}(z) \right|
&= \left| \sum_{i=1}^n e_i z_i + \sum_{\{i,j\}\in\mathcal{E}_Q}
e_{ij}z_i z_j \right|
\nonumber\\
&\leq
\sum_{i=1}^n |e_i|z_i + \sum_{\{i,j\}\in\mathcal{E}_Q} |e_{ij}|z_i z_j
\nonumber\\
&\leq \frac{\Delta}{2} \bigl(n+|\mathcal{E}_Q|\bigr)
= \frac{\Delta}{2}M_Q,
\end{align}
where the first inequality follows from the triangle inequality and the fact that $z_i,z_i z_j\in \{0,1\}$. The final inequality in~\eqref{eq:uniform_quantization_bound_banded} follows from \eqref{eq:coefficient_count_bound}.
\end{proof}

\subsection{Proof of Corollary~\ref{cor:quant_order}}
\label{app:proof_cor_quant_order}

\begin{proof}
By Theorem~\ref{thm:quant},
\begin{align}
&\left| \left[ W_{r,\Delta}^{\mathrm q}(z)
- W_{r,\Delta}^{\mathrm q}(z') \right]
- \left[ W_r(z)-W_r(z') \right] \right| \nonumber\\
&\qquad\leq \left| W_{r,\Delta}^{\mathrm q}(z)-W_r(z) \right|
+ \left| W_{r,\Delta}^{\mathrm q}(z')-W_r(z') \right| \nonumber\\
&\qquad \leq \Delta M_Q.
\end{align}
Under~\eqref{eq:uniform_order_condition}, the perturbation of the candidate-energy gap is strictly smaller than the magnitude of the unquantized gap. Therefore, the quantized and unquantized gaps have the same sign. Since $W_{r,\Delta}^{\mathrm q} =\Delta\widetilde{W}_{r,\Delta}$ and $\Delta>0$, the metrics $W_{r,\Delta}^{\mathrm q}$ and $\widetilde{W}_{r,\Delta}$ induce identical strict comparisons, which proves~\eqref{eq:uniform_order_preservation}.
\end{proof}

\subsection{Proof of Corollary~\ref{cor:pair_specific_quantization}}
\label{app:proof_cor_pair_specific_quantization}

\begin{proof}
Using the coefficient errors defined in the proof of Theorem~\ref{thm:quant}, the perturbation of the candidate-energy gap is
\begin{equation}
\sum_{i=1}^n e_i(z_i-z'_i)
+ \sum_{\{i,j\}\in\mathcal{E}_Q}
e_{ij} \bigl( z_i z_j-z'_i z'_j \bigr).
\end{equation}
The triangle inequality and the bounds $|e_i|\leq\Delta/2$ and $|e_{ij}|\leq\Delta/2$ yield \eqref{eq:pair_specific_gap_bound}. Under \eqref{eq:pair_specific_certificate}, the perturbation is strictly
smaller than the magnitude of the unquantized energy gap. Hence, the quantized and unquantized gaps have the same sign, which preserves the strict order of $z$ and $z'$.
\end{proof}

\subsection{Proof of Proposition~\ref{prop:precision_mismatch}}
\label{app:proof_prop_precision_mismatch}

\begin{proof}
For $z\in\{0,1\}^n$, define $v_z\triangleq r-x_z$. Then \begin{align}
\left| E_r(z)-\widehat{E}_r(z) \right|
& = \frac{1}{2} \left| v_z^{\mathsf T} (Q-\widehat{Q}) v_z \right|
\nonumber\\
&\leq \frac{1}{2} \lVert Q-\widehat{Q}\rVert_2 \lVert v_z\rVert_2^2
\nonumber\\
&\leq \frac{\varepsilon}{2} \lVert r-x_z\rVert_2^2.
\label{eq:candidate_energy_mismatch}
\end{align}
Since $x_z\in\{-1,+1\}^n$, we have $\lVert x_z\rVert_2=\sqrt{n}$. The triangle inequality therefore gives $\lVert r-x_z\rVert_2 \leq \lVert r\rVert_2+\sqrt{n}$, which proves \eqref{eq:uniform_energy_mismatch_bound}.
For any $z,z'\in\{0,1\}^n$,
\begin{align}
&\left| 
\left[ \widehat{E}_r(z)-\widehat{E}_r(z') \right]
- \left[ E_r(z)-E_r(z') \right] \right| \nonumber\\
&\qquad\leq
\left| \widehat{E}_r(z)-E_r(z) \right|
+ \left| \widehat{E}_r(z')-E_r(z') \right| \nonumber\\
&\qquad\leq \varepsilon \left( \lVert r\rVert_2+\sqrt{n} \right)^2.
\label{eq:pairwise_gap_mismatch}
\end{align}
If $\left|E_r(z)-E_r(z')\right| > \varepsilon\bigl(\lVert r\rVert_2+\sqrt{n}\bigr)^2$, then the bound in~\eqref{eq:pairwise_gap_mismatch} is strictly smaller than the magnitude of the matched-metric energy gap. Therefore, the matched and mismatched energy gaps have the same sign, which proves~\eqref{eq:mismatch_order_preservation}.
\end{proof}

\subsection{Proof of Proposition~\ref{prop:pair_specific_precision_mismatch}}
\label{app:proof_prop_pair_specific_precision_mismatch}

\begin{proof}
Let $a\triangleq r-x_z$, $b\triangleq r-x_{z'}$. Then
\begin{align}
&\left[ \widehat{E}_r(z)-\widehat{E}_r(z') \right]
- \left[ E_r(z)-E_r(z') \right] \nonumber\\
&\qquad= \frac{1}{2} \left( a^{\mathsf T}\Delta_Q a -
b^{\mathsf T}\Delta_Q b \right).
\label{eq:pair_gap_perturbation}
\end{align}
Since $\Delta_Q$ is symmetric,
\begin{equation}
a^{\mathsf T}\Delta_Q a - b^{\mathsf T}\Delta_Q b
= (a-b)^{\mathsf T}\Delta_Q(a+b).
\end{equation}
Moreover, $a-b=x_{z'}-x_z$, $a+b=2r-x_z-x_{z'}$. Therefore,
\begin{align}
&\left| \left[ \widehat{E}_r(z)-\widehat{E}_r(z') \right]
- \left[ E_r(z)-E_r(z') \right] \right| \nonumber\\
&\qquad= \frac{1}{2} \left| (x_z-x_{z'})^{\mathsf T} \Delta_Q
(x_z+x_{z'}-2r) \right| \nonumber\\
&\qquad\leq
\frac{\lVert\Delta_Q\rVert_2}{2} \lVert x_z-x_{z'}\rVert_2 \lVert x_z+x_{z'}-2r\rVert_2,
\end{align}
which proves \eqref{eq:pair_specific_precision_bound}. Under \eqref{eq:pair_specific_precision_certificate}, the perturbation is strictly smaller than the magnitude of the matched-metric energy gap. Therefore, the matched and mismatched energy gaps have the same sign, which preserves the strict order of $z$ and $z'$. Finally, since $x_z=s\odot\bigl(\mathbf{1}_n-2z^{\mathsf T}\bigr)$, the vector $x_z-x_{z'}$ has magnitude $2$ precisely at the coordinates where $z$ and $z'$ differ and is zero elsewhere. Hence, $\lVert x_z-x_{z'}\rVert_2^2 = 4d_{\mathrm H}(z,z')$, which proves~\eqref{eq:bpsk_hamming_distance}.
\end{proof}

\subsection{Proof of Proposition~\ref{prop:fullspace_weights}}
\label{app:proof_prop_fullspace_weights}

\begin{proof}
By Theorem~\ref{thm:trellis}, each length-$n$ trellis path beginning at the initial state corresponds to exactly one $z\in\{0,1\}^n$, and its total path cost is $W_r(z)$. Hence,
\begin{equation}
\sum_{\substack{P:\,\text{$n$-branch trellis path}\\
\text{beginning at }u_1}}
\!\!\!\!\!\!\!\!\!
\exp\bigl\{-\operatorname{cost}(P)\bigr\}
= \!\! \sum_{z\in\{0,1\}^n} \! \exp\bigl\{-W_r(z)\bigr\}
= Z_r.
\label{eq:path_partition_identity}
\end{equation}
We prove by induction on $i$ that $F_i(v)$ equals the sum of $\exp\{-\operatorname{cost}(P)\}$ over all length-$i$ partial paths beginning at $u_1$ and ending in state $v$. The initialization in~\eqref{eq:forward_initialization} establishes the claim for $i=0$. The recursion in~\eqref{eq:forward_partition_recursion} extends each length-$(i-1)$ partial path along every admissible outgoing branch and multiplies its weight by $\exp\{-\gamma_i(u,b)\}$. Thus, it accounts for every length-$i$ partial path exactly once. Summing $F_n(v)$ over all states $v\in\mathcal{U}_\nu$ gives \eqref{eq:path_partition_identity} and hence \eqref{eq:trellis_partition}.
The state space contains $2^\nu$ elements, including the singleton state space for $\nu=0$. In the full-state implementation, each state has two outgoing branches at each of the $n$ layers, yielding the stated branch-update count. Since $F_i$ depends only on $F_{i-1}$, two consecutive forward layers suffice. Finally, \eqref{eq:emitted_fullspace_weight} follows directly from~\eqref{eq:fullspace_weight}.
\end{proof}

\subsection{Proof of Theorem~\ref{thm:single_app}}
\label{app:proof_thm_single_app}

\begin{proof}
All probabilities in this proof are conditioned on $R=r$. Conditional on $J=j$, the occupancy set $\mathcal I$ is uniformly distributed over the $\binom{L-1}{M-1}$ subsets of $\{1,\ldots,L\}\setminus\{j\}$ having cardinality $M-1$. The event that the first hit occurs at position $q$ and corresponds to the transmitted codeword is $\{H=q,J=q\} = \left\{ J=q,\, \mathcal I\cap\{1,\ldots,q-1\}=\varnothing \right\}$. Conditional on $J=q$, all $M-1$ elements of $\mathcal I$ must belong to $\{q+1,\ldots,L\}$, which contains $L-q$ positions. Therefore,
\begin{equation}
\Pr \left\{H=q,J=q\mid R=r\right\} = w_q \frac{\binom{L-q}{M-1}}{\binom{L-1}{M-1}}
= \frac{A_q}{\binom{L-1}{M-1}}.
\label{eq:correct_first_hit_probability}
\end{equation}
For a fixed $t>q$, the event that the first hit occurs at position $q$ and is nontransmitted requires $J=t$, $\, q\in\mathcal I$, $\, \mathcal I\cap\{1,\ldots,q-1\}=\varnothing$.

Given $J=t$ and $q\in\mathcal I$, the remaining $M-2$ elements of $\mathcal I$ must be selected from $\{q+1,\ldots,L\}\setminus\{t\}$, which contains $L-q-1$ positions. Hence,
\begin{equation}
\Pr \left\{H=q,J=t\mid R=r\right\} = w_t \frac{\binom{L-q-1}{M-2}}{\binom{L-1}{M-1}},
\quad t=q+1,\ldots,L.
\end{equation}
Summing over $t=q+1,\ldots,L$ yields
\begin{equation}
\Pr\left\{H=q,J>q\mid R=r\right\} = T_q \frac{\binom{L-q-1}{M-2}}{\binom{L-1}{M-1}}
= \frac{B_q}{\binom{L-1}{M-1}}.
\label{eq:incorrect_first_hit_probability}
\end{equation}
Since $H=q$ implies $J\geq q$, the disjoint events $\{H=q,J=q\}$ and $\{H=q,J>q\}$ partition $\{H=q\}$. Consequently,
\begin{equation}
\Pr\left\{H=q\mid R=r\right\} = \frac{A_q+B_q}{\binom{L-1}{M-1}}.
\end{equation}
Dividing~\eqref{eq:correct_first_hit_probability} by this probability proves~\eqref{eq:first_hit_correctness}.
\end{proof}

\subsection{Proof of Proposition~\ref{prop:abandonment_tail_loss}}
\label{app:proof_prop_abandonment_tail_loss}

\begin{proof}
Let $\widehat{C}_{\mathrm{full}}$ and $\widehat{C}_{\mathrm{ab}}$ denote the outputs of the two decoders, with $\mathsf{ABANDON}$ distinct from every codeword, and define $\mathcal{E}_{\mathrm{full}} \triangleq \left\{\widehat{C}_{\mathrm{full}}\neq C\right\}$, $\mathcal{E}_{\mathrm{ab}} \triangleq \left\{\widehat{C}_{\mathrm{ab}}\neq C\right\}$. If the full decoder returns an incorrect codeword, the decoder with abandonment either returns the same first codebook hit or abandons before reaching it. Hence, $\mathcal{E}_{\mathrm{full}} \subseteq \mathcal{E}_{\mathrm{ab}}$.

On the event $Z^n\in\mathcal{A}_{\tau(R)}(R)$, the query corresponding to the transmitted codeword lies within the queried prefix. Therefore, the first codebook hit also lies within this prefix. Since the two decoders use the same query order, tie-breaking rule, and membership test, they return the same codeword. Thus,
\begin{equation}
\mathcal{E}_{\mathrm{ab}} \setminus \mathcal{E}_{\mathrm{full}}
\subseteq
\left\{ Z^n\notin\mathcal{A}_{\tau(R)}(R) \right\}.
\label{eq:additional_error_inclusion}
\end{equation}
Together with $\mathcal{E}_{\mathrm{full}}\subseteq\mathcal{E}_{\mathrm{ab}}$, this gives
\begin{equation}
P_{\mathrm e}^{\mathrm{ab}} - P_{\mathrm e}^{\mathrm{full}}
= \Pr\left\{ \mathcal{E}_{\mathrm{ab}} \setminus \mathcal{E}_{\mathrm{full}} \right\}
\leq \Pr \left\{ Z^n\notin\mathcal{A}_{\tau(R)}(R) \right\},
\end{equation}
which proves~\eqref{eq:general_abandonment_bound}.

By the law of total probability,
\begin{align}
\Pr \left\{ Z^n\notin\mathcal{A}_{\tau(R)}(R) \right\}
&= \mathbb{E} \left[ \Pr \left\{ Z^n\notin\mathcal{A}_{\tau(R)}(R) \mid R \right\} \right] \nonumber\\
&= \mathbb{E} \left[ T_{\tau(R)}^{\calC}(R) \right],
\end{align}
which proves~\eqref{eq:codebook_tail_abandonment_bound}. Under \eqref{eq:matched_flip_posterior}, $T_{\tau}^{\calC}(r)=T_{\tau}(r)$ for $P_R$-almost every $r$, yielding \eqref{eq:posterior_tail_abandonment_bound}. The final bound follows from $T_{\tau(r)}(r)\leq\varepsilon$ almost surely.
\end{proof}

\vspace{20pt}

% ======================================================================
\section{Extended Related Work}
\label{app:sec:related_work}
% ======================================================================

\paragraph{Code-structure-dependent soft decoding}
Many soft decoders are defined for a prescribed representation of the code. Belief-propagation decoding of LDPC codes operates on a parity-check factor graph, turbo decoding uses the constituent-code trellises and the interleaver, and CRC-aided successive-cancellation list decoding uses the polar transform and the CRC constraint~\cite{gallager1962low,kschischang2001factor,berrou1993near, arikan2009channel, niu2012crc,tal2015list}. Their message updates, state recursions, or list expansions therefore depend on the corresponding code representation. Consequently, a receiver that supports several code families may require distinct decoding algorithms or code-specific implementations.

\vspace{4pt}

\paragraph{Ordered-statistics decoding}
Ordered-statistics decoding (OSD) is a reliability-based soft-decision decoding method for binary linear block codes~\cite{fossorier1995soft, yue2025guesswork}. OSD ranks the received coordinates by decreasing reliability magnitude, applies the induced permutation to the columns of a generator matrix, and performs Gaussian elimination, with additional column exchanges when required, to obtain a systematic generator matrix whose information positions form a most reliable basis. The standard unpruned order-$m$ reprocessing stage applies test error patterns of Hamming weight at most $m$ to the hard decisions on the $k$ basis positions and re-encodes the resulting information vectors to produce candidate codewords. For $0\leq m\leq k$, this reprocessing set contains $\sum_{j=0}^{m}\binom{k}{j}$ test error patterns. Increasing $m$ enlarges the reprocessing set and can reduce the performance gap to ML decoding, but increases the number of pattern evaluations and re-encodings. Standard OSD also requires a binary generator-matrix representation and a reliability-dependent permutation and basis transformation for each received block. OSD is therefore applicable across binary linear codes, but its candidate-generation stage depends directly on the generator matrix; unlike GRAND, it does not access the codebook solely through a membership test.

\vspace{4pt}

\paragraph{GRAND}
GRAND formulates ML decoding as noise-effect enumeration rather than direct codeword search~\cite{duffy2019capacity}. For a binary additive channel, possibly with memory, let $y$ denote the hard-decision received word and let $z$ denote a candidate binary noise effect. GRAND tests whether the induced word $y\xor z$ belongs to the codebook. For equiprobable codewords, if the noise effects are queried in nonincreasing order of their joint block probabilities under the channel model, with ties resolved by a fixed deterministic rule, then, without abandonment, the first query satisfying $y\xor z\in\calC$ yields an ML codeword~\cite{duffy2019capacity}. GRAND with abandonment (GRANDAB) terminates after a prescribed number of membership queries. Under the random-codebook asymptotics of~\cite{duffy2019capacity}, GRANDAB remains capacity-achieving at rates below capacity when the query limit is chosen with an exponent strictly exceeding the entropy rate of the noise process. GRAND thereby separates the channel-dependent ordering of noise effects from the code-dependent membership test.

\vspace{4pt}

\paragraph{Soft-input GRAND and ORBGRAND}
For additive memoryless channels, SGRAND uses real-valued channel observations to construct, separately for each received vector, an exact likelihood ordering of binary noise effects~\cite{solomon2020soft}. Without abandonment, querying candidates in this order yields ML decoding and provides an exact soft-input reference for the memoryless setting. The original construction maintains a dynamic candidate set, which can be implemented using a max-heap, and updates the set after each query. Because the candidate priorities depend on the received vector, this construction does not provide the predetermined rank-domain query schedule used for direct parallel pattern generation in ORBGRAND.

ORBGRAND instead sorts the coordinates in nondecreasing order of reliability magnitude, with the least reliable coordinate assigned rank one, and combines the resulting permutation with an integer-valued model of the ordered reliability curve~\cite{duffy2022ordered}. Let $\tilde z_j$ indicate whether the coordinate with reliability rank $j$ is flipped. Basic ORBGRAND models the ordered reliability curve by a zero-intercept line, yielding the logistic weight $w_{\mathrm L}(\tilde z) =\sum_{i=1}^{n} i \, \tilde z_i$. Candidate noise effects are generated in nondecreasing order of $w_{\mathrm L}(\tilde z)$, with a fixed deterministic rule for equal weights. Full ORBGRAND replaces the single-slope model with a quantized piecewise-linear approximation to the ordered reliability curve while retaining integer-partition-based pattern generation~\cite{duffy2022ordered}. Integrated implementations have demonstrated a multi-code, multi-rate hard-decision GRAND decoder and a universal soft-decision ORBGRAND decoder~\cite{riaz2021multi,riaz2024sub}. A finite-blocklength analysis of ORBGRAND derives an ORBGRAND-specific random-coding union-type achievability bound on the ensemble-average error probability and a second-order achievable-rate expansion with a corresponding normal approximation for general binary-input memoryless channels~\cite{li2026finite}.

For a fixed reliability permutation, the ORBGRAND metric is additive across the ranked coordinates and contains no pairwise terms. It therefore contains no terms that represent pairwise interactions in a matched finite-memory likelihood metric. Moreover, coordinates that interact locally under the channel model need not have adjacent reliability ranks. Although the stored permutation maps each rank-domain pattern back to the corresponding physical coordinates, the rank-additive metric itself does not encode adjacency-dependent interactions.

\paragraph{Correlation-aware GRAND and exact path enumeration}
Several GRAND variants incorporate channel memory directly rather than use interleaving to separate correlated symbols in the decoder coordinate order~\cite{an2022keep, duffy2023using, duffy2025decoding}. ORBGRAND-AI partitions the received sequence into nonoverlapping blocks, evaluates channel-model probabilities for alternative substitutions within each block, and approximates the posterior probability of a complete noise-effect sequence by a product of blockwise terms. It therefore accounts for within-block dependence but omits cross-block dependence in the approximating metric. Universal noise-guessing decoders have also been developed for unknown discrete additive channels generated by finite-state unifilar processes. These decoders achieve random-coding universality relative to ML decoding and admit upper bounds on their average guessing complexity~\cite{miyamoto2025universal_journal}. That setting differs from the known continuous Gaussian precision model considered here. For memoryless soft inputs, an error-pattern-tree representation has been used to parallelize SGRAND while preserving ML optimality~\cite{wan2025parallelism}. That construction addresses parallel enumeration for a memoryless likelihood metric and does not represent the pairwise terms induced by a non-diagonal Gaussian precision matrix.

Li and Zhang introduced SGRAND-ISI for a specified finite-block linear Gaussian ISI channel, together with lower-complexity ORB-type approximations based on sequence reliability~\cite{li2026grand_isi}. Under complete, untruncated enumeration and without abandonment, SGRAND-ISI is equivalent to ML decoding for its specified finite-block ISI metric. The relation between that model and the present one is stated most precisely through their finite-block negative log-likelihoods. Consider $r=x+n$, $n\sim\mathcal N(0,Q^{-1})$, $Q\succ0$, and let \(A\in\mathbb R^{n\times n}\) satisfy \(Q=A^{\mathsf T}A\). The candidate-dependent part of the negative log-likelihood is $\frac{1}{2}(r-x)^{\mathsf T}Q(r-x) = \frac{1}{2}\lVert Ar-Ax\rVert_2^2$. Thus, defining $\bar r=Ar$ and $v=An$ gives the equivalent transformed model $\bar r=Ax+v$, $v\sim\mathcal N(0,I_n)$. If \(Q\) is banded, its upper-triangular Cholesky factor can be chosen with the same half-bandwidth. The resulting whitening transform is banded and triangular but need not be Toeplitz. A banded precision matrix therefore admits a finite-memory triangular whitening transform, but it need not correspond to a time-invariant finite-impulse-response convolution under a prescribed finite-block boundary convention.

LP-GRAND and SGRAND-ISI are finite-block metric-equivalent when their candidate-dependent negative log-likelihoods differ only by a positive scale and a candidate-independent term. Under this condition, they necessarily induce the same candidate-word order. Agreement of an interior impulse response or an asymptotic spectrum alone does not establish finite-block metric equivalence because boundary terms can alter candidate energies. Accordingly, the stationary AR(1) colored-noise experiment used for the principal BLER results is not treated as a reproduction of the finite-block ISI model in~\cite{li2026grand_isi}, and no novelty is claimed for exact ML GRAND when the two finite-block metrics coincide. The present formulation constructs the binary energy directly from a specified sparse precision matrix, supports exact enumeration using a validated low-width decomposition, uses the same graphical representation for min-sum and sum-product recursions, and derives order-stability bounds for coefficient quantization and precision-matrix mismatch.

Ranked path enumeration on trellises is classical. List-Viterbi algorithms and general $K$-shortest-path methods provide several constructions for ordered path lists~\cite{seshadri1994list,eppstein1998finding}. The numerical baseline used here retains up to the $K$ lowest-cost partial paths ending at each trellis state after every stage. We therefore refer to it as a \emph{batch state-list $K$-best Viterbi} dynamic program rather than as the classical serial list-Viterbi algorithm. LP-GRAND uses a suffix-guided priority queue to enumerate complete paths on the precision-induced layered graph, with priorities formed using suffix costs computed by dynamic programming. The method-specific component is the construction of the Gaussian binary-energy graph and its use with a codebook-membership test, not generic $K$-best path enumeration. Section~\ref{sec:numerical_evaluation} also includes an internally implemented first-order burst reference for an explicitly specified square two-tap finite-block ISI matrix. This reference is not an implementation of the published SGRAND-ISI algorithm.

Recent ORB-type analysis formalizes a class of GRAND algorithms whose test order depends only on the reliability ranking and derives exact finite-budget expressions for BLER, the stopping-time distribution, and the average number of tests for random-code ensembles~\cite{wan2026orbtype}. Within the error-pattern set under consideration, testing patterns in nonincreasing order of average guessing posterior minimizes both the ensemble-average BLER and the average number of tests. LP-GRAND lies outside this rank-only class because its query order depends on the received magnitudes and the pairwise couplings induced by the precision matrix. Consequently, results obtained using different finite-block channel matrices, codes, or stopping rules do not constitute a controlled pointwise comparison; such a comparison requires these elements to coincide.

\paragraph{Soft-output GRAND}
SOGRAND augments a soft-input GRAND search with posterior estimates of the correctness of the returned codewords and of the probability that the transmitted codeword has not yet been identified~\cite{yuan2025soft, duffy2024soft}. In list decoding, it estimates the posterior probability associated with each returned codeword and the probability that the correct codeword is absent from the returned list. These blockwise quantities can also be converted into bitwise soft outputs. The single-output first-hit expression in Section~\ref{sec:softout} addresses this framework under its random-codebook occupancy model. LP-GRAND's sum-product recursion supplies the full-space partition function required to normalize the noise-effect weights and compute the residual ambient probability mass outside the enumerated set. The occupancy model is then applied to estimate codeword-level posterior probabilities from these ambient probabilities. No new general SOGRAND identity is introduced here. In this paper, we focus on exact likelihood-ordered enumeration for the pairwise Gaussian binary energy induced by a sparse precision matrix.

\vspace{15pt}

\section{Exploratory BLER-Target Interpolation}
\label{app:exploratory_crossing}

Table~\ref{tab:snr_targets_appendix} reports descriptive estimates of the nominal SNR at which the empirical BLER reaches $p_\star=10^{-3}$. Let $\gamma_a<\gamma_b$ be two adjacent simulated values of $(E_b/N_0)_{\mathrm{dB}}$, and suppose that their empirical BLERs satisfy
\begin{equation}
0 < \widehat{P}_{\mathrm e,b} < p_\star < \widehat{P}_{\mathrm e,a}.
\end{equation}
Linear interpolation of $\log_{10}\widehat{P}_{\mathrm e}$ between these two SNR values gives
\begin{equation}
\widehat{\gamma}_\star \triangleq \gamma_a
+ \frac{\log_{10}p_\star-\log_{10}\widehat{P}_{\mathrm e,a}}{\log_{10}\widehat{P}_{\mathrm e,b} - \log_{10}\widehat{P}_{\mathrm e,a}}
\left(\gamma_b-\gamma_a\right).
\label{eq:exploratory_snr_crossing}
\end{equation}

Only $3$, $5$, and $9$ block errors were observed at $4$~dB for LP-GRAND, ExactBlockProduct with $b=8$, and our ORBGRAND-AI implementation with $b=8$, respectively. The interpolated values are therefore sensitive to the empirical BLERs at this sparse-error endpoint and should not be interpreted as precisely estimated SNR gains at BLER $10^{-3}$.

\begin{table}[h]
\centering
\caption{Descriptive log-linear interpolation at $\mathrm{BLER}=10^{-3}$. The final column gives the difference between each interpolated crossing and that of LP-GRAND.}
\label{tab:snr_targets_appendix}
\scriptsize
\setlength{\tabcolsep}{3.5pt}
\resizebox{\columnwidth}{!}{%
\begin{tabular}{lcc}
\toprule
Method
& \shortstack{Interpolated nominal\\
$(E_b/N_0)_{\mathrm{dB}}$}
& \shortstack{Difference from\\LP-GRAND (dB)}
\\
\midrule
LP-GRAND
& $3.42$
& --
\\
ExactBlockProduct, $b=8$
& $3.69$
& $0.27$
\\
ORBGRAND-AI, $b=8$
& $3.95$
& $0.53$
\\
\bottomrule
\end{tabular}%
}
% \vspace{30pt}
\end{table}

\vspace{25pt}

\section{First-Order Gauss--Markov Precision Matrix}
\label{app:gm_precision}

Let $N=(N_1,\ldots,N_n)^{\mathsf T} \sim\mathcal{N}(\mathbf{0}_n,\Sigma)$, where $\Sigma_{ij}=\sigma^2\rho^{|i-j|}$, $\sigma^2>0$, and $|\rho|<1$. For $n\geq2$, this distribution is equivalently generated by $N_1\sim\mathcal{N}(0,\sigma^2)$, $N_i=\rho N_{i-1}+V_i$, $i=2,\ldots,n$, where $V_2,\ldots,V_n \stackrel{\mathrm{i.i.d.}}{\sim} \mathcal{N}\!\left(0,\sigma^2(1-\rho^2)\right)$ and the innovations are independent of $N_1$. For \(u=(u_1,\ldots,u_n)^{\mathsf T}\), the negative log-density can be written as
\begin{equation}
-\log f_N(u) = C + \frac{u_1^2}{2\sigma^2} +
\sum_{i=2}^{n} \frac{(u_i-\rho u_{i-1})^2}{2\sigma^2(1-\rho^2)},
\label{eq:gm_negative_log_density}
\end{equation}
where $C$ is independent of $u$. Expanding the quadratic terms in \eqref{eq:gm_negative_log_density} and collecting the coefficients in $-\log f_N(u) = C+\frac{1}{2}u^{\mathsf T}Qu$ gives
\begin{align}
Q_{11}=Q_{nn} &= \frac{1}{\sigma^2(1-\rho^2)}, \\
Q_{ii} &= \frac{1+\rho^2}{\sigma^2(1-\rho^2)},
\qquad i=2,\ldots,n-1, \\
Q_{i,i+1}=Q_{i+1,i} &= -\frac{\rho}{\sigma^2(1-\rho^2)},
\qquad i=1,\ldots,n-1.
\end{align}
The second expression is present only for \(n\geq3\), and all remaining entries are zero. Therefore, \(Q=\Sigma^{-1}\) is the tridiagonal precision matrix given in~\eqref{eq:gm_precision}.

\vspace{30pt}

\section{Additional Query-Budget and Quantization Results}
\label{app:extra_qmax_quant}

The top panel of Fig.~\ref{fig:appendix_qmax_quant} shows the codeword-return rate as a function of the valid membership-test budget. For a given $q_{\max}$, the codeword-return rate is the fraction of frames for which the decoder returns a codeword within at most $q_{\max}$ valid membership tests. The values are computed from the recorded first-hit query counts of the verified paired $10^4$-frame experiment at nominal $(E_b/N_0)_{\mathrm{dB}}=2$ and $\rho=0.5$; no additional decoding runs are performed.

\begin{figure}[!t]
\centering
\begin{minipage}{0.96\linewidth}
\centering
\includegraphics[width=\linewidth]
{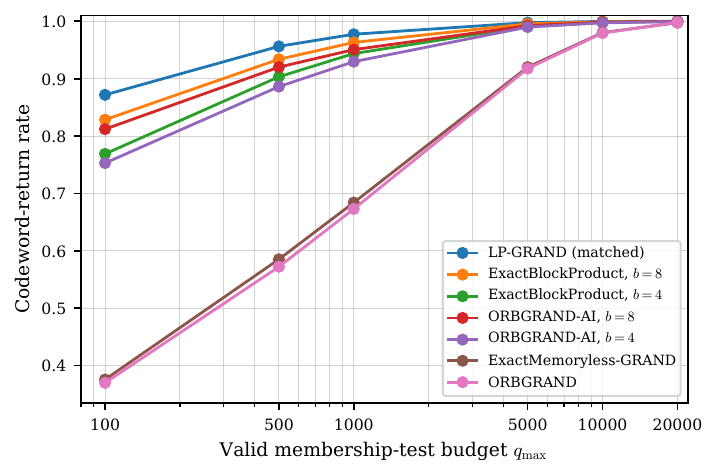}
\end{minipage}

\vspace{1mm}

\begin{minipage}{0.97\linewidth}
\centering
\includegraphics[width=\linewidth]
{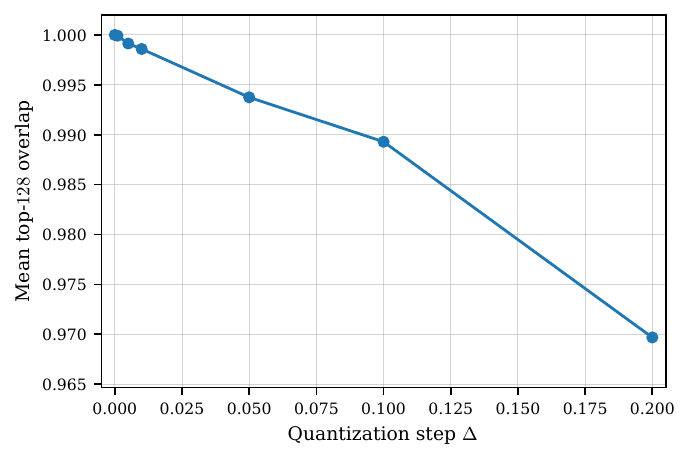}
\end{minipage}
\caption{Additional decoding and ordering results. Top: codeword-return rate as a function of the valid membership-test budget, shown on a
logarithmic horizontal axis and computed from the recorded first-hit query counts of the verified paired $10^4$-frame experiment at nominal
$(E_b/N_0)_{\mathrm{dB}}=2$ and $\rho=0.5$. Bottom: sample mean of the framewise top-$128$ overlap under coefficient quantization over $100$ common frames. The corresponding mean union-restricted inversion rates are reported in Fig.~\ref{fig:robustness}.}
\label{fig:appendix_qmax_quant}
\end{figure}

For LP-GRAND, increasing $q_{\max}$ from $100$ to $5000$ increases the codeword-return rate from $0.8718$ to $0.9977$ and decreases the empirical BLER from $0.1336$ to $0.0281$. The codeword-return rate and BLER are not complementary: an abandoned frame is counted as a block error, but a decoder that returns a codeword can also be in error if the first codebook hit is not the transmitted codeword.

The bottom panel shows the effect of coefficient quantization on the top-$128$ candidate set. For each of the same $100$ frames used in Fig.~\ref{fig:robustness}, we compute $O_{128}(W_r,W_{r,\Delta}^{\mathrm q})$ as defined in \eqref{eq:topk_overlap}. The plotted value is the arithmetic mean of this framewise overlap. Thus, the overlap measures agreement between the two top-$128$ candidate sets, whereas the union-restricted inversion rate in~\eqref{eq:union_inversion_rate} measures strict pairwise order reversals among comparable candidates in their union.

\vspace{20pt}

% ======================================================================
\section{Suffix-Guided Best-First Enumeration}
\label{app:k_shortest_paths}
% ======================================================================
 
This appendix presents the suffix-guided best-first enumerator used for the binary trellis in Section~\ref{sec:graphical}. The procedure performs best-first search on the tree of path prefixes. Each partial path is prioritized by the sum of its accumulated cost and the exact minimum cost of completing it to a full trellis path. Equivalently, the procedure is an A* search that uses the exact cost-to-go. Related best-first and ranked-path methods are well established~\cite{hart1968formal,eppstein1998finding}.

For a trellis state $u$ immediately before layer $i$, let $J_i(u)$ denote the minimum cost of completing layers $i,\ldots,n$. These suffix costs satisfy
\begin{equation}
J_i(u) = \min_{b\in\{0,1\}} \left\{
\gamma_i(u,b) + J_{i+1}\bigl(T(u,b)\bigr) \right\},
\quad i=n,\ldots,1,
\label{eq:suffix_dp}
\end{equation}
with $J_{n+1}(u)=0$ for every state at the terminal layer. Because the trellis is finite and acyclic, this recursion is valid for arbitrary real-valued branch costs.

Let $p_i=(z_1,\ldots,z_i)$ denote a partial path containing the first $i$ branch labels, and let $p_0$ denote the empty path. Its accumulated cost is
\begin{equation}
g(p_i) \triangleq \sum_{\ell=1}^{i}
\gamma_\ell(u_\ell,z_\ell),
\qquad g(p_0)\triangleq0,
\label{eq:partial_path_cost}
\end{equation}
where $u_\ell$ is the state immediately before layer $\ell$. If $u_{i+1}$ is the state reached by $p_i$, assign the path the priority
\begin{equation}
f(p_i) \triangleq g(p_i)+J_{i+1}(u_{i+1}),
\qquad i=0,\ldots,n.
\label{eq:astar_key}
\end{equation}
Thus, $f(p_i)$ is the minimum total cost of any complete path extending $p_i$. In particular, $f(p_n)=g(p_n)$ for every complete path.

The priority queue is initialized with $p_0$, whose state is the initial trellis state $u_1$ and whose priority is $J_1(u_1)$. Each queue entry is assigned the key $\bigl(f(p_i),s(p_i)\bigr)$, where $s(p_i)$ is a unique insertion index assigned in increasing order. Queue entries are removed in lexicographic order of these keys.
When a removed entry represents a complete path, its branch-label sequence is emitted. Otherwise, let the removed path be $p_i=(z_1,\ldots,z_i)$ with $i<n$. For each $b\in\{0,1\}$, define the child path $p_{i+1}\triangleq(z_1,\ldots,z_i,b)$, with
\begin{align}
u_{i+2} &=T(u_{i+1},b),
\label{eq:child_state}
\\
g(p_{i+1}) &=g(p_i)+\gamma_{i+1}(u_{i+1},b).
\label{eq:child_cost}
\end{align}
Each child is inserted with primary priority $f(p_{i+1})$ as defined in~\eqref{eq:astar_key}. Branch labels are processed in a fixed order, and the corresponding insertion indices are assigned in that order. These indices provide a deterministic ordering among queue entries having equal primary priority; they do not, in general, induce lexicographic ordering of the complete branch-label sequences. This insertion-index rule is the deterministic tie-breaking policy used throughout the paper.

For every prefix $p_i$, the value $f(p_i)$ equals the minimum cost of any complete path extending that prefix. Therefore, when a complete path is removed from the queue, no unemitted complete path has smaller cost. The emitted complete paths are consequently ordered by nondecreasing total cost, even when some branch costs are negative. Each nonempty prefix has a unique parent and is inserted only when that parent is expanded; hence, no complete path is emitted more than once. If the search continues until the queue is empty, every complete path is emitted exactly once. By Theorem~\ref{thm:trellis}, the corresponding branch-label sequences are distinct noise-effect patterns ordered by nondecreasing $W_r$. These statements apply in exact arithmetic; the ordering obtained using the stored binary64 coefficients is validated separately in Section~\ref{ssec:complete_order_validation}.

For a precision matrix of half-bandwidth at most $\nu$, a full-state suffix implementation that allocates all $2^\nu$ states at every layer stores $(n+1)2^\nu$ values $J_i(u)$, including those at the terminal layer. It evaluates both branch labels for each of the $n2^\nu$ nonterminal state--layer pairs and therefore performs exactly $n2^{\nu+1}$ branch evaluations. Its storage requirement is $\mathcal{O}(n2^\nu)$.

For online enumeration, let $B_K$ denote the number of priority-queue removals up to and including the removal that emits the $K$th complete path, and let $H_K$ denote the maximum queue occupancy during these iterations. Exactly $K$ removals emit complete paths. The remaining $B_K-K$ removals expand nonterminal prefixes and produce $2(B_K-K)$ child insertions, in addition to the initial insertion of $p_0$.

Assume that each prefix is represented by its terminal state, depth, accumulated cost, and a parent pointer, so that branch extension and priority evaluation require $\mathcal{O}(1)$ time. With a binary heap, the queue-processing time through the $K$th emitted path is $\mathcal{O}\!\left(B_K\log(1+H_K)\right)$. This bound excludes reconstruction of the emitted branch-label sequences. Explicitly materializing $K$ sequences of length $n$ requires an additional $\mathcal{O}(Kn)$ time. The queue stores at most $H_K$ entries, while retaining one predecessor record for every generated prefix requires $\mathcal{O}(B_K)$ additional storage through the $K$th emission. The quantities $B_K$ and $H_K$ depend on the requested number of outputs and the path costs; when $K$ is allowed to grow to $2^n$, both may be exponential in $n$.

\renewcommand{\refname}{Appendix References}
\putbib
\end{bibunit}

\end{document}